\documentclass[11pt]{article}

\usepackage{amssymb}
\usepackage{amsmath}
\usepackage{amsthm}
\usepackage{graphicx}
\usepackage{amssymb}
\usepackage{amsfonts}
\usepackage{amsbsy}
\usepackage{verbatim}
\usepackage{stmaryrd}
\usepackage{fancyhdr}

\usepackage{enumerate}
\usepackage{hyperref}
 
\usepackage{graphicx}
\usepackage{comment}
\usepackage{color}
\usepackage{comment}
\includecomment{versiona}
\excludecomment{versionb}

\label{defis_LtX}
\newtheorem{prop}{Proposition}[section]
\newtheorem{theorem}{Theorem}[section]
\newtheorem{lemma}{Lemma}[section]

\newtheorem{remark}{Remark}[section]

\newtheorem{assumption}{Assumption}

\newcommand{\be}{\begin{equation}}
\newcommand{\ee}{\end{equation}}
\newcommand{\ba}{\begin{align}}
\newcommand{\ea}{\end{align}}
\newcommand{\PP}{\mathcal{P}}

\newcommand{\cl}{\mathcal{L}}
\newcommand{\La}{\Lambda}

\newcommand{\N}{\mathbb{N}}

\newcommand{\R}{\mathbb{R}}

\newcommand{\skii}{\sum_{k_1=1}^{N}}

\newcommand{\dd}{,\dots,}

\newcommand{\kk}{\mathbf{k}}
\newcommand{\e}{\mathbf{e}}
\newcommand{\q}{\mathbf{q}}
\newcommand{\bfd}{\mathbf{d}}

\newcommand{\x}{\mathbf{x}}
\newcommand{\z}{\mathbf{z}}

\newcommand{\Z}{\mathbf{Z}}
\newcommand{\Znax}{\Z_{\boldsymbol\alpha,x}^{(n)}}

\newcommand{\bsa}{\boldsymbol\alpha}

\newcommand{\ZZ}{\mathbb{Z}}

\newcommand{\s}{\mathfrak{S}}

\newcommand{\lv}{\left\vert}
\newcommand{\rv}{\right\vert}
\newcommand{\ml}{\mathfrak{L}}
\newcommand{\pa}{\partial}

\begin{document}

\title{Markov semigroups with hypocoercive-type generator in infinite dimensions: ergodicity and smoothing}
\author{ V. Kontis\footnote{v.kontis@imperial.ac.uk,  
Department of Epidemiology and Biostatistics,
Imperial College London, UK },
 M. Ottobre\footnote{michelaottobre@gmail.com, Department of Mathematics,
 Heriot Watt University, Edinburgh, UK } and B. Zegarlinski\footnote{b.zegarlinski@imperial.ac.uk, Department of Mathematics, Imperial College London, London, UK} \footnote{Supported by Royal Society Wolfson RMA} }

\date{\today}
\maketitle

\begin{abstract}
\noindent We start by considering finite dimensional Markovian dynamics in $\R^m$ generated by operators of hypocoercive type and for such models we obtain short and long time pointwise estimates for all the derivatives, of any order and in any direction, along the semigroup. We then look at infinite dimensional models (in  $(\R^m)^{\ZZ^d}$) produced by the  interaction of infinitely many finite dimensional dissipative dynamics of the type indicated above.  For these infinite dimensional models we study finite speed of propagation of information, well-posedness of the semigroup, time behaviour of the derivatives and strong ergodicity problem. 
       \end{abstract}

\section{Introduction} \label{S.1:Introduction}

In this paper we consider infinite dimensional models of interacting dissipative systems with noncompact state space. In particular we develop a basis for the construction and analysis of dissipative semigroups whose generators are given in terms of noncommuting vector fields and for which the equilibrium measures are not a priori known. The ergodicity theory in the case where an invariant measure is not given in advance, in the noncompact subelliptic setup, is an  interesting and challenging problem which was initially studied in \cite{DKZ2011}, and which we extend in new directions in this paper, developing a strategy based on generalised gradient bounds. In the following we  will first present the main results of the paper and we will then  relate them to existing results in the literature. 

Hypoelliptic operators of hypocoercive type have received a lot of attention in recent years (see e.g. \cite{Herau2007,V} and references therein), as they naturally arise in non-equilibrium statistical mechanics, for example in the context of the heat bath formalism.  These are second order operators on $\R^m$ in  H\"ormander form 
$$L= Z_{0}^2+B,$$
where  $Z_{0}$ and $B$ are first order differential operators. The principal part is spanned by at least one field $Z_{0}$ which, together with the term of first order $B,$ generate fields $Z_{j+1} \equiv  [B,Z_j]$, $j=0,..,N-1$ spanning the full Lie algebra. Therefore by H\"ormander theorem, (see e.g. \cite{H1}, \cite{BLU}, \cite{CSCV} and references therein), such semigroups have the strong smoothing property. Motivated by  \cite{V},  we will refer to these generators as {\em hypocoercive-type} operators (see Remark \ref{remhypocoer}).
 
At the beginning of the paper,  in Section 2,  we describe a systematic inductive method which allows to obtain quantitative short and long time estimates  for the space-derivatives of the semigroup generated by $L$.  We obtain pointwise bounds on the derivatives  of any order and in any space direction.  The techniques of Section 2 were  originally developed in \cite{MO_thesis} and are based on combining the hypocoercive method presented in \cite{Herau2007, V} with the classic Bakry-Emery semigroup approach (\cite{BE}).  Section \ref{Sec:subs1to2} contains an explanation of our  technique and  its relations with the aforementioned methods in a simplified setting,  so that the involved notation of   Section \ref{Sec:subs2to2}, which is devoted to proving the time behaviour of the derivatives in full generality, does not obfuscate the idea behind the method we present.  While obtaining such estimates is an interesting problem in itself,  further motivation for obtaining {\em pointwise}  estimates comes from the fact that in the infinite dimensional situation we are interested in this paper, typically one does not have any reference measure.  As a consequence, since we do not have the integration by parts trick at our disposal, generally we need to sacrifice estimates  in direction of  $B$. To the best of our knowledge, a purely analytical method adapted to obtaining pointwise bounds on the time-behaviour of the derivatives of any order of degenerate Markov semigroups was so far lacking.  We now come to present the infinite dimensional problem tackled in the subsequent sections of the paper.

Once we have studied the finite dimensional diffusion in $\R^m$ generated by the operator $L$, we study systems of infinitely many interacting diffusions of hypocoercive type. This is done by considering  the lattice $\mathbb{Z}^d$ and, roughly speaking, ``placing"  an isomoprhic copy of our $\R^m$-diffusion  at each point of such a  lattice.  Finally, we let these dynamics interact, obtaining in this way an infinite dimensional Markovian dynamics in $(\R^m)^{\mathbb{Z}^d}$. 

In Section 3 we provide a general construction of Markov semigroups in an infinite dimensional setup with the underlying space given  as a subset of an infinite product space (including an infinite  product of noncompact Lie algebras). We improve on the results described in \cite{DKZ2011} for semigroups with all generating fields present in the principal part of the generator. Firstly, we relax the conditions on the structure of the Lie groups and the principal part. Secondly, we obtain a generalisation of the allowed interaction including a  second order perturbation part dependent on fields acting on different coordinates. In addition, we prove stronger finite speed of propagation of information estimates providing a tree bound decay when derivatives with respect to many different coordinates act on the semigroup. This in particular allows us to prove smoothness of the semigroups in our general infinite dimensional setup, thus filling an important gap in the literature  (one may also expect that our estimates will provide some additional information about the equilibrium measure).   
%
%
 Additionally this allows to provide some new criteria for ergodicity of the semigroups.

In Section 4 we provide a strategy for proving the existence of invariant measures for a semigroup. Assuming a Lyapunov-type condition for the generator of a finite dimensional semigroup acting on a suitable unbounded function with compact level sets, we formulate conditions on the interaction allowing to  apply a weak compactness criterion for the generator of the infinite dimensional non product semigroup constructed in Section 3.

In Section 5 we provide a  criterion for uniqueness of the invariant measure using first order as well as higher order estimates. 

 In general, it is an art how to apply our criterion to particular models. We provide a number of  explicit examples in \cite{MVII}.
In this  companion paper, we provide a concrete illustration of the full flexibility of  the theory developed in this paper with a variety of application areas and including examples where the methods we provide here are applied in a non-standard way. 
This includes a number of examples of infinite dimensional models with smoothing and ergodicity  estimates,  where precise dependence on parameters can be obtained or where one has long time concentration along some directions only. 
\subsection{Relation with the literature}
As we have already remarked, regarding the finite dimensional framework,  the techniques of Section \ref{S.2:Short and long time behaviour} result from combining the Bakry-Emery semigroup approach (\cite{BE}) with the hypoelliptic/hypocoercive methods proposed by H\'erau and Villani (\cite{Herau2007, V}). As is well known, the semigroup approach leads to pointwise estimates, but it is mainly designed for elliptic dynamics. The more recent framework proposed in \cite{Herau2007, V} is devised for degenerate diffusions but it requires an a priori knowledge of the invariant measure $\mu$ of the semigroup and it indeed produces estimates in the weighted space $L^2(\mu)$.  Both of these techniques are entirely analytical. Combining these approaches results in a method that enjoys the perks of both: it is suited to the degenerate setup and it produces pointwise estimates. Moreover, it doesn't require any knowledge about the invariant measure and it can be adapted to tackle infinite dimensional problems, as we show in this paper.     As far as the finite dimensional setting is concerned, another viable approach to study the time behaviour of the space derivatives of the semigroup is the probabilistic one, via Malliavin calculus
(see for example \cite{Crisan} and references therein). However it might be technically involved to extend this technique to the infinite dimensional framework that we are aiming for. In contrast, the method we propose is easy to extend to the infinite dimensional setup. 

We now come to explaining how our results for the infinite dimensional dynamics relate to existing ones in the literature.  The problem of construction and ergodicity of dissipative dynamics for infinite dimensional interacting particle systems with bounded state space has a long history, see e.g. \cite{L}, \cite{GZ} and references therein. For a construction of Markov semigroups on the space of continuous functions acting on  an infinite dimensional underlying space (well suited to study strong ergodicity problems), we refer to \cite{Z1} in case of fully elliptic operators, to \cite{DKZ2011}  for the  subelliptic setup, and to \cite{LuZ} for  L\'evy type generators; these constructions will be even more extended in this paper.

 An interesting approach via stochastic differential equations can be  found in  \cite{DaPZ} (see also  \cite{BT}, \cite{BHT} and references therein). We mention also another approach via  Dirichlet forms theory (see e.g. \cite{AKR}, \cite{R} and reference therein), which is well adapted to $L_2$ theory. 

For symmetric semigroups, recent progress has been made in proving the log-Sobolev inequality
for infinite dimensional H\"ormander-type generators $\mathcal{L}$ which are symmetric in the weighted space 
$L_2(\mu)$,  defined with respect to  a suitable nonproduct measure $\mu$ 
(\cite{LZ},
\cite{HZ}, \cite{IP}, \cite{I}, \cite{IKZ}). One can therefore expect an extension of the established
strategy (\cite{Z1}) for proving strong pointwise
ergodicity for the corresponding Markov semigroups $P_t\equiv e^{t\mathcal{L}}$
(or, in case of the compact spaces, even in the uniform norm as e.g. in \cite{GZ}).
To obtain a fully fledged theory in this direction which could include for example configuration spaces given by infinite products of general noncompact nilpotent Lie groups other than
Heisenberg-type groups, one needs to conquer a (finite dimensional) problem of
sub-Laplacian bounds (of the corresponding control distance) 
which for the moment remains still very hard.  We remark that in the fully elliptic case a strategy based on classical Bakry-Emery arguments and involving a restricted class of interactions can be achieved (even for nonlocal generators; see e.g. \cite{LuZ}). In the case of the stochastic strategy of \cite{DaPZ}, the convexity assumption enters via a dissipativity condition in a suitable Hilbert space and does not improve the former one as far as ergodicity is concerned (although on the other hand it allows to study a number of stochastically natural models). In the subelliptic setup, involving subgradients, this strategy faces serious obstacles, see e.g. comments in \cite{BBBC}.

 To summarize, the purpose of this paper is twofold: regarding the finite dimensional setup, we improve on the methods presented in \cite{V} by adapting the hypocoercive techniques to problems in which an invariant measure might not be a priori known; in infinite dimensions, we provide results about systems of infinitely many interacting diffusions, thereby completing and extending the framework of \cite{ DKZ2011, LuZ, Z1}. 

\section{Short and long time behaviour of $n^{th}$ order derivatives in finite dimensions} \label{S.2:Short and long time behaviour}
Consider a second order differential operator on $\R^m$, of the form
\begin{equation}\label{operator}
L=Z_0^2+B, 
\end{equation}
where $Z_0$ and $B$ are first order differential operators. We will assume   the following commutator  structure:

\begin{assumption}[\textbf{CR.I}] Assume that for some  $N \in \mathbb{N}$, $N\geq 1$, there exist $N$  differential operators $Z_1,\dots,Z_N$, such that
the following commutator relations hold true:
\begin{align} 
\label{comm} 
 [B,Z_j]&=Z_{j+1},   \qquad \qquad \; \mbox{for all }\,\,  j=0,\dots,N-1, \\
[B,Z_N]&=\sum_{j=0}^{N}c_jZ_j,  \nonumber \\
[Z_i, Z_j]&=\sum_{h=0}^N c_{ijh} Z_h, 
\qquad \mbox{for all }\,\, 0\leq i,j\leq N, \nonumber
\end{align}
for some constants $c_j \in \mathbb{R},\, j=0,..,N-1$, $c_N \in [0,\infty)$  and $c_{ijh} \in \mathbb{R}$, with $c_{0jh}\equiv 0 \,\, for \,\, h\geq j-1$.
\end{assumption}
The main result of this section is Theorem \ref{thm.1}, regarding the time behaviour of fields of any order along the semigroup. In order to state such a theorem we first need to introduce some notation and to further detail our framework. 
 
We will assume that the collection of differential operators $B$ and $Z_0, Z_1\dd Z_N$ span $\R^m$ at each point.\footnote{Strictly speaking, this assumption is not needed in the finite dimensional case. However it will be needed for the infinite dimensional problem. In $\R^m$, it is simply the case that one will obtain estimates in all the directions that can be obtained from the successive commutators between $Z_0$ and $B$, including $Z_0$ but not $B$.}
For $Z_{k_j}$, $k_l\in\{0,..,N\}$, $l=1,..,n$, and $n\in\mathbb{N}$, we set
\[ \Z_{\kk,n}\equiv  \Z_{k_1,..,k_n} := Z_{k_1}\cdot{\dots}\cdot Z_{k_n}.
\]
In this section we will be referring to  terms of the form 
$ \Z_{\kk,n}f $ and $ Z_0\Z_{\kk,n} f  \equiv \left( Z_0Z_{k_1}\cdot{\dots}\cdot Z_{k_{n}}f \right)$
as terms of length $n$ and   terms of length $n+1$ starting with $Z_0$, respectively. 
We will use $\mathbf{e}_l$, $l=1,..,n$, with $(\mathbf{e}_l)_m=\delta_{lm}$, for the standard basis in $\mathbb{R}^n$, and
$\kk\equiv \sum_{l=1,..,n}k_l \mathbf{e}_l$ with non-negative integer coefficients $k_l $, ${l=1,..,n}$,
and we set
\[|\kk|_n:= \sum_{l=1}^n k_l.\]

In the following $\parallel\cdot\parallel_{\infty}$ indicates the supremum norm. We will use the notation $P_t\equiv e^{tL}$, $t\geq 0$,  for the semigroup generated by the operator $L$ and set $f_t\equiv e^{tL }f$  for any continuous bounded function $f$. 


For some strictly positive constants $a_{\kk,n}\equiv a_{k_1,..,k_n}, b_{\kk,n}\equiv b_{k_1,..,k_n}$, 
$0\leq k_l\leq N$, $l=1,..,n$, and $d$ (to be chosen later),  
we define the following time dependent quadratic forms
\[ \label{gamma_0}
Q^{(0)}_t f_t=d\vert f_t\vert^2 
\]
and  
\begin{align}
\bar \Gamma^{(1)}_t f_t&\equiv \sum_{j=0}^N a_j t^{2j+1}
\vert Z_j f_t\vert^2, \nonumber \\
\label{gamma_1} \Gamma^{(1)}_t f_t&\equiv \bar \Gamma^{(1)}_t f_t +
\sum_{j=1}^{N}b_jt^{2j}(Z_{j-1}f_t)(Z_{j}f_t),\\
Q^{(1)}_t f_t &\equiv  \Gamma^{(1)}_t f_t + Q^{(0)}_t f_t. \label{13B}
\end{align} 
For general $n\geq 2$ we set, 
\begin{align}
\bar \Gamma^{(n)}_t f_t  
&\equiv
 \sum_{  |\kk|_{n}=0}^{nN} 
a_{\kk,n}
t^{2|\kk|_{n} +n}
\vert \Z_{\kk,n} f_t\vert^2, \nonumber\\ 
\Gamma^{(n)}_t f_t  
&\equiv
\bar \Gamma^{(n)}_t f_t  
+
 \sum_{0 \leq |\kk|_{n}:\, k_1\geq 1} ^{nN}
b_{\kk,n} t^{2|\kk|_{n} +n-1}
(\Z_{\kk-\mathbf{e}_1,n} f_t)
(\Z_{\kk,n} f_t) , \nonumber
\\
 Q^{(n)}_t f_t  
&\equiv \Gamma^{(n)}_t f_t+ 
  Q^{(n-1)}_t f_t.  \label{7star}
\end{align} 

\begin{theorem}\label{thm.1} 
Suppose the operators $B,Z_j$, $j=1,..,N$ satisfy Assumption \textup{(\textbf{CR.I})} and $P_t\equiv e^{tL}$ is a Markov semigroup with generator $L$ given by 
\eqref{operator}. Then for all
$n\in\mathbb{N}$ and for all $0 \leq l\leq n$ there exist strictly positive constants $a_{\kk,l} , b_{\kk,l}, \bar d_l,d_l$ and $T\in (0,\infty]$ such that 
\be \label{gammabound} 
\bar d_l\bar \Gamma_t^{(l)} f_t \leq \Gamma_t^{(l)} f_t  \leq d_l \left(P_t f^2 - (P_t f)^2\right),   \qquad\mbox{for all }\,\, 1\leq l\leq n 
\mbox{ and } 0<t<T. 
\ee
Moreover if $c_j\equiv 0$ and $c_{0jh}\equiv 0$ for all $j=1,..,N$, then $T=\infty$, and in particular we have
\be\label{gammabound1}
\|\Z_{\kk,n} f_t\|_\infty^2\leq \frac{C}{t^{2 \lv \kk\rv_n+n}}\| P_tf^2-(P_tf)^2\|_{\infty}\leq  \frac{C}{t^{2 \lv \kk\rv_n+n}}\inf_{c\in\mathbb{R}}\|  f-c\|_{\infty}^2, \qquad\mbox{for all }\,\, t>0,
\ee
with some constant $C\in(0,\infty)$ independent of $t$ and $f$.
\end{theorem}
Before coming to the proof of Theorem \ref{thm.1}, we make a couple of remarks in order to give some more intuition about the statement of such a theorem.
\begin{remark}\label{rem:flex}
\textup{In words, Theorem \ref{thm.1} states the following: under the general commutator relations of Assumption \textup{(\textbf{CR.I})}, the time behaviour of the first inequality in \eqref{gammabound1} is valid only for $0<t<T$ with $T$ small enough, typically $T<1$, i.e. in full generality we can only obtain a short time estimate.  However, if we assume for example that the fields $Z_i, i=0 \dd N$ commute and that $Z_N$ commutes with $B$, then  the time behaviour \eqref{gammabound} is valid for any $t>0$. In this paper we work under the relatively general Assumption \textup{(\textbf{CR.I})}. We would like to emphasize that the technique we use to prove  Theorem \ref{thm.1} is quite flexible and might give better results depending on the case at hand. In particular one might be able to improve on the time interval in which the estimate is valid when exact knowledge of the constants appearing in  Assumption \textup{(\textbf{CR.I})} is available. This improvement might  also be obtained in cases where}
$$
 [B,Z_j]= \alpha_j Z_{j+1},   \qquad \mbox{\textup{for all}}\,\,\,  
 j=0,\dots,N-1, 
 $$
 \textup{for some large positive  constants $\alpha_j$. When the generator contains  a dilation operator it is also possible to obtain exponential decay. We have illustrated this fact with an  (infinite dimensional) example in \cite{MVII}, see also Theorem \ref{thm5.1}.} 
\end{remark}
\begin{remark}\label{remhypocoer}
\textup{Notice that the above proposition is coherent with H\"ormander's rank nomenclature, as it agrees  with the heuristics according to which for any differential operators $X$ and $Y$,  $r(XY)=r(X)+r(Y)$ and $r([X,Y])=r(X)+r(Y)$, where $r(X)$ denotes the rank of the operator $X$. In particular, $Z_0$ is an operator of rank 1 and $B$ is an operator of rank 2, so that $r(Z_j)=2j+1$, for any $0\leq j \leq N$ and 
$r(Z_{k_1}\cdot{\dots}\cdot Z_{k_n})=2\sum_{j=1}^n k_j+n$. \\
Because new vector fields are obtained only through commutators with the rank 2 operator $B$, we will refer to $L$ as to  a hypocoercive-type operator, in analogy with the setting considered in \cite{V}.
However we would like to stress that despite this clear analogy, the setting in which we are going to work is quite far from the one of the hypocoercivity theory. Indeed, as we have mentioned in the introduction, here we do not assume the existence of a reference (equilibrium) measure and the estimates we obtain are pointwise. }
\end{remark}
\begin{remark}\textup{Because of the linearity of the operator, all the results of Theorem \ref{thm.1} still hold if
$$
L=\sum_{i=1}^MZ_{0,i}^2+B,
$$
for some $M>1$. We do not present the results in such generality only to avoid having cumbersome notations, especially in the proof of Theorem \ref{thm.1}  and in the infinite dimensional setting. }

\end{remark}
The proof of Theorem \ref{thm.1} is quite lengthy although in principle not complicated. The lengthy calculations that such a proof requires might obscure the simple idea behind it; especially, they might conceal the flexibility of our approach.   In order to clearly  explain the strategy of proof, we gather in  Section \ref{Sec:subs1to2} below  a simple explanation of the principle behind our approach with a sketch of the  proof of Theorem \ref{thm.1} in the simple case $m=2$ and $N=n=1$.  The full proof of Theorem \ref{thm.1} is instead deferred to Section \ref{Sec:subs2to2}. 

\subsection{Strategy of proof of Theorem \ref{thm.1}: combining semigroup and hypocoercivity methods}\label{Sec:subs1to2} 
In this section we fix $n=N=1$ and $m=2$, i.e. we consider a Markov generator on $\R^2$ of the form  \eqref{operator} and we assume that $B$, $Z_0$ and $Z_1:=[B,Z_0]$ span $\R^2$ at each point.\footnote{In many applications one finds that only $Z_0$ and $Z_1$ are actually needed to fully span $\R^2$. See for example \cite{MVII}. } We are interested in determining the time behaviour of the fields $Z_0$ and $Z_1$, along the semigroup $f_t\equiv e^{t L} f$, i.e. we want to study the time behaviour of  $Z_0 f_t$  and $Z_1 f_t$. Notice that in this simple case the quadratic form $Q^{(1)}_t$ defined in \eqref{13B} - which, for the purposes of this section, we will just denote by $Q_t$ - reduces to
$$
Q_t (f_t)= d \lv f_t\rv^2+  a_0 t \lv Z_0 f_t\rv^2 +a_1 t^3 \lv Z_1 f_t \rv^2 + b t^2 (Z_0 f_t) (Z_1 f_t),
$$
where $d,a_0, a_1$ and $b$ are strictly positive constants to be determined later.
 To explain why we use such a time-dependent quadratic form,     let us start with a simple observation: suppose we consider, instead of $Q_t$, the function $\tilde{Q}_t$ defined as follows:
$$
\tilde{Q}_t (f_t) = d \lv f_t\rv^2+  a_0 t \lv Z_0 f_t\rv^2 +a_1 t^3 \lv Z_1 f_t \rv^2 \,.
$$
If we could prove 
\be\label{tildeQ}
\pa_t (\tilde{Q}_tf_t) < 0 \qquad \mbox{for } t \mbox{ in some interval say } [0,T], 
\ee
then  we would be done as the above would imply
$$
\tilde{Q}_tf_t< \tilde{Q}_0f = d \lv f\rv^2 \,\,\Longrightarrow \,\, 
\lv Z_0 f_t\rv^2 < \frac{d a_0^{-1}}{t} \lv  f\rv^2 \mbox{ and }
\lv Z_1 f_t\rv^2 < \frac{d a_1^{-1}}{t^3} \lv  f\rv^2,
$$
for all $t \in [0, T]$.  However, as long as we use the time dependent form $\tilde{Q}_t$,  \eqref{tildeQ} is in general not true. Indeed, roughly speaking, in order to prove \eqref{tildeQ}, one usually needs to prove that
$$
\pa_t (\tilde{Q}_tf_t) < -\kappa \left( \lv Z_0 f_t\rv^2 +
 \lv Z_1 f_t \rv^2\right), \quad \mbox{ for some } \kappa>0.
$$
If we use the form $\tilde{Q}_t$, the negative  terms $- \kappa  \lv Z_1 f_t \rv^2$ will not appear in the expression for $\pa_t (\tilde{Q}_tf_t)$. The mixed term $(Z_0 f_t) (Z_1 f_t)$  is added to the quadratic form precisely to solve this issue. Such a trick was introduced in \cite{Herau2007}  and then pushed forward in \cite{V}. However in both works the authors used quadratic forms which did not contain the pointwise values of the function $f_t$ and its derivatives, but rather the weighted $L^2$ norm of such quantities. It is important to stress that, 
using the  quadratic Young's inequality, i.e.
\be\label{young}
 \forall \, x,y\in \R,\, \mathfrak{d}>0\qquad \qquad
\lv xy\rv \leq \frac{\lv x\rv^2}{2\mathfrak{d}}+\frac{\mathfrak{d}\lv y\rv^2}{2},
\ee
with $\mathfrak{d}$   a constant times a suitable positive 
power of $t$, 
we can show that there exists a suitable choice of the constant $b$ such that $Q_t$ is still  positive. Indeed, choosing $\mathfrak{d}=t/b$, we obtain \be\label{boundbelow}
Q_t(f_t)\geq d \lv f_t\rv^2+ t (a_0- b^2/2)  \lv Z_0 f_t\rv^2 + t^3 
(a_1 -1/2)  \lv Z_1 f_t \rv^2 \geq 0. 
\ee
Hence,  choosing $a_0 > b^2/2$ and $a_1 >1/2$ guarantees the positivity of $Q_t f_t$. 
Unfortunately, even after this modification, it is still the case that the inequality $\pa_t (Q_t f_t)<0$ is in general not true.  We therefore devise  another strategy, which makes use of the classic Bakry-Emery semigroup approach: instead of trying to prove that  $\pa_t ( Q_tf_t)<0$, we show
\be\label{propro}
\pa_{s}(P_{t-s}Q_s(f_s)) <0, \quad \mbox{for } t \mbox{ in some interval  } [0,T].
\ee
Integrating the above inequality in $[0,t]$, we obtain 
\begin{align*}
& P_0 (Q_t(f_t))- P_t (Q_0 f) <0  \\
 \Rightarrow &  \,\,Q_t (f_t) < d \|f \|_{\infty}^2 \quad \mbox{for } t\in [0,T],
\end{align*} 
which, thanks to \eqref{boundbelow},  implies the desired bounds. Notice that in the above we used the contractivity of the Markov semigroup. In general one will just have $Q_t(f_t)< P_t (f_t^2)$. 

 A straightforward calculation shows that proving the property \eqref{propro} reduces to showing 
$$
(-L+ \pa_t) (Q_t (f_t)) <0. 
$$
Proving such an inequality is done by repeatedly using the Young's inequality, in the same way  shown in \eqref{boundbelow}. We now turn to the full  proof of Theorem \ref{thm.1}.

 \subsection{Proof of Theorem \ref{thm.1}}\label{Sec:subs2to2}
Throughout the proof of Theorem  \ref{thm.1} we will often use some elementary  facts, which we gather in Lemma  \ref{lem.1}, Lemma \ref{lem.2} and Lemma \ref{lem.3} below, for the reader's convenience. 

\begin{lemma} \label{lem.1}
For any $n\in\mathbb{N}$ and $\kk\equiv (k_1,..,k_n)$ and any smooth function $f$ the following relations hold true:
\begin{align}
& L \vert \Z_{\kk, n}f \vert^2 -
2\left(L\Z_{\kk, n} f \right)
\left( \Z_{\kk, n}f \right)=
 +2\vert Z_0 \Z_{\kk, n}f \vert^2,  \nonumber\\
& Z_0^2 \vert \Z_{\kk, n}f \vert^2  =
2\left(Z_0^2\Z_{\kk, n} f \right)
\left( \Z_{\kk, n}f \right) 
 +2\vert Z_0 \Z_{\kk, n}f \vert^2,  \nonumber 
\end{align}
If $[Z_N, B]=0$ then
\begin{align}
& [\Z_{\kk, n},B]=-\sum_{1\leq j:k_j\neq N}^{n} 
\Z_{\kk+\mathbf{e}_j,n} . \label{rec1}
\end{align}
The above equality also simply holds if $k_j\neq N$ for all $j=1 \dd N$. Finally, if $c_{0jh}=0$ for all $j$ (i.e. if $[Z_0, Z_j]=0$ for all $j$) then $[\Z_{\kk,n}, L]=[\Z_{\kk,n},B]$ for any $n \geq 1$.
\end{lemma}
\begin{proof}[Proof of Lemma \ref{lem.1}] 
The first relation is a general property of a generator of (sub)diffusion with second order part given by $Z_0$ and the second is a just a different version
of the same. 
Recalling that for any three operators $X, Y$ and $W$, 
\begin{equation}\label{commrule}
[XY,W]=X[Y,W]+[X,W]Y,
\end{equation}
from (\ref{comm}), for $k_1\neq N$, we have
\begin{equation}\label{it1}
[\Z_{\kk, n}, B]= 
Z_{k_1}[Z_{k_2}\cdot{\dots}\cdot Z_{k_{n}},B]-
Z_{k_1+1}Z_{k_2}\cdot{\dots}\cdot Z_{k_{n}} \equiv Z_{k_1}[Z_{k_2}\cdot{\dots}\cdot Z_{k_{n}},B]- \Z_{\kk+\mathbf{e}_1, n}.
\end{equation}
Iterating (\ref{it1}) one obtains (\ref{rec1}). Regarding  the last statement, this can be obtained,  when
$[Z_0,Z_j]=0$ for any $j$, by using \eqref{commrule}.
\end{proof}
\begin{lemma} \label{lem.2}
Let $X$ and $Y$ be first order differential operators and $\ml=X^2+Y$. Assume $\ml$ generates a semigroup such that for any smooth functions
$h$ also  $h_t=e^{t\ml} h$ is smooth. Then for any differential operators $W, V$ (of any order $ l\geq 0$), we have
\begin{align} 
\left(-\ml +\frac{\pa}{\pa_t}\right)  Wh_t \cdot Vh_t &= - X^2  \left( Wh_t \cdot Vh_t\right)- Y\left( Wh_t \cdot Vh_t\right)+(W \ml h_t)(Vh_t) + (W h_t)(V\ml h_t)  \nonumber \\
&=-2  XW h_t  \cdot XV h_t +  ([W,\ml]h_t) \cdot V h_t +  W h_t \cdot ([V,\ml]h_t).\nonumber
\end{align}
%
%
\end{lemma}

\begin{lemma}\label{lem.3}
Suppose 
\[
[Z_{k},Z_0] = -\sum_{0\leq l< k-1}^{N} c_{0kl}Z_{l}\,  .
\]
Then there exist real numbers ${\boldsymbol\eta}_{\kk,\kk',n}\equiv{\boldsymbol\eta}_{\kk,\kk'}$, ${\boldsymbol\zeta}_{\kk,\kk',n}\equiv{\boldsymbol\zeta}_{\kk,\kk'}$, such that
\be
[\Z_{\kk,n},Z_0^2]
 =  \sum_{\kk'}{\boldsymbol\eta}_{\kk,\kk'} Z_0\Z_{\kk',n} 
 +\sum_{\kk'}{\boldsymbol\zeta}_{\kk,\kk'} \Z_{\kk',n}  \nonumber
\ee
with
$0\leq |\kk'|_n < |\kk|_n -1\leq nN$.
\end{lemma}

The proof of the above Lemma \ref{lem.3} can be found after the proof of Theorem \ref{thm.1}.

\begin{proof}[Proof of Theorem \ref{thm.1}] 
We will show that given $n\in\mathbb{N}$, for all $0 \leq l\leq n$ one can choose the coefficients $a_{\kk,l} , b_{\kk,l}, d_l\in(0,\infty)$, so that
$$
\mbox{for all }\,\,  0 \leq l \leq n,  \qquad
\partial_s \left[ P_{t-s}\left(Q_s^{(l)}f_s\right)\right] < 0;
$$
hence, integrating on $[0,t]$, for $t\in(0,T]$ (for some $T>0$ to be determined later), we get
\be 
Q_t^{(l)} f_t  =P_0\left(Q_t^{(l)} f_t\right) < P_t\left(Q_0^{(l)}f_0\right) \equiv d  P_t f^2 .\nonumber
\ee
 
\noindent \label{Simplified Case}
 Because 
$$
\partial_s  \left[P_{t-s}\left(Q_s^{(n)}f_s\right)\right]=P_{t-s}\left(
-LQ_s^{(n)}f_s+\partial_s Q_s^{(n)}f_s\right),
$$
and the semigroup $P_t$ preserves positivity, the whole thing boils down to proving that $\forall n\geq 1$ there exist strictly positive constants 
$
\{ a_{\kk,n}: \, 0\leq|\kk |_n\leq n N\},
$ 
$
\{ b_{\kk,n} : \, 0\leq|\kk |_n\leq n N\,  \mbox{ with } k_1\geq 1  \}
$ and $d_n\in(0,\infty)$ 
such that 
\begin{equation}\label{what_to_prove}
\forall t>0, \qquad\qquad
\left(-L+\partial_t\right)\left(Q_t^{(n)} f_t\right)\, \leq \, 0. \phantom{AAAA}
\end{equation}
In order to streamline the proof we first consider  Assumption  \textup{(\textbf{CR.I})} with $c_j=0$ and $c_{0jh} = 0$ for all $j=1,..,N$, i.e. we first prove 
\eqref{gammabound1}. Later  we will explain how to remove this restriction (at the cost of obtaining bounds that are valid only for small $T$) and obtain \eqref{gammabound}.  

\medskip

$\bullet$ {\bf Proof of \eqref{gammabound1}.}   Suppose that Assumption  \textup{(\textbf{CR.I})} holds  with $c_j=0$ and $c_{0jh} = 0$ for all $j=1,..,N$. 
We will prove (\ref{what_to_prove}) by induction on $n$. 
The inductive basis, i.e. the proof that for  $n=0$ there exists $d\in(0,\infty)$ such that $\forall t>0 \quad \left(-L+\partial_t\right)Q^{(0)}_tf_t\leq 0$, is straightforward. Indeed
$$
\left(-L+\partial_t \right)\vert f_t\vert^2= -Z_0^2 \vert f_t\vert^2-B\vert f_t\vert^2+ 2f_tL f_t=-2 \vert Z_0 f_t\vert^2 \leq 0, 
$$
where we simply used the fact that $Z_0^2$ is a second order differential operator and $B$ is a first order differential operator.

Now we make an inductive assumption that for any $n\geq 1$ and for all $l=1,..,n-1$ there exist strictly positive constants  
$
\{a_{\kk,l}: \, 0\leq|\kk |_l\leq l N\}, 
$ 
$
\{  b_{\kk,l} : \, 0\leq|\kk |_l\leq l N \mbox{ with } k_1\geq 1 \} 
$ and $d_l,\bar d_l\in(0,\infty)$, 
 such that  
\[\,\forall t>0 \qquad 
\left(-L+\partial_t\right)\left(Q_t^{(n-1)} f_t\right)\leq 0,
\]
and 
\[ \forall \, l=1,..,n-1 \qquad \bar d_l\bar \Gamma_t^{(l)} f_t \leq \Gamma_t^{(l)} f_t\, .\] 
Under this inductive  assumption we need to prove  that there exist strictly positive constants  \linebreak 
$\{a_{\kk,n}: \, 0\leq|\kk |_n\leq n N\},$ 
$\{  b_{\kk,n} : \, 0\leq|\kk |_n\leq nN  \mbox{ with } k_1\geq 1 \}$  and $d_n, \bar d_n\in(0,\infty)$  such that 

$$\forall t>0 \qquad
\left(-L+\partial_t\right)
\left(Q^{(n)}_t f_t\right)\leq 0
$$
and 
$$
\bar d_l\bar \Gamma_t^{(n)} f_t \leq \Gamma_t^{(n)} f_t\, .
$$

\noindent 
Because we are assuming that 
$\left(-L+\partial_t\right)\left(Q_t^{(n-1)} f_t\right)\leq 0$, for some appropriate choice of the constants, looking at \eqref{7star},  we only need to study the following quantity: 
\begin{align*}
\left(-L+\partial_t\right)\left(\Gamma^{(n)}_t f_t\right)=&\left(-L+\partial_t\right)
\sum_{|\kk|_n=0}^{nN} a_{\kk,n} t^{2|\kk|_n+n}
\vert \Z_{\kk,n}f_t\vert^2\\
&+\left(-L+\partial_t\right) \sum_{0\leq |\kk|_n: k_1\geq 1}^{nN}
b_{\kk,n} t^{2|\kk|_n+n-1}
 (\Z_{\kk-\mathbf{e}_1,n}f_t)(\Z_{\kk,n}f_t).
\end{align*}
%
We stress that throughout this calculation, we intend for all the constants $b_{N+1, k_2 \dd k_n}$ to be equal to zero.
To further expand the expression on the right hand side of the above, we use  Lemma \ref{lem.1}  together with  Lemma \ref{lem.2}, for our generator (\ref{operator}), and we obtain
\begin{align}
&\left(-L+\partial_t\right)\left(\Gamma^{(n)}_t f_t\right)=\nonumber\\
&-2\sum_{|\kk|_n=0}^{nN} a_{\kk,n} t^{2|\kk|_n+n}
\vert  Z_0 \Z_{\kk,n}f_t\vert^2  \label{I1}\\
&- 2\sum_{0\leq |\kk|_n: k_1\geq 1}^{nN}
b_{\kk,n} t^{2|\kk|_n+n-1}
 ( Z_0\Z_{\kk-\mathbf{e}_1,n}f_t)( Z_0\Z_{\kk,n}f_t) \label{I2}\\
&
+\sum_{|\kk|_n=0}^{nN}  a_{\kk,n} (2|\kk|_n+n)t^{2|\kk|_n+n-1}
\vert \Z_{\kk,n}f_t\vert^2
\label{II1}\\
&
+ \sum_{0\leq |\kk|_n: k_1\geq 1}^{nN}
b_{\kk,n} (2|\kk|_n+n-1) t^{2|\kk|_n+n-2}
 (\Z_{\kk-\mathbf{e}_1,n}f_t)(\Z_{\kk,n}f_t)
\label{II2}\\
&+ 2  \sum_{|\kk|_n=0}^{nN} a_{\kk,n} t^{2|\kk|_n+n}
  \Z_{\kk,n}f_t\cdot   [ \Z_{\kk,n},L]f_t              \label{III1}\\
&+\sum_{0\leq |\kk|_n: k_1\geq 1}^{nN}
b_{\kk,n} t^{2|\kk|_n+n-1} \left\{
 ([\Z_{\kk-\mathbf{e}_1,n},L]f_t) \Z_{\kk,n}f_t   +  (\Z_{\kk-\mathbf{e}_1,n}f_t)[\Z_{\kk,n},L]f_t  
 \right\}. \label{III2}
\end{align}
We control these terms as follows. 
Let us set 
$$[\mathrm{I}]\equiv(\ref{I1})+(\ref{I2}),\quad [\mathrm{II}]\equiv(\ref{II1})+(\ref{II2}),
\quad [\mathrm{III}]\equiv(\ref{III1})+(\ref{III2})
$$ 
and study these addends separately. Recall now  the  quadratic Young's inequality \eqref{young}, which we will repeatedly use. In particular we will choose    $\mathfrak{d}$  in \eqref{young} to be a constant times a suitable positive power of $t$. The time dependent factor will be relevant for bounds involving factors with differential operators of different rank to obtain time dependences of appropriate homogeneity. 
We have
\begin{align}\label{est:I}
[\mathrm{I}]\leq &
-2\sum_{|\kk|_n=0}^{nN} a_{\kk,n} t^{2|\kk|_n+n}
\vert  Z_0 \Z_{\kk,n}f_t\vert^2  \\
&
+ 2  \sum_{0\leq |\kk|_n: k_1\geq 1}^{nN}
b_{\kk,n}\left( b_{\kk,n}  t^{2|\kk|_n+n-2}
 | Z_0\Z_{\kk-\mathbf{e}_1,n}f_t|^2+
t^{2|\kk|_n+n} |Z_0\Z_{\kk,n}f_t|^2 /b_{\kk,n} \right).
 \label{Isecond}
\end{align}
We look separately at the terms with $k_1=0$ and at  the terms with $k_1>0$. In doing so, we need to notice that  terms of the form $\lv Z_0^2Z_{k_2}\cdot{\dots}\cdot Z_{k_{n}}f_t\rv$ (i.e. those with $k_1=0$) come from  (\ref{est:I}) when $k_1=0$ but also from the first addend in $(\ref{Isecond})$ when $k_1=1$. Hence  
\begin{align*}
[\mathrm{I}]&\leq  2 \sum_{k_2,..,k_n=0}^N
\left(-a_{0,k_2,\dots,k_{n}}+b^2_{1,k_2,\dots,k_{n}}
\right)t^{2|\kk-k_1\e_1 |_n+n}\lv Z_0^2\Z_{k_2,..,k_{n}}f_t \rv^2\\
&+ 2 \sum_{0\leq |\kk|_n : k_1\geq 1}^{nN}\left(-a_{k_1,k_2,\dots,k_{n}}+b^2_{k_1+1,k_2,\dots,k_{n}}
+1\right)t^{^{2|\kk|_n+n}}
\lv Z_0\Z_{\kk,n} f_{t}\rv^2, 
\end{align*}
with the understanding that $b_{k_j+1}=0$ if $k_j=N$. 
Thus we can make term [I] nonpositive choosing
\begin{equation}
\label{cond.1}
\frac12 a_{k_1,k_2,\dots,k_{n}} > b^2_{k_1+1,k_2,\dots,k_{n}}
+1, \qquad k_1\geq 0.
\end{equation}
We can and do assume that similar \textsl{strict} inequality is satisfied on induction level $n-1$. 
We repeat the same kind of procedure for $[\mathrm{II}]$, applying first Young's inequality and then looking \nopagebreak[0] separately at the two cases $k_1=0$ and $k_1>0$ to get
\begin{align*}
[\mathrm{II}]&\leq 
\sum_{ |\kk|_n\geq 0}^{nN} 
a_{\kk,n} (2|\kk|_n+n) t^{2|\kk|_n+n-1}
 \lv \Z_{\kk,n} f_t \rv^2\\
&+\frac12\sum_{0\leq |\kk|_n : k_1\geq 1}^{nN}  b_{\kk,n} 
(2|\kk|_n+n-1)
\left( \varepsilon t^{2|\kk|_n+n-3}
\lv \Z_{\kk-\e_1,n} f_t \rv^2 
+ \varepsilon^{-1} t^{2|\kk|_n+n-1}
\lv \Z_{\kk,n} f_t \rv^2   \right)\\
&\leq (2nN+n)\sum_{ |\kk|_n\geq 0}^{nN} 
\left(a_{0,k_2,\dots,k_{n}} + \frac{\varepsilon}2 b_{1,k_2,\dots,k_{n}}\right)
t^{2|\kk|_n+n-1}\lv Z_0\Z_{k_2,.., k_{n}}f_t \rv^2 \\
&+ (2nN+n)\sum_{0\leq |\kk|_n : k_1\geq 1}^{nN} 
\left( a_{k_1,\dots,k_{n}}+{\varepsilon}^{-1}b_{k_1+1,k_2,\dots,k_{n}}+
\frac{\varepsilon}2 b_{k_1,..,k_n}   \right)
t^{2|\kk|_n+n-1}\lv \Z_{\kk,n} f_t \rv^2, 
\end{align*}
 with ${\varepsilon}< (2nN+n)^{-1}$  to be chosen later.
\nopagebreak[0]
Before turning to [\textup{III}] notice that, because of our current simplified assumption, we have 
$[Z_N,B]= [Z_0, Z_j]=0$ for all $j$; therefore we have $[\Z_{\kk,n} , L]=[\Z_{\kk,n} ,B]$ (see last statement of Lemma \ref{lem.1}) and  by (\ref{rec1})
\[[\Z_{\kk,n}, L] = -\sum_{1\leq j:k_j\neq N}^{n} 
\Z_{\kk+\mathbf{e}_j,n} .
\] 
Using this, we have
\begin{align*}
[\mathrm{III}] &= -2  \sum_{|\kk|_n=0}^{nN} a_{\kk,n} t^{2|\kk|_n+n}
  \Z_{\kk,n}f_t\cdot   \left(\sum_{1\leq j:k_j\neq N}^{n} 
\Z_{\kk+\e_j,n} f_t  \right)            \\
&\quad -\sum_{0\leq |\kk|_n: k_1\geq 1}^{nN}
b_{\kk,n} t^{2|\kk|_n+n-1} \left(
\left( \sum_{1\leq j:k_j \neq N}^{n} 
\Z_{\kk-\e_1+\e_j,n} f_t \right) \Z_{\kk,n}f_t    \right)   \\
&\quad -\sum_{0\leq |\kk|_n: k_1\geq 1}^{nN}
b_{\kk,n} t^{2|\kk|_n+n-1} \left(
  \Z_{\kk-\mathbf{e}_1,n}f_t \left( \sum_{1\leq j:k_j\neq N}^{n} 
\Z_{\kk+\mathbf{e}_j,n} f_t\right)  
 \right) \\
 & \equiv [\mathrm{III}a]+[\mathrm{III}b]+[\mathrm{III}c].
 \end{align*}
 Each of the sums on the right hand side is bounded using (\ref{young})
 to adjust the power of $t$ according to the rank of the corresponding differential operators as follows:
 \begin{align*}
[\mathrm{III}a] &\leq    \sum_{|\kk|_n=0}^{nN} a_{\kk,n} t^{2|\kk|_n+n}  \left(
a_{\kk,n}  t^{-1}\lv\Z_{\kk,n}f_t\rv^2 +
  a_{\kk,n}^{-1} n  t \sum_{1\leq j:k_j\neq N}^{n}  
\lv\Z_{\kk+\e_j,n} f_t\rv^2   \right)            \\
&\leq  \sum_{|\kk|_n=0}^{nN} \left( a_{\kk,n}^2 +n^2 \right) t^{2|\kk|_n+n-1}  \lv\Z_{\kk,n}f_t\rv^2
\end{align*}
Also,  
\begin{align*}
 [\mathrm{III}b]  &=  -\sum_{0\leq |\kk|_n: k_1\geq 1}^{nN}
b_{\kk,n} t^{2|\kk|_n+n-1} 
 \lv\Z_{\kk,n}f_t \rv^2     \\
   & -\sum_{0\leq |\kk|_n: k_1\geq 1}^{nN}
b_{\kk,n} t^{2|\kk|_n+n-1} \left(
\left( \sum_{2\leq j:k_j\neq N}^{n} 
\Z_{\kk-\e_1+\e_j,n} f_t \right) \Z_{\kk,n}f_t    \right) \\
 &\leq  -\sum_{0\leq |\kk|_n: k_1\geq 1}^{nN}
b_{\kk,n} t^{2|\kk|_n+n-1} 
 \lv\Z_{\kk,n}f_t \rv^2   \\
 & +\frac12 \sum_{0\leq |\kk|_n: k_1\geq 1}^{nN}
b_{\kk,n} t^{2|\kk|_n+n-1} \left(
(n-1) \sum_{2\leq j:k_j\neq N}^{n} 
\lv\Z_{\kk-\e_1+\e_j,n} f_t \rv^2 +\lv\Z_{\kk,n}f_t \rv^2   \right)  \\
 & =  \sum_{0\leq |\kk|_n: k_1\geq 1}^{nN}
\left( -\frac12b_{\kk,n} 
 +\frac{(n-1) }2  
\sum_{2\leq j:k_j\neq N}^{n} b_{\kk+\e_1-\e_j,n} 
  \right) t^{2|\kk|_n+n-1}  \lv\Z_{\kk,n}f_t \rv^2\, . 
\end{align*} 
  \begin{align*}
 [\mathrm{III}c] &\leq
 \frac12 \sum_{0\leq |\kk|_n: k_1\geq 1}^{nN}
b_{\kk,n} t^{2|\kk|_n+n-1} \left( b_{\kk,n} t^{-2}
 \lv \Z_{\kk-\e_1,n}f_t \rv^2 +
 n b_{\kk,n}^{-1} t^{2}\sum_{1\leq j:k_j\neq N}^{n} 
\lv\Z_{\kk+\e_j,n} f_t\rv^2
 \right)  \\
 & = \sum_{0\leq |\kk|_n: 0\leq k_1\leq N-1}^{nN}
\frac{b_{\kk+\e_1,n}^2}2  t^{2|\kk|_n+n-1}  
 \lv \Z_{\kk,n}f_t \rv^2 +
   \frac{n}2 \sum_{0\leq |\kk|_n: k_1\geq 1}^{nN} 
t^{2|\kk|_n+n+1 }\sum_{1\leq j:k_j\neq N}^{n} 
\lv\Z_{\kk+\e_j,n} f_t\rv^2
 \\
 & \leq  \sum_{0\leq |\kk|_n: k_1=0}^{nN}
\frac{b_{\kk+\e_1,n}^2}2  t^{2|\kk|_n+n-1}  
 \lv Z_0\Z_{k_2,..,k_n,n}f_t \rv^2 
 +
\sum_{0\leq |\kk|_n: k_1\geq 1}^{nN}
\left(\frac{b_{\kk+\e_1,n}^2}2 + \frac{n^2}2\right)  t^{2|\kk|_n+n-1}  
 \lv \Z_{\kk,n}f_t \rv^2.
\end{align*}
 We now combine and reorganize the bounds of  [\textup{III}] separating terms with $k_1=0$ which need to be offset by level $(n-1)$,
 (if necessary scaling coefficients of $Q_s^{(n-1)}$ by a positive  sufficiently large constant), and the ones with $k_1\geq 1$ which can only be offset by negative contribution in [\textup{III}b], as follows.
 \begin{align*}
[\mathrm{III}] &\leq 
\sum_{|\kk|_n=0}^{nN} \left( a_{\kk,n}^2 +n^2 \right) t^{2|\kk|_n+n-1}  \lv\Z_{\kk,n}f_t\rv^2  \\
& +\sum_{0\leq |\kk|_n: k_1\geq 1}^{nN}
\left( -\frac12b_{\kk,n} 
 +\frac{(n-1) }2  
\sum_{2\leq j:k_j\neq N}^{n} b_{\kk+\e_1-\e_j,n} 
  \right) t^{2|\kk|_n+n-1}  \lv\Z_{\kk,n}f_t \rv^2 \phantom{AAAAA} \\
& +\!\!\!\!   \sum_{0\leq |\kk|_n: k_1=0}^{nN} \!\!\!\!
\frac{b_{\kk+\e_1,n}^2}2\,  t^{2|\kk|_n+n-1}  
 \lv Z_0\Z_{k_2,..,k_n,n}f_t \rv^2 
 +\!\!\!\!
\sum_{0\leq |\kk|_n: k_1\geq 1}^{nN}
\left(\frac{b_{\kk+\e_1,n}^2}2 + \frac{n^2}2\right)  t^{2|\kk|_n+n-1}  
 \lv \Z_{\kk,n}f_t \rv^2.
\end{align*} 
Therefore
 \begin{align*}
 [\mathrm{III}] &\leq 
 \sum_{0\leq |\kk|_n: k_1=0}^{nN}  \left(a_{\kk,n}^2 +n^2+
\frac{b_{\kk+\e_1,n}^2}2  \right)  t^{2|\kk|_n+n-1} 
 \lv Z_0\Z_{k_2,..,k_n,n}f_t \rv^2   \\
& + \!\!
\sum_{0\leq |\kk|_n: k_1\geq 1}^{nN} \!\!
\left(a_{\kk,n}^2 +n^2 \!
-\frac12b_{\kk,n} 
 +\frac{(n-1) }2  \!\!\! 
\sum_{2\leq j:k_j\neq N}^{n} b_{\kk+\e_1-\e_j,n}  
+\frac{b_{\kk+\e_1,n}^2}2 + \frac{n^2}2\right)   
t^{2|\kk|_n+n-1}  \lv \Z_{\kk,n}f_t \rv^2  . 
\end{align*}
 Combining this with bound  of  [\textup{II}] and  separating the terms with $k_1=0$ (which need to be offset by level $(n-1)$) and the ones with
 $k_1\geq 1$ which can only be offset by the  negative contribution in [\textup{III}b], we obtain
\be 
[\mathrm{II}]+ [\mathrm{III}]\leq  
\sum_{0\leq |\kk|_n: k_1=0}^{nN}  
 \mathcal{A}_{k_2,\dots,k_{n}} 
 t^{2|\kk|_n+n-1} 
 \lv Z_0\Z_{k_2,..,k_n,n}f_t \rv^2 \\
+ \sum_{0\leq |\kk|_n : k_1\geq 1}^{nN} 
 \mathcal{B}_{\kk,n}
t^{2|\kk|_n+n-1}\lv \Z_{\kk,n} f_t \rv^2, \nonumber
\ee
where
$
 \mathcal{A}_{k_2,\dots,k_{n}}\equiv  
(2nN+n)\left(a_{0,k_2,\dots,k_{n}} + \frac{\varepsilon}2 b_{1,k_2,\dots,k_{n}}
\right) 
+\left( a_{0,k_2,\dots,k_{n}}^2 +n^2+
\frac{b_{1,k_2,\dots,k_{n}}^2}2 \right) 
$ and 
\begin{align}
 \mathcal{B}_{\kk,n}&\equiv 
(2nN+n)\left( a_{k_1,\dots,k_{n}}+{\varepsilon}^{-1}b_{k_1+1,k_2,\dots,k_{n}}+
\frac{\varepsilon}2 b_{k_1,..,k_n}   \right) \nonumber \\  
\qquad &+  a_{\kk,n}^2 +n^2 
-\frac12b_{\kk,n} 
 +\frac{(n-1) }2  
\sum_{2\leq j:k_j\neq N}^{n} b_{\kk+\e_1-\e_j,n}  
+\frac{b_{\kk+\e_1,n}^2}2 + \frac{n^2}2    
\nonumber\\
&= -\frac12\left( 1 - \varepsilon n(2N+1)\right) b_{\kk,n} +
(2nN+n)\left( a_{\kk,n}+{\varepsilon}^{-1}b_{\kk+\e_1,n}\right) 
 \nonumber\\  
\qquad &+  a_{\kk,n}^2 +n^2 
 +\frac{(n-1) }2  
\sum_{2\leq j:k_j\neq N}^{n} b_{\kk+\e_1-\e_j,n}  
+\frac{b_{\kk+\e_1,n}^2}2 + \frac{n^2}2.    \nonumber
\end{align}
\label{p.10} We note that the terms involving $\mathcal{A}_{k_2,\dots,k_{n}}$ can be offset by $Q^{(n-1)}_t$ (possibly at a cost of redefining its coefficients by multiplying them by a sufficiently large positive constant). On the other hand choosing
$\varepsilon n(2N+1)<1$  and the coefficients so that we have
\begin{align*}
 \mathcal{B}_{\kk,n} &\leq -\frac12\left( 1 - \varepsilon n(2N+1)\right) b_{\kk,n} +(2nN+n)\left( a_{\kk,n}+{\varepsilon}^{-1}b_{\kk+\e_1,n}\right)  \\  
 &+  a_{\kk,n}^2 +n^2 
 +\frac{(n-1) }2  
\sum_{2\leq j:k_j\neq N}^{n} b_{\kk+\e_1-\e_j,n}  
+\frac{b_{\kk+\e_1,n}^2}2 + \frac{n^2}2 \, < \, 0,
\end{align*}
this and \eqref{cond.1} can be represented as the following condition
\begin{equation} \label{fruit1}
a_{\kk,n} >> b^2_{\kk+\e_1,n},\qquad
b_{\kk,n} >> a_{\kk,n}, \qquad b_{\kk,n} >> b_{\kk+\e_1,n}, \qquad b_{\kk,n} >> b_{\kk+\e_1-\e_j,n}, \,\, j\geq 2\, ,
\end{equation} 
with a convention $x>>y$ meaning $x\geq Cy^2+C'$ with some constants $C,C'\in[1,\infty)$ sufficiently large and  possibly dependent on $n$, but not on $\kk$.
In this way we get (\ref{what_to_prove}).

We are now left with proving the following  statement
\begin{equation} \label{bGamma}
\bar d_n \bar \Gamma^{(n)}_sf_s\leq \Gamma^{(n)}_sf_s,
\end{equation}
for some $\bar d_n\in(0,\infty)$.
     To this end, we will use the lower bound implied
 by the quadratic Young inequality  
 \be  \label{younglower}
  -\frac{\lv x\rv^2}{\mathfrak{d}}
  - \mathfrak{d}\lv y\rv^2  \leq  xy ,\qquad \forall x,y\in \R,\, \mathfrak{d}>0. \nonumber
 \ee
We separate the terms with $k_1=0$ from 
the terms with $k_1 > 0$ so we get
\begin{align*}
\Gamma_t^{(n)}f_t &=   \sum_{k_2,..,k_n=0}^N\left[
a_{0,k_2 \dd k_n} t^{2\left(\sum_{j=2}^{n} k_j\right)+n}
\lv Z_0Z_{k_2} \cdot{ \dots } \cdot Z_{k_{n}} f_t\rv^2
\right.\\
&  + b_{1,k_2 \dd k_{n}}t^{2\left(\sum_{j=2}^{n}k_j \right)+n+1} 
\left( Z_0Z_{k_2} \cdot{ \dots } \cdot Z_{k_{n}} f_t \right)
\left( Z_1Z_{k_2} \cdot{ \dots } \cdot Z_{k_{n}} f_t \right)
\left. \right]\\
&+\!\!\!\sum_{k_2,..,k_n=0}^N \skii a_{k_1,..,k_n} t^{2\lv \kk \rv_{n}+n} 
\lv Z_{k_1} \cdot{ \dots } \cdot Z_{k_{n}} f_t\rv^2\\
&+ \!\!\!\sum_{k_2,..,k_n=0}^N \sum_{k_1=2}^N b_{k_1,..,k_n}  t^{2\lv \kk \rv_{n}+n-1}  \left( Z_{k_1-1}Z_{k_2} \cdot{ \dots } \cdot Z_{k_{n}} f_t \right) 
 \left( Z_{k_1}Z_{k_2} \cdot{ \dots } \cdot Z_{k_{n}} f_t \right).
\end{align*}
We now use the inequality  \eqref{younglower} on the second and fourth line of the above equations, with   $\mathfrak{d}=\frac{t}{b_{1,k_2 \dd k_{n}}}$ and $\mathfrak{d}=\frac{t}{b_{k_1,k_2 \dd k_{n}}}$, respectively and obtain

 \begin{align*}
 \Gamma_t^{(n)}f_t &  \geq  \sum_{k_2,..,k_n=0}^N\left[
a_{0,k_2 \dd k_{n}} t^{2\left(\sum_{j=2}^{n} k_j\right)+n}
\lv Z_0Z_{k_2} \cdot{ \dots } \cdot Z_{k_{n}} f_t\rv^2
\right.\\
& \left. + t^{2\left(\sum_{j=2}^{n}k_j \right)+n+1}  \left(
-\frac{ b_{1,k_2 \dd k_{n}}^2\lv Z_0Z_{k_2} \cdot{ \dots } \cdot Z_{k_{n}} f_t \rv^2}{t}- t
\lv Z_1Z_{k_2} \cdot{ \dots } \cdot Z_{k_{n}} f_t \rv^2 \right)
 \right]\\
&+\!\!\!\sum_{k_2,..,k_n=0}^N \skii a_{k_1,..,k_n} t^{2\lv \kk \rv_{n}+n} 
\lv Z_{k_1} \cdot{ \dots } \cdot Z_{k_{n}} f_t\rv^2\\
&+ \!\!\!\sum_{k_2,..,k_n=0}^N \sum_{k_1=2}^N   t^{2\lv \kk \rv_{n}+n-1} \left( -\frac{ b_{k_1,..,k_n}^2\lv Z_{k_1-1}Z_{k_2} \cdot{ \dots } \cdot Z_{k_{n}} f_t \rv^2 }{t}
 -t \lv Z_{k_1}Z_{k_2} \cdot{ \dots } \cdot Z_{k_{n}} f_t \rv^2\right)\\
 & \geq  \sum_{k_2,..,k_n=0}^N \left(a_{0,k_2 \dd k_{n}}- b_{1,k_2 \dd k_{n}}^2\right)
 t^{2\left(\sum_{j=2}^{n} k_j\right)+n} \lv Z_0Z_{k_2} \cdot{ \dots } \cdot Z_{k_{n}} f_t\rv^2\\
&+ \sum_{k_2,..,k_n=0}^N\skii\left(a_{k_1,..,k_n}- b_{k_1+1,k_2 \dd k_{n}}^2-1\right)
 t^{2\lv \kk \rv_{n}+n} \lv Z_{k_1} \cdot{ \dots } \cdot Z_{k_{n}} f_t\rv^2.
 \end{align*}
Because of \eqref{fruit1}, $a_{k_1,..,k_n}- Cb_{k_1+1,k_2 \dd k_{n}}^2-C' >0$, with some $C,C'\geq 1$ so one can choose
$ c_{n} >0$ so that the desired bound \eqref{bGamma} is satisfied.
This ends the proof of the simplified case, i.e the proof of \eqref{gammabound1}. 

\medskip

$\bullet$ {\bf Proof of \eqref{gammabound}.}   We now turn to the proof of \eqref{gammabound}, i.e. we remove our simplifying assumption. In this case the expression \eqref{I1}-\eqref{III2} remains unaltered, as well as  the analysis of the terms $[\mathrm{I}]$ and $[\mathrm{II}]$, as we used our simplifying assumption only to estimate $[\mathrm{III}]$. We therefore concentrate on the terms  $[\mathrm{III}]$.

Note that if  in Assumption \textup{(\textbf{CR.I})} $c_{0jh}=0$ for all $j$ but $c_j\neq 0$,
then    \eqref{rec1} no longer holds. More precisely, if $[Z_0, Z_j]=0$ for all $j$ then it is still true that 
$[\Z_{\kk,n}, B]=[\Z_{\kk,n}, L]$, but in this case \eqref{rec1} needs to be modified to take into account $[Z_N, B]\neq 0$. So, 
 when we expand the expression for $[\mathrm{III}]$,  we get the following additional terms:
 \begin{align*}\label{AdTerms}
& [\mathrm{ATI}]=-\sum_{i=0,..,N} c_i \left(2
\sum_{0\leq |\kk|_n}^{nN} \sum_{j=1,..,n} \delta_{k_j,N} a_{\kk,n} t^{2|\kk|_n+n} \Z_{\kk,n}f_t
\cdot  \Z_{\kk+(i-N)\e_j,n} \right.\\
& 
+ \left. \sum_{0\leq |\kk|_n: k_1\geq 1}^{nN} \sum_{j=1,..,n} \delta_{k_j,N}
b_{\kk,n} t^{2|\kk|_n+n-1}
\left( (  \Z_{\kk-\mathbf{e}_1+(i-N)\e_j,n}f_t)(  \Z_{\kk,n}f_t) +(  \Z_{\kk-\mathbf{e}_1,n}f_t)(  \Z_{\kk+(i-N)\e_j,n}f_t) \right)  \right)
 \end{align*}
Since by our assumption $c_N\geq 0$, we either get additional negative term (when $i=N$) with coefficient $c_N a_{\kk,n}  t^{2|\kk|_n+n}$ which can be used to beat those coming from the second sum with mixed terms, or we can apply quadratic Young inequality
to get terms as before but with a higher power of $t$
which for sufficiently small time do not change inequality obtained before in the simplified case.
 \label{p.13}
Now we discuss the general case, $c_j \neq 0$,  $c_{0jh} =0$ for $h\geq j-1$ and not all of them are equal to zero. In this case it is no longer true that $[\Z_{\kk,n},L]=[\Z_{\kk,n},B]$, which is why we need to use Lemma \ref{lem.3} to study $[\mathrm{III}]$.
Using such a lemma we find that,  together with the terms in  [$\mathrm{ATI}$], we also have the following additional contributions to $[\mathrm{III}]$:
\begin{align}
[\mathrm{ATII}]&=  2\sum_{|\kk|_n=0}^{nN} a_{\kk,n} t^{2|\kk|_n+n}
 \Z_{\kk,n}f_t [\Z_{\kk,n},Z_0^2]f_t \nonumber \\ 
 &+\sum_{0\leq |\kk|_n: k_1\geq 1}^{nN}
b_{\kk,n} t^{2|\kk|_n+n-1} \left(
 ([\Z_{\kk-\mathbf{e}_1,n},Z_0^2]f_t) \Z_{\kk,n}f_t   +  (\Z_{\kk-\mathbf{e}_1,n}f_t)[\Z_{\kk,n},Z_0^2]f_t   \right) \nonumber \\ 
%
&=2\sum_{|\kk|_n=0}^{nN} a_{\kk,n} t^{2|\kk|_n+n}
 \Z_{\kk,n}f_t \left( \sum_{|\kk'|_n<|\kk|_{n}-1}{\boldsymbol\eta}_{\kk,\kk'} Z_0 \Z_{\kk',n}   f_t  \right) \nonumber \\
&  + 2\sum_{|\kk|_n=0}^{nN} a_{\kk,n} t^{2|\kk|_n+n}
 \Z_{\kk,n}f_t \left( \sum_{|\kk'|_n<|\kk|_{n}-1}{\boldsymbol\zeta}_{\kk,\kk'} \Z_{\kk',n} f_t  \right)
\nonumber\\
&+ \!\!\!\!\sum_{0\leq |\kk|_n: k_1\geq 1}^{nN}
b_{\kk,n} t^{2|\kk|_n+n-1} 
\left(\sum_{|\kk'|_n<|\kk-\e_1|_n-1} \!\!{\boldsymbol\eta}_{\kk,\kk'} Z_0 \Z_{\kk',n}f_t  \cdot \Z_{\kk,n}f_t 
 +\sum_{|\kk'|_n<|\kk-\e_1|_n-1}\!\! {\boldsymbol\zeta}_{\kk,\kk'} \Z_{\kk',n} f_t \cdot \Z_{\kk,n}f_t  \right)
 \nonumber \\
 &+ \!\!\!\! \sum_{0\leq |\kk|_n: k_1\geq 1}^{nN}
\!\!\!   b_{\kk,n} t^{2|\kk|_n+n-1}  
 \left(\sum_{|\kk'|_n<|\kk|_n-1} \!\!\!\!{\boldsymbol\eta}_{\kk,\kk'} (\Z_{\kk-\mathbf{e}_1,n}f_t)\cdot Z_0\Z_{\kk',n} f_t 
 + \!\!\!\! \sum_{|\kk'|_n<|\kk|_n-1}\!\!\!\! {\boldsymbol\zeta}_{\kk,\kk'}  (\Z_{\kk-\mathbf{e}_1,n}f_t)\cdot \Z_{\kk',n} f_t \right). \nonumber
\end{align}
Because of our restriction on $|\kk'|$, all new terms come with a higher power of $t$ and therefore for sufficiently small time they can be offset by the principal terms discussed in the first stage (when all  $c_j$ and $c_{0kl}$ were assumed to be zero). 
 The proof is concluded once we observe that in order to prove the lower bound \eqref{bGamma}, we did not use the simplified form of Assumption ({\bf CR.I}) and hence such a bound still holds in this general case.   
\end{proof}
\begin{proof}[Proof of Lemma \ref{lem.3} ]
Observe that, with $\{X,Y\}\equiv XY+YX=2XY-[X,Y]$,  we have
\begin{align*} 
[Z_{k_j},Z_0^2] & = \{Z_0,[Z_{k_j},Z_0]\}=  -\sum_{l_j=0}^{N} c_{0 k_j l_j}\{Z_0,Z_{l_j}\}  \\
& =  -2\sum_{l_j=0}^{N} c_{0 k_j l_j} Z_0Z_{l_j} +  \sum_{l_j=0}^{N}
\gamma_{k_jl_j}Z_{l_j},  
\end{align*}
with
\[\gamma_{k_jl_j}\equiv \sum_{l=0}^{N} c_{0k_jl} c_{0ll_j}. \]
Using the commutator relation (\ref{commrule}) and  the above, we get   for $0\leq k_1,..,k_n\leq N$ 
\begin{align}
[Z_{\kk,n},Z_0^2]
&=\sum_{j=1}^{n}Z_{k_1}\cdot{\dots}\cdot Z_{k_{j-1}}[Z_{k_j},Z_0^2]Z_{k_{j+1}}\cdot{\dots}\cdot Z_{k_{n}}\nonumber\\
&=-2 \sum_{j=1}^{n}\sum_{l_j=0}^{N} c_{0k_jl_j} Z_{k_1}\cdot{\dots}\cdot Z_{k_{j-1}}Z_0
Z_{l_j} Z_{k_{j+1}}\cdot{\dots}\cdot Z_{k_{n}} \nonumber\\
& +\sum_{j=1}^{n}\sum_{l=0}^{N}\gamma_{k_jl} 
Z_{k_1}\cdot{\dots}\cdot Z_{k_{j-1}}Z_{l}Z_{k_{j+1}}\cdot{\dots}\cdot Z_{k_{n}}. \nonumber
\end{align}
We repeat the commutation process involving  the operator $Z_0$ until we bring it to the left. In this way we obtain
\be  
[\Z_{\kk,n},Z_0^2]
 =  \sum_{\kk'}{\boldsymbol\eta}_{\kk,\kk'} Z_0\Z_{\kk',n} 
 +\sum_{\kk'}{\boldsymbol\zeta}_{\kk,\kk'} \Z_{\kk',n}   \nonumber
\ee
with the following linear operators 
$$
\sum_{\kk'}{\boldsymbol\eta}_{\kk,\kk'} \Z_{\kk',n}\equiv -2 \sum_{j=1}^{n}\sum_{l_j=0}^{N} c_{0k_jl_j}  Z_{k_1}\cdot{\dots}\cdot Z_{k_{j-1}}
Z_{l_j} Z_{k_{j+1}}\cdot{\dots}\cdot Z_{k_{n}} 
$$
and
\begin{align*}
\sum_{\kk'}{\boldsymbol\zeta}_{\kk,\kk'} Z_{\kk',n}
&\equiv 
+ 2 \sum_{j=2}^{n}\sum_{i=2}^{j-1}\sum_{l_i,l_j=0}^{N} c_{0k_il_i} c_{0k_jl_j} Z_{k_1}\cdot{\dots}Z_{k_{i-1}} 
Z_{l_i} Z_{k_{i+1}} \cdot{\dots}\cdot Z_{k_{j-1}} 
Z_{l_j} Z_{k_{j+1}}\cdot{\dots}\cdot Z_{k_{n}} \nonumber\\
&\qquad +\sum_{j=2}^{n}\sum_{l_i,l_j=0}^{N} c_{0k_1l_1} c_{0k_jl_j} Z_{l_1}Z_{k_2}\cdot{\dots}\cdot Z_{k_{j-1}} 
Z_{l_j} Z_{k_{j+1}}\cdot{\dots}\cdot Z_{k_{n}}\nonumber \\
&\qquad +\sum_{j=1}^{n}\sum_{l_j=0}^{N}\gamma_{k_jl_j} 
Z_{k_1}\cdot{\dots}\cdot Z_{k_{j-1}}Z_{l_j}Z_{k_{j+1}}\cdot{\dots}\cdot Z_{k_{n}} .\nonumber
\end{align*}
Finally we note that because of our assumption on $c_{0,j,h}$, the summation over $\kk'$ is restricted by a condition $|\kk'|<|\kk|-1$. 
 \end{proof} 
\section{Infinite dimensional semigroups} 
\label{S.3:Infinite dimensional semigroups}
From this section on we focus on infinite dimensional dynamics on $(\R^m)^{\ZZ^d}$. 
The present section is organized as follows: in Section \ref{infinitDim} we present the setting and notation used in this infinite dimensional context. In view of the heavily computational nature of this part of the paper,  Section \ref{infinitDim} is complemented with Subsection \ref{subsubsec3}, which explains the strategy used throughout this section in a simplified scenario. In Section \ref{sec:exi} we prove the well posedness of the infinite dimensional dynamics generated by the operator \eqref{genlambda}   (Theorem \ref{existencethm}) and in Section \ref{sec:smoothingng}  the smoothing properties of the associated infinite dimensional semigroup (Theorem \ref{thm:smooth}).  Section \ref{sec:strsec} provides the preliminary estimates needed to prove the results of Section  \ref{sec:exi} and Section \ref{sec:smoothingng}, in particular finite speed of propagation of information type of bounds. 
 \subsection{Setting and notation} \label{infinitDim}
Let us first introduce the relevant spaces and the metrics that they are endowed with. 
The set $\mathbb{Z}^d$, $d \in \mathbb{N}$ with a distance ${dist} (x, y) \equiv \sum_{l = 1}^d | x^l - y^l |$, will be called a lattice (here $x^l$ is just the $l-$th component of $x \in \mathbb{Z}^d$). If a set $\Lambda \subset \mathbb{Z}^d$ is finite, we
denote that by $\Lambda \subset \subset \mathbb{Z}^d$. If $\Lambda \subset \subset \mathbb{Z}^d$ and $x \in \mathbb{Z}^d$ we denote by $d(x, \Lambda)$
the length of the shortest tree connecting each component of $x$ and $\Lambda (f)$. The space $\R^m$ is instead endowed with a metric $\mathbf{d}$. 

 Let $\Omega \equiv (\mathbb{R}^m)^{\mathbb{Z}^d}$. For a set $\Lambda \subset \subset
\mathbb{Z}^d$ and $\omega \equiv (\omega_x \in \mathbb{R}^m)_{x \in
\mathbb{Z}^d} \in \Omega$ we define its projection $\omega_{\Lambda} \equiv
(\omega_x \in  \mathbb{R}^m)_{x \in \Lambda}$ and set $\Omega_{\Lambda}\equiv (\R^m)^{\Lambda}$. A smooth
function $f : \Omega \to \mathbb{R}$ is called a cylinder function iff there
exists a set $\Lambda \subset \subset \mathbb{Z}^d$ and a smooth function
$\phi_{\Lambda} : \Omega_{\Lambda} \to \mathbb{R}$ such that 
$f (\omega) =\phi_{\Lambda} (\omega_{\Lambda})$. The smallest set for which such
representation is possible for a given cylinder function $f$ is denoted by
$\Lambda (f)$. We will then say that $f$ is localized in $\Lambda$. It is
known  (see e.g. \cite{GZ}) that  the set of cylinder functions is dense in the set of
continuous functions on $\Omega$.
 
If $Z$ is a differential operator in $\mathbb{R}^m$, we denote by  $Z_x$ an isomorphic copy of the operator $Z$ acting only  on the  variable $\omega_x$, i.e. $Z_x$ is a copy of $Z$ acting on the copy of $\R^m$  placed at $x \in \ZZ^d$. 
In particular we will consider families of first order operators $D_x, Y_{\alpha, x}$, $x \in \mathbb{Z}^d$ and $\alpha \in I$ for some finite index set $I$, which are isomorphic copies of operators at the origin $x_0 \equiv 0$.  In other words, $D$ and $\{Y_{\alpha}\}_{\alpha \in I}$ are first order operators on $\R^m$; $D_x$ and $\{Y_{\alpha, x}\}_{\alpha \in I}$ are, 
for every $x \in \ZZ^d$, copies of $D$ and $\{Y_{\alpha}\}_{\alpha \in I}$, respectively,  acting on the copy of $\R^m$ placed at $x \in \ZZ^d$. 

We will assume the following commutation relations 
\begin{assumption}[\textbf{GCR}] For any $x,y\in\ZZ^d$ we have:
\begin{itemize}
 \item If $x \neq y$, then 
 $$[Y_{\alpha ,x} , Y_{\beta,y} ]=
[Y_{\alpha,x},D_y]=0,  \qquad  \mbox{for any } \alpha, \beta \in I; $$ 
\item For every $\alpha \in I$, and $x \in \mathbb{Z}^d$
 \begin{align*}
  &
  [Y_{\alpha,x} ,D_x]
  ={\kappa}_{\alpha} Y_{\alpha,x}, 
   {\hspace{2em}} {\kappa}_{\alpha} \geq 0\\
  & [Y_{\alpha,x} ,Y_{\beta ,x} ]
  =\sum_{\gamma \in I} c_{\alpha \beta \gamma} Y_{\gamma,x}\, , 
\end{align*}
with some real constants $c_{\alpha \beta \gamma}$. 
\end{itemize}
We will denote 
$c\equiv \sup_{\alpha,\beta,\gamma \in I}\lv c_{\alpha\beta\gamma} \rv.$
\end{assumption}
\begin{remark}\label{rem3.1}\textup{
We remark that in general,
if the constants $c_{\alpha \beta \gamma} \neq 0$, a compatibility condition
(coming from Jacobi identity) may force all $\kappa_{\alpha} = 0$. The case
when $\kappa_{\alpha} > 0$ for all $\alpha$ will be called {\em stratified case}.}
\end{remark}
Later it will be convenient to use the following notation for operators of
order $n \in \mathbb{N}$:
\[ \label{Yops} \mathbf{Y}_{{\boldsymbol{\iota}}, \mathbf{x}}^{(n)}
   \equiv Y_{\iota_1, x_1} \dots Y_{\iota_n, x_n} \]
where  $\mathbf{x} \equiv (x_1, .., x_n)$, $x_i \in \mathbb{Z}^d$, i.e.  ${\bf x} \subset \ZZ^d$ is a subset of  $\ZZ^d$ of cardinality $n$. Also, we denote
\[ | \mathbf{Y}_{\mathbf{x}}^{(n)} f|^2 \equiv
   \sum_{{\boldsymbol\iota}}|
   \mathbf{Y}_{{\boldsymbol{\iota}}, \mathbf{x}}^{(n)} f|^2 \equiv
   \sum_{\iota_1, ..., \iota_n \in I} |Y_{\iota_1, x_1} ..Y_{\iota_n, x_n}
   f|^2. \]
For some $J \subset I$ (arbitrary but fixed) we set
\[ \label{Y0sqr} Y_{J, x}^2 \equiv \sum_{\alpha \in J} Y_{\alpha, x}^2 \]
and
\[ |Y_{J, x} f|^2 \equiv \sum_{\alpha \in J} |Y_{\alpha, x} f|^2 . \]
For $x \in
\mathbb{Z}^d$, if $\mathbf{q}_x \equiv \{q_{\iota, x} \}_{\iota \in I}$ is a collection of real valued functions (more details about these functions are given below), we set
\[ \mathbf{q}_x \cdot Y_x \equiv \sum_{\iota \in I} q_{\iota, x} Y_{\iota, x}. \]
Analogously, for  $\mathfrak{S}_{xy} \equiv \{ \mathfrak{S}_{\alpha \beta, xy}
\}_{\alpha \beta \in J}$, we introduce
\[ \mathfrak{S}_{xy} \cdot Y_x Y_y \equiv \sum_{\alpha, \beta \in J}
   \mathfrak{S}_{\alpha \beta, xy} \cdot Y_{\alpha, x} Y_{\beta, y} \]
and write
\[ \mathfrak{S}_{xy} \cdot (Y_x f)  (Y_y g) \equiv \sum_{\alpha, \beta \in
   J} \mathfrak{S}_{\alpha \beta, xy} \cdot (Y_{\alpha, x} f) \cdot
   (Y_{\beta, y} g). \]
For every $\boldsymbol{\gamma}\subset \boldsymbol{\iota}\subset I$  and $\z \subset \x \subset \ZZ^d$ we will also use the notation $\check{\ensuremath{\boldsymbol{\gamma}}} \equiv
{\boldsymbol\iota} \setminus \ensuremath{\boldsymbol{\gamma}}$
and $\check{\mathbf{z}} \equiv \mathbf{x} \setminus \mathbf{z}$.\\   
Both $\mathbf{q}_x$ and $\mathfrak{S}_{xy}$ will be assumed to be smooth
functions (of $\omega$), which  can depend on $\omega_y$, $y \neq
x$; all entries of these ``matrices" are assumed to be real valued cylinder functions, so each of the  $\mathbf{q}_x$ and $\mathfrak{S}_{xy}$  only depend on a finite number of coordinates in $\ZZ^d$. It is also
assumed that $\mathfrak{S}_{xy} < 2 \delta_{xy}$ in the sense of quadratic forms.  To stress the cardinality of $\mathbf{x}=(x_1 \dd x_n) \subset \ZZ^d$ as a subset of $\ZZ^d$, we write $\lv \mathbf{x}\rv=n$ (same thing for $\mathbf{\iota}=(\iota_1 \dd \iota_{\ell}) \in I^{\ell} $, we write $\lv \mathbf{\iota} \rv = \ell$.) A
number of additional technical conditions, necessary for development of
nontrivial infinite dimensional theory, will be provided later.\\
For a finite set $\Lambda \subset \subset \mathbb{Z}^d$ we consider the
following Markov generator
\begin{equation}\label{genlambda}
 \mathcal{L}_{\Lambda} = \sum_{x \in \mathbb{Z}^d} L_x +
   \sum_{y \in \Lambda}  \mathbf{q}_y \cdot Y_y + \sum_{y, y' \in \Lambda} 
   \mathfrak{S}_{yy'} \cdot Y_y Y_{y'} 
\end{equation}
where
\[ L_x \equiv Y^2_{J, x} + B_x - \lambda D_x \]
with some constant $\lambda \geq 0$ and
\[ B_x \equiv \sum_{\alpha \in I} b_{\alpha, x} Y_{\alpha, x} \equiv \mathbf{b}_x
   \cdot Y_x , \]
with 
$\mathbf{b}_x \equiv \{ b_{\alpha, x} \in \mathbb{R}\}_{\alpha \in I}$.    We will refer to $\mathbf{q}_x$ and $\mathfrak{S}_{x,y}$ as to {\em interaction functions.} When such functions satisfy the following two conditions
\begin{align}
  Y_{{\alpha},y}\mathbf{q}_x &
  {\equiv}0 {\hspace{2em}}if{\hspace{1em}}dist(y,x){\geq}R \label{fr1}\\
  \mathfrak{S}_{{\gamma}{\gamma}',yy'} &
  {\equiv}0 {\hspace{2em}}if{\hspace{1em}}dist(y,y'){\geq}R \label{fr2}
\end{align}
for some $R>0$, we talk about {\em finite range interaction}. 

We remark that in case of finite range interaction   the operators
$$L_{\Lambda_R}\equiv\sum_{ x \in\Lambda_R } L_x +
   \sum_{y \in \Lambda_R}  \mathbf{q}_y \cdot Y_y + \sum_{y, y' \in \Lambda_R} 
   \mathfrak{S}_{yy'} \cdot Y_y Y_{y'} 
$$
defined with $\Lambda_R\equiv\{x:d(x, \Lambda)\leq R\}$
and 
$$
L_{\Lambda_R^c}\equiv\sum_{ d(x, \Lambda) > R } L_x
$$
commute; therefore
the semigroup generated by $\mathcal{L}_{\Lambda}=L_{\Lambda_R}+L_{\Lambda_R^c}$ is well defined as  product semigroup, denoted  by 
$\PP^{\Lambda}_t \equiv e^{t\mathcal{L}_{\Lambda}}\equiv e^{t L_{\Lambda_R}} e^{t L_{\Lambda_R^c}}$; the first factor here $e^{t L_{\Lambda_R}}$ acts on finite dimensions and is well defined by the standard finite dimensional analysis (see e.g. \cite{CSCV} and references there in), while the second 
$e^{t L_{\Lambda_R^c}}\equiv \prod_{x\in\Lambda_R^c}e^{t L_x}$ 
is an infinite product (of commuting semigroups, each acting on finite dimensions).  In other words,   one way of intuitively understanding the dynamics generated by $\mathcal{L}_{\Lambda}$ \eqref{genlambda} is the following: each of the operators $L_x$ is an hypoelliptic diffusion of the type studied in Section \ref{S.2:Short and long time behaviour} taking place in the copy of $\R^m$ placed at $x \in \ZZ^d$. If the last two addends in  the definition \eqref{genlambda} of $\mathcal{L}_{\Lambda}$   were identically zero, then the dynamics generated by $\mathcal{L}_{\Lambda}$ would simply consist of infinitely many copies of the same hypoelliptic diffusion evolving independently of each other. The last two addends in  \eqref{genlambda} make such diffusions interact.  However, because $\Lambda$ contains only a finite number of points in $\ZZ^d$ (and the interaction functions will always assumed to have finite range), only finitely many of such diffusions  interact ``directly"  under the action of $\mathcal{L}_{\Lambda}$.  
The main purpose of this  section is to show that, in the limit $\Lambda \rightarrow \ZZ^d$,  the semigroup generated  on  $(\R^m)^{\ZZ^d}$ by the operator formally given by
\begin{equation}
 \mathcal{L} = \sum_{x \in \mathbb{Z}^d} L_x +
\sum_{y \in \mathbb{Z}^d}  \mathbf{q}_y \cdot Y_y + \sum_{y, y' \in \mathbb{Z}^d}
   \mathfrak{S}_{yy'} \cdot Y_y Y_{y'} 
\end{equation}
is well posed. In order to achieve this result, some further technical assumptions on the interaction functions will be necessary, see statement of Theorem \ref{existencethm}; some of these assumptions are purely technical.   In order to explain the structure of the remainder of the section and clarify the approach used to construct the infinite dimensional semigroup, we add the following Subsection \ref{subsubsec3}, which should hopefully serve as a navigational chart through the technical results of Section 3. 
\subsubsection{Structure of Section 3}\label{subsubsec3}
In  this subsection we explain the strategy that we are going to use to construct the infinite dimensional semigroup, in its simplest version. In order to do so, we work in a simplified scenario.   The details of the general strategy illustrated in this remark need  technical modifications in our setting, but the bulk of the approach  remains analogous.  
\begin{itemize}
\item
Only for the purpose of this subsection, consider the operator
$$
\mathcal{L}_{\Lambda} = \sum_{x \in \mathbb{Z}^d} L_x +
\sum_{y \in \Lambda}  \mathbf{q}_y \cdot Y_y= 
\sum_{x \in \mathbb{Z}^d} L_x +
\sum_{y \in \Lambda} \sum_{i \in I} q_{i,y}  Y_{i,y}\, , 
$$
generating the semigroup $\PP_t^{\Lambda}$. 
 We want to show 
$$
\lim_{\Lambda \rightarrow \ZZ^d} \PP_t^{\Lambda}f(x)= \PP_t f(x),
$$
for every cylinder function $f$. 
We consider two sets $\bar{\Lambda},\bar{\Lambda}'\subset \ZZ^d$  such that 
 $\Lambda(f)\subset \bar{\Lambda} \subset \bar{\Lambda}'$  and  construct an increasing sequence of sets. Here for simplicity we take  $\{\Lambda_{\mathfrak{m}}\}_{ 0\leq{\mathfrak{m}}\leq \mathcal{N}}$, such that $\Lambda(f)\subset \Lambda_0=\bar{\La}$, $\La_{\mathcal{N}}=\bar{\Lambda}'$ and $\Lambda_{{\mathfrak{m}}+1}\setminus \Lambda _{\mathfrak{m}}=\{h_{\mathfrak{m}}\}$, i.e. $\Lambda_{{\mathfrak{m}}+1}$ is obtained from $\Lambda_{\mathfrak{m}}$ by adding the singleton $h_{\mathfrak{m}}$. 
We denote by $\mathcal{L}_{\Lambda_{\mathfrak{m}}}$ the Markov generator
$$
\mathcal{L}_{\Lambda_{\mathfrak{m}}}\equiv \sum_{x \in \mathbb{Z}^d} L_x +
\sum_{y \in \Lambda_{\mathfrak{m}}} \sum_{i \in I} q_{i,y}  Y_{i,y}
$$
and $\PP_t^{\Lambda_{\mathfrak{m}}}$ the corresponding semigroup. If we  show that the sequence $\{\PP_t^{\Lambda_{\mathfrak{m}}}f(x)\}$ is a Cauchy sequence then we are done. 
From the identity
\be
\PP_t^{\Lambda_{{\mathfrak{m}}+1}}f-\PP_t^{\Lambda_{\mathfrak{m}}}f=
\int_0^t ds  \frac{d}{ds}\left( 
\PP_{t-s}^{\Lambda_{\mathfrak{m}}}\PP_s^{\Lambda_{{\mathfrak{m}}+1}}f
\right)
=\int_0^t ds \left[\PP_{t-s}^{\Lambda_{\mathfrak{m}}} (\cl_{\Lambda_{{\mathfrak{m}}+1}}-\cl_{\Lambda_{\mathfrak{m}}})
\PP_s^{\Lambda_{{\mathfrak{m}}+1}}f\right], 
\ee
we  have \begin{align}
& \|\PP_t^{\bar{\Lambda}}f-\PP_t^{\bar{\Lambda}'}f\|_{\infty}\leq 
\sum_{\mathfrak{m}=0}^{\mathcal{N}-1}
\|\PP_t^{\Lambda_{{\mathfrak{m}}+1}}f-
\PP_t^{\La_{\mathfrak{m}}}f\|_{\infty}\nonumber\\ 
&\leq \sum_{\mathfrak{m}=0}^{\mathcal{N}-1}
\int_0^t ds \|\PP_{t-s}^{\Lambda_{\mathfrak{m}}} (\cl_{\La_{{\mathfrak{m}}+1}}-\cl_{\Lambda_{\mathfrak{m}}})
\PP_s^{\Lambda_{{\mathfrak{m}}+1}}f\|_{\infty}\nonumber\\
&\leq \sum_{\mathfrak{m}=0}^{\mathcal{N}-1} \int_0^t ds 
\sum_{i \in I} \|q_{i,h_{\mathfrak{m}} }Y_{i,h_{\mathfrak{m}}}
f_s^{\Lambda_{{\mathfrak{m}}+1}}\|_{\infty}\label{Cauchy2}
\end{align}
The above is a simplified version of the calculation in  the proof of Theorem \ref{existencethm} - in that setting also second derivatives of $f_s^{\Lambda_{{\mathfrak{m}}+1}}$ would appear in the last step, and this is one of the reasons why in general one cannot choose the simple sequence of increasing sets that we are choosing here. In any event, what is important to notice is that in \eqref{Cauchy2} appears the derivative of $f_s^{\Lambda_{{\mathfrak{m}}+1}}$ at $h_{\mathfrak{m}}$. This brings us to the next point. 
\item Recall that in the above we fixed a cylinder function $f$, localised in $\Lambda(f)$. From the construction in the previous point, $h_{\mathfrak{m}} \notin \Lambda(f)$. Hence the need to find estimates on the derivatives of $Y_{i, x}\PP_t^{\Lambda}$ at a point $x \in \ZZ^d$ which is out of $\Lambda(f)$. This is precisely the kind of estimates that we recover in Theorem \ref{fspthm}. In order to study the well posedness of the infinite dimensional semigroup we would need, in our case, only first and second order derivatives. We find the estimates for derivatives of any order (Lemma \ref{lem.n}) as they will be needed in Section 5.
\item Finally, once the infinite dimensional semigroup is obtained we prove, for such a semigroup, smoothing results similar to those shown to hold in Section \ref{S.2:Short and long time behaviour} for the finite dimensional case. Such results will be used to study the ergodicity of the dynamics.
\end{itemize}

\subsection{Strong approximation property}\label{sec:strsec}

We begin  by stating the preliminary result of Proposition \ref{pro3.1}. In the statement of Proposition \ref{pro3.1} the following notation will be used: if $\mathbf{z}=(z_1 \dd z_l) \subset \mathbb{Z}^d$ and $y \in \mathbb{Z}^d$, then
\begin{equation}\label{gstar}
\lv \mathbf{Y}_{(\mathbf{z},y)}^{(l+1)} f \rv^2 \equiv
\sum_{\iota_1 \dd \iota_{l+1} \in I} \lv Y_{\iota_1, z_1}
   Y_{\iota_2, z_2} \dd Y_{\iota_l, z_l}, Y_{\iota_{l+1}, y}  f  \rv^2\, .
\end{equation}

\begin{prop} \label{pro3.1}
Let $\mathcal{L}_{\Lambda}$ be the generator \eqref{genlambda} and suppose that the commutator relations of Assumption {\upshape(\textbf{GCR})} hold for each of the $L_x$. 
Moreover, assume the interaction functions are such that
\begin{align*}
& i) \sup_{\alpha, z} || q_{\alpha, z} ||_{\infty} <\infty\\
& ii) \mathfrak{S}_{\gamma \gamma', yy'} \equiv
   \delta_{y \neq y'}  \mathfrak{S}_{\gamma \gamma', yy'} (\omega_y,
   \omega_{y'})\\
& iii) \sup_{z \in \mathbb{Z}^d}  \sum_{y \in \mathbb{Z}^d} 
   \sum_{\gamma \gamma' \in J} (| \mathfrak{S}_{\gamma' \gamma, yz} | +
   | \mathfrak{S}_{\gamma \gamma', zy} |) < \infty\\
& iv)\sup_{(\mathbf{\iota}, \mathbf{x}) : | {\boldsymbol\iota} | = n}
   \sum_{y \in \mathbb{Z}^d}  \sum_{\gamma, \gamma' \in J}  \sum_{l = 1}^{n -
   1} \sum_{{({\boldsymbol\beta}, \mathbf{z}) \subset (\boldsymbol\iota,
   \mathbf{x})}\atop{| {\boldsymbol\beta} | = l}} \sum_{y' \in \check{\mathbf{z}}} 
   \left|  \mathbf{Y}_{\check{{\boldsymbol\beta}}, \check{\mathbf{z}}}^{(n - l)}
   \mathfrak{S}_{\gamma \gamma', yy'}  \right| \hspace{0.25em} <
   \hspace{0.25em} \infty\\
& v) \sum_{y \in \mathbb{Z}^d} \sum_{\beta \in I}  \sum_{k =
   1}^{n - 1} \sup_{({\boldsymbol\iota}, \mathbf{x})} \sum_{({\boldsymbol\gamma},
   \mathbf{z}) \subset ( {\boldsymbol\iota}, \mathbf{x}) : | {\boldsymbol\gamma} |
   = k} \| \mathbf{Y}_{\check{{\boldsymbol\gamma}}, \check{\mathbf{z}}}^{(n -
   k)}  q_{\beta, y} \|_{\infty} < \infty\,.
\end{align*}
Then for any $\Lambda \subset \mathbb{Z}^d$, for any cylinder function $f$ with $\Lambda(f)\subset \Lambda$ and  for any  ${\bf x}=(x_1 \dd x_n) \subset \ZZ^d$ we have 
\begin{align}
{\frac{\partial}{\partial s}}\mathcal{P}_{t-s}^{\Lambda}\left|\mathbf{Y}_{\mathbf{x}}^{n}f_s^{\Lambda}\right|^2
  &{\leq}\mathcal{P}_{t-s}^{\Lambda}\left\{\mathbf{v}_n{\hspace{0.25em}}|\mathbf{Y}_{\mathbf{x}}^{n}f_s^{\Lambda}|^2+\sum_{l=1}^{n-1}\sum_{{\mathbf{z}{\subset}\mathbf{x}}\atop{|\mathbf{z}|=l}}
  \mathcal{B}_{\mathbf{x},n}^{(l)}(\mathbf{z}){\hspace{0.25em}}\left|
  \mathbf{Y}_{\mathbf{z}}^{(l)}f_s^{\Lambda}\right|^2\right. \nonumber\\
  &
  \left.+{\varepsilon} \sum_{l=1}^{n-1}\sum_{y\in\Lambda}\sum_{\mathbf{z}{\subset}\mathbf{x}:|\mathbf{z}|=l}\mathcal{A}_{\mathbf{x},n}^{(l)}(\mathbf{z},y)\left|Y_{J,y}\mathbf{Y}_{\mathbf{z}}^{(l)}f_s^{\Lambda}\right|^2+{\varepsilon} \sum_{l=0}^{n-1}\sum_{{\mathbf{z}{\subset}\mathbf{x},y{\in}{\Lambda}}\atop{|\mathbf{z}|=l}} B_{\mathbf{x},n}^{(l)}(\mathbf{z},y){\hspace{0.25em}}|\mathbf{Y}_{(\mathbf{z},y)}^{(l+1)}
f_s^{\Lambda}|^2\right\} \label{statpro3.1}
\end{align}
for some constants $\varepsilon \in (0, 1)$, $\mathcal{B}_{\mathbf{x},n}^{(l)}(\mathbf{z}), \mathcal{A}_{\mathbf{x},n}^{(l)}(\mathbf{z},y), B_{\mathbf{x},n}^{(l)}(\mathbf{z},y)>0$ and $\mathbf{v}_n$, 
independent of $f$ and $t$. 
\end{prop}
\begin{remark}\label{remfri}\textup{
We refrain from writing here a full expression of the constants $\mathcal{B}_{\mathbf{x},n}^{(l)}(\mathbf{z}), \mathcal{A}_{\mathbf{x},n}^{(l)}(\mathbf{z},y), B_{\mathbf{x},n}^{(l)}(\mathbf{z},y)>0$ (however such expressions can be found in the proof of Proposition \ref{pro3.1}). What is important for our purposes is that, in case of finite range interaction (see \eqref{fr1}-\eqref{fr2}), such coefficients vanish unless $diam (\mathbf{x}
\setminus \mathbf{z}) \leq R$.
}
\end{remark}
\begin{proof}[Proof of Proposition \ref{pro3.1}]
The proof of this proposition is deferred to Appendix A.
\end{proof}
%

Integrating the differential inequality \eqref{statpro3.1} gives
\begin{align*}
  \left|\mathbf{Y}_{\mathbf{x}}^{(n)}f_{t}^{\Lambda}\right|^2
  & \leq e^{\mathbf{v}_{n}t}\mathcal{P}_{t}^{\Lambda}|\mathbf{Y}_{\mathbf{x}}^{(n)}f|^2
  +\sum_{l=1}^{n-1}\sum_{{\mathbf{z}{\subset}\mathbf{x}}\atop{|\mathbf{z}|=l}} \mathcal{B}_{\mathbf{x},n}^{(l)}(\mathbf{z}) {\hspace{0.25em}} 
  \int_{0}^{t} ds\, e^{\mathbf{v}_{n}(t-s)} \mathcal{P}_{t-s}^{\Lambda} \left|\mathbf{Y}_{\mathbf{z}}^{(l)} f_s^{\Lambda}\right|^2\\
  & + {\varepsilon} \sum_{l=1}^{n-1}\sum_{y\in\Lambda}\sum_{\mathbf{z}{\subset}\mathbf{x}:|\mathbf{z}|=l} \mathcal{A}_{\mathbf{x},n}^{(l)}(\mathbf{z},y) \int_{0}^{t}d s\, e^{\mathbf{v}_{n}(t-s)}\mathcal{P}_{t-s}^{\Lambda}\left|Y_{J,y}\mathbf{Y}_{\mathbf{z}}^{(l)}f_s^{\Lambda}\right|^2\\
  &+{\varepsilon} \sum_{l=0}^{n-1}
  \sum_{{\mathbf{z}{\subset}\mathbf{x},y{\in}{\Lambda}}\atop{|\mathbf{z}|=l}} B_{\mathbf{x},n}^{(l)}(\mathbf{z},y) {\hspace{0.25em}}\int_{0}^{t}d s\, 
e^{\mathbf{v}_{n}(t-s)}\mathcal{P}_{t-s}^{\Lambda}|\mathbf{Y}_{(\mathbf{z},y)}^{(l+1)}f_s^{\Lambda}|^2.
\end{align*}
Taking the supremum norm, the above bound can be simplified as follows
\begin{align}
  \|\mathbf{Y}_{\mathbf{x}}^{(n)}f_{t}^{\Lambda}\|^2_{\infty}
  & \leq e^{\mathbf{v}_{n}t}\|\mathbf{Y}_{\mathbf{x}}^{(n)}f\|_{\infty}^2+\sum_{l=1}^{n-1}\sum_{{\mathbf{z}{\subset}\mathbf{x}}\atop{|\mathbf{z}|=l}}\mathcal{B}_{\mathbf{x},n}^{(l)}(\mathbf{z}) {\hspace{0.25em}}
  \int_{0}^{t}ds\, e^{\mathbf{v}_{n}(t-s)}\|\mathbf{Y}_{\mathbf{z}}^{(l)}f_s^{\Lambda}\|^2_{\infty} \nonumber\\
  & +{\varepsilon} \sum_{l=1}^{n-1}\sum_{y\in\Lambda}\sum_{\mathbf{z}{\subset}\mathbf{x}:|\mathbf{z}|=l}\mathcal{A}_{\mathbf{x},n}^{(l)}(\mathbf{z},y)\int_{0}^{t}ds\, e^{\mathbf{v}_{n}(t-s)}\|Y_{J,y}\mathbf{Y}_{\mathbf{z}}^{(l)}f_s^{\Lambda}\|^2_{\infty} \nonumber \\
  & +{\varepsilon} \sum_{l=0}^{n-1}
  \sum_{{\mathbf{z}{\subset}\mathbf{x},y{\in}{\Lambda}}\atop{|\mathbf{z}|=l}} B_{\mathbf{x},n}^{(l)}(\mathbf{z},y) {\hspace{0.25em}} \int_{0}^{t}ds\,e^{\mathbf{v}_{n}
(t-s)} 
  \|\mathbf{Y}_{(\mathbf{z},y)}^{(l+1)}f_s^{\Lambda}\|^2_{\infty}, \label{98}
\end{align}
where we have used the  contractivity property of the Markov semigroup with respect
to the supremum norm.  
The norm in the first term on the right hand side
does not depend on time and is zero if $\mathbf{x} \cap \Lambda (f) \equiv
\{x_1,..,x_n\}\cap \Lambda (f)=
\emptyset $; the second sum involves lower order terms and integration with
respect to time (and it may be empty if $n = 1$); the third sum involves
integration with respect to time and differentiations at sites which are not
in $\mathbf{x}$ and are performed in mild directions (from principal part of
the generator with indices from $J$), but the order can be up to $n$; the last
is of similar nature as the third,  except that all directions are
involved.

In the case where the interaction is of finite range,
one can simplify  expression \eqref{98} considerably by using Remark \ref{remfri}. Indeed in this case  we can replace all the constants on the RHS of \eqref{98}   by their supremum
$C_0$ and restrict the summation over $y$ by a condition $d(y, \mathbf{x})
\leq R$. For the rest of the paper we set
$$\| \cdot \|\equiv \| \cdot \|_{\infty}.$$ 
Then we have the following result.
\begin{lemma}\label{lem.n}  Under the assumptions of Proposition \ref{pro3.1}, for all $\x=(x_1 \dd x_n) \subset \ZZ^d$, if  \eqref{fr1} and \eqref{fr2} hold, then 
  \begin{align*}
    \|\mathbf{Y}_{\mathbf{x}}^{(n)}f_{t}^{\Lambda}\|^2
    &
    \leq e^{\mathbf{v}_{n}t}\|\mathbf{Y}_{\mathbf{x}}^{(n)}f\|^2+C_{0} \sum_{l=1}^{n-1}\sum_{{\mathbf{z}{\subset}\mathbf{x}}\atop{|\mathbf{z}|=l,diam(\mathbf{x}{\setminus}\mathbf{z}){\leq}R}} {\hspace{0.25em}}\int_{0}^{t} ds\, e^{\mathbf{v}_{n}(t-s)}\|\mathbf{Y}_{\mathbf{z}}^{(l)}f_s^{\Lambda}\|^2\\
    &
    +C_{0} \sum_{l=1}^{n-1}\sum_{{y{\in}{\Lambda}}\atop{d(y,\mathbf{x}){\leq}R}} \sum_{{\mathbf{z}{\subset}\mathbf{x}:}\atop{|\mathbf{z}|=l,diam(\mathbf{x}{\setminus}\mathbf{z}){\leq}R}} \int_{0}^{t}ds\, e^{\mathbf{v}_{n}(t-s)}\|Y_{J,y}\mathbf{Y}_{\mathbf{z}}^{(l)}f_s^{\Lambda}\|^2\\
    &
    +C_{0} \sum_{l=0}^{n-1}\sum_{{\mathbf{z}{\subset}\mathbf{x},y{\in}{\Lambda}}\atop{|\mathbf{z}|=l,d(y,\mathbf{x}){\leq}R}} {\hspace{0.25em}}\int_{0}^{t}ds\, e^{\mathbf{v}_{n}(t-s)} \|\mathbf{Y}_{(\mathbf{z},y)}^{(l+1)} f_s^{\Lambda}\|^2 . 
  \end{align*}
\end{lemma}
The special cases $n = 1, 2$ will be immediately relevant for the construction
of the limit of the semigroups $P_t^{\Lambda}$ as $\Lambda \to \mathbb{Z}^d$,  so we state such cases explicitly in the next Lemma \ref{lem.n12}. 

\begin{lemma}\label{lem.n12}
Under the assumptions of Lemma \ref{lem.n}, 
for $n = 1$, we have
  \be \label{3.1.37} 
  \|Y_x f_t^{\Lambda} \|^2 \leq e^{\mathbf{v}_1 t} \|Y_x f\|^2 
     + C_0  \sum_{y\in\Lambda \atop dist(y,x)\leq R}   \hspace{0.25em} 
     \int_0^t ds\, e^{\mathbf{v}_1 (t - s)} \|Y_y
     f_s^{\Lambda} \|^2 
     \ee
  For $n = 2$, with some $C_1, C_2 \in (0, \infty)$ dependent only on $C_0$ and
  $c_{\alpha, \beta, \gamma}$, we have
  \begin{align*}
    \|\mathbf{Y}_{\mathbf{x}}^{(2)}f_{t}^{\Lambda}\|^2
    &
    \leq e^{\mathbf{v}_{2}t}\|\mathbf{Y}_{\mathbf{x}}^{(2)}f\|^2+C_{0}\sum_{z\in\mathbf{x}} {\hspace{0.25em}}\int_{0}^{t}ds\, e^{\mathbf{v}_{2}(t-s)} \|Y_{z} f_s^{\Lambda}\|^2\\
    & +C_1 \sum_{{y\in\Lambda}\atop{d(y,\mathbf{x}){\leq}R}} {\hspace{0.25em}}\int_{0}^{t}ds\, e^{\mathbf{v}_{2}(t-s)}\|Y_y f_s^{\Lambda}\|^2\\
    & + C_{2} \sum_{{z{\in}\mathbf{x},y{\in}{\Lambda}}\atop{d(y,\mathbf{x}){\leq}R}} {\hspace{0.25em}} 
    \int_{0}^{t}ds\,
     e^{\mathbf{v}_{2}(t-s)}\|\mathbf{Y}_{(y,z)}^{(2)}f_s^{\Lambda}\|^2.  \end{align*}
The constants $\mathbf{v}_{1}$ and $\mathbf{v}_{2}$ in  the above are as in the statement of Proposition \ref{pro3.1}.
\end{lemma}
Using the above lemmata, we prove the following result.
\begin{theorem}[Finite speed of propagation of information]\label{fspthm}
Suppose the assumptions of Lemma  \ref{lem.n} hold. Then  for any smooth cylinder function $f$ with $\Lambda (f) \subset \Lambda$ and
  for any $n \in \mathbb{N}$, there exist constants $B, c, v \in (0, \infty)$,
  independent of $f$ but possibly dependent on $n$, such that for all 
  $\x=(x_1\dd x_n)\subset \ZZ^d$, 
  \begin{equation}\label{FSP} 
  \| \mathbf{Y}_{\mathbf{x}}^{(n)} f_t^{\Lambda} \|^2 \leq
    Be^{ct - v \cdot d (\mathbf{x}, \Lambda (f))}  \sum_{l = 1, .., n} \|
    \mathbf{Y}^{(l)} f\|^2
  \end{equation}
  where
\begin{equation}\label{gstargstar}  
\| \mathbf{Y}^{(l)} f\|^2 \equiv \sum_{\mathbf{z} : | \mathbf{z} | = l}
     \| \mathbf{Y}^{(l)}_{\mathbf{z}} f\|^2 
\end{equation}
  and we recall that  $d (\mathbf{x}, \Lambda (f))$ denotes the length of the shortest tree
  connecting each component of $\mathbf{x}$ and $\Lambda (f)$.
\end{theorem}
\begin{proof}[Proof of Theorem \ref{fspthm}]The case $n = 1$ is well known, see e.g. \cite{DKZ2011,GZ} and references there in. The estimate is essentially based on
inductive use of the Gronwall type inequality (\ref{3.1.37}) using the fact
that the first term on its right hand side is zero unless $x \in \Lambda (f)$,
so if you start from $d (x, \lambda (f)) \geq NR$ to get a nonzero term you
need to make at least $N$ steps producing multiple integral of that order
which is responsible for a factor of the form $e^{Ct} (N!)^{- 1}$ with $C \leq
\mathbf{v}_1 + C_0  (2 R)^d$ (for more details see \cite{DKZ2011,GZ}). 

For $n = 2$, using \eqref{FSP} with $n=1$, we get 
\begin{align*}
  \|\mathbf{Y}_{\mathbf{x}}^{(2)}f_{t}^{{\Lambda}}\|^{2}
  &
  \leq e^{\mathbf{v}_{2}t}\|\mathbf{Y}_{\mathbf{x}}^{(2)}f\|^{2}+B_{1} e^{\bar{c}_{1} t-\bar{v}_{1} d(z,{\Lambda}(f))} \sum_{y{\in}\mathbf{x}} {\hspace{0.25em}} \|Y_{y} f\|^{2}\\
  & +2 C_{0} \sum_{{z{\in}\mathbf{x},y{\in}{\Lambda}}\atop{d(y,\mathbf{x}){\leq}R}} {\hspace{0.25em}}\int_{0}^{t}ds\, e^{\mathbf{v}_{2}(t-s)}\|\mathbf{Y}_{(y,z)}^{(2)}f_{s}^{{\Lambda}}\|^{2}, \nonumber 
\end{align*}
for some constants $\bar{c}_1,\bar{v}_1 \in (0, \infty)$. 
We will use this relation inductively taking into the account that as long as
$\mathbf{z} \nsubseteq \Lambda (f)$, we have $\| \mathbf{Y}_{\mathbf{z}}^{(2)}
f\|^2 = 0$. Thus the first term on the right hand side will not give nonzero
contribution until we apply our procedure at least $N \equiv d (\mathbf{x},
\Lambda (f)) / (2 R)$ times, but to reach that we will produce multiple
integral of order $N$ giving a factor $(N!)^{- 1}$. This implies the following
bound
\begin{equation*}
  \| \mathbf{Y}_{\mathbf{x}}^{(2)} f_t^{\Lambda} \|^2 \leq B_2
  e^{c_2 t - v_2 d (\mathbf{x}, \Lambda (f))}  \left( \sum_{\mathbf{z}} \|
  \mathbf{Y}_{\mathbf{z}}^{(2)} f\|^2 + \sum_y \hspace{0.25em} \|Y_y f\|^2
  \right).
\end{equation*}
The general case is proved by induction with respect to $n$. We suppose that
\[ 
\mbox{for all }  1 \leq k \leq n - 1 \mbox{ and } \mathbf{z} : |
   \mathbf{z} | = k, \hspace{2em} \| \mathbf{Y}_{\mathbf{z}}^{(k)}
   f_t^{\Lambda} \|^2 \leq B_k  \sum_{z \in \mathbf{z}} e^{c_k t - v_k d (z,
   \Lambda (f))}  \sum_{l = 1, .., n} \| \mathbf{Y}^{(l)} f\|^2. 
\]
Then, using Lemma \ref{lem.n}, we get
\begin{align*}
  \|\mathbf{Y}_{\mathbf{x}}^{(n)}f_{t}^{{\Lambda}}\|^{2}
  &
  \leq e^{\mathbf{v}_{n}t}\|\mathbf{Y}_{\mathbf{x}}^{(n)}f\|^{2}+C_{0}\sum_{l=1}^{n-1}\sum_{{\mathbf{z}{\subset}\mathbf{x}}\atop{|\mathbf{z}|=l}}{\hspace{0.25em}}\int_{0}^{t}ds\, e^{\mathbf{v}_{n}(t-s)}\|\mathbf{Y}_{\mathbf{z}}^{(l)}f_{s}^{{\Lambda}}\|^{2}\\
  & +C_{0} \sum_{l=1}^{n-1}\sum_{{\mathbf{z}{\subset}\mathbf{x},y{\in}{\Lambda}}
  \atop{|\mathbf{z}|=l,d(y,\mathbf{x}){\leq}R}} {\hspace{0.25em}}\int_{0}^{t}ds\, e^{\mathbf{v}_{n}(t-s)} \int_{0}^{t}ds\, e^{\mathbf{v}_{n}(t-s)}\|Y_{J,y}\mathbf{Y}_{\mathbf{z}}^{(l)}f_{s}^{{\Lambda}}\|^{2}\\
  & +C_{0} \sum_{l=0}^{n-1}\sum_{{\mathbf{z}{\subset}\mathbf{x},y{\in}{\Lambda}}\atop{|\mathbf{z}|=l,d(y,\mathbf{x}){\leq}R}} {\hspace{0.25em}}\int_{0}^{t}ds\, e^{\mathbf{v}_{n}(t-s)}\|\mathbf{Y}_{(\mathbf{z},y)}^{(l+1)}f_{s}^{{\Lambda}}\|^{2}.
\end{align*}
\end{proof}
\subsection{Existence of the infinite dimensional semigroup}
\label{sec:exi}
In this section we prove, through an approximation procedure, that the infinite dimensional semigroup is well posed.  We work under the assumption that the interaction functions are bounded, together with their derivatives of any order. Furthermore, we assume that the interaction is short range and we denote by $R>0$ the range of interaction.  This is the meaning of the assumptions in the following theorem. 
\begin{theorem}\label{existencethm}
  Suppose Assumption {\upshape{(\textbf{GCR}) }} is satisfied and, for every $x \in
  \mathbb{Z^{}}^d$, the fields $\{ Y^{}_{J, x}, B_x \}$ form a
  H{\"o}rmander system. Moreover assume \eqref{fr1} and \eqref {fr2} and conditions i) and ii) of Proposition \ref{pro3.1} hold, together  with
$$ 
\sup_{\ensuremath{\boldsymbol{\gamma}}, \mathbf{z}, \beta,
     y, k} \| \mathbf{Y}_{\ensuremath{\boldsymbol{\gamma}}, \mathbf{z}}^{(k)} 
     q_{\beta, y} \|_{\infty} < \infty, 
\qquad 
     \sup_{\ensuremath{\boldsymbol{\alpha}}, \mathbf{z}, \gamma, \gamma', y,
     y', k} \| \mathbf{Y}_{\ensuremath{\boldsymbol{\alpha}}, \mathbf{z}}^{(k)}
     \mathfrak{S}_{\gamma \gamma', yy'}
     \|_{\infty} < \infty.
$$
    Then, for any continuous compactly supported cylinder function $f,$ the
  following limit exists
  \[ \PP_t f \equiv \lim_{\Lambda \rightarrow  \ZZ^{\ensuremath{\boldsymbol{d}}}}
     \PP^{\Lambda}_t f \]
  and its extension defines a strongly continuous Markov semigroup
  on $\mathcal{C} (\Omega)$. Moreover,  $\PP_t (C (\Omega)) \subset C^{\infty}
  (\Omega)$.  In addition, for any continuous compactly supported function \(f\), for all $n \in \mathbb{N}$ and  all \(\mathbf{x} \in \ZZ^d\) with \(|\mathbf{x}|=n\), we have
  \begin{align*}
  \mathbf{Y}^{(n)}_{\mathbf{x}} \PP_t f = \lim _{\Lambda \rightarrow \ZZ^d}\mathbf{Y}^{(n)}_{\mathbf{x}} \PP^{\Lambda}_t f .
  \end{align*}
\end{theorem}
\proof We consider a lexicographic order $(\{x_k\in\ZZ^d, \preceq\}_{k\in\mathbb{N}})$ on the lattice so that
\[x_k \preceq x_{k+1} \iff \hat{d}(x_k,0) \leq \hat{d}(x_{k+1},0)\]
with $\hat{d}(x,y)\equiv \max_{l=1,..,d}|x^l-y^l|,$ and such that $\Lambda_j \equiv \{ x_i : i \leq
j \}$ is a connected set. For a smooth cylinder function $f$ with bounded
derivatives and $\Lambda (f) \subset \Lambda_j$, we have
\[ 
\left| \PP^{\Lambda_{j + 1}}_t f - \PP^{\Lambda_j}_t f\right |  =\left| \int_0^t
 {ds}\, (\PP^{\Lambda_j}_{t - s}   \left(\mathcal{L}_{\Lambda_{j +
   1}} - \mathcal{L}_{\Lambda_j}\right) \PP^{\Lambda_{j + 1}}_s f)) 
   \right|  
\]
Using the definition of the generators and our finite speed of propagation of
information estimate Theorem \ref{fspthm}, we get (acting analogously to \eqref{Cauchy2} and using the notation \eqref{gstar} and \eqref{gstargstar})
\begin{eqnarray*}
  \left| \PP^{\Lambda_{j}}_t f - \PP^{\Lambda_{j-1}}_t f\right |  &\leq & \int_0^t ds\,  \left(  \left\| q_{x_j} \right\| \cdot \|
  \mathbf{Y}^{(1)}_{x_j} \PP^{\Lambda_{j + 1}}_s f \| \right)  \\
  && +\int_0^t
  ds\,  \left( \sum_{y \in \Lambda_{j + 1}}  ( \|
  \mathfrak{S}_{yx_j} \| \cdot \| \mathbf{Y}^{(2)}_{(y,x_j)}
  \PP^{\Lambda_{j + 1}}_s f \|  +  \| \mathfrak{S}_{x_j y} \| \cdot \|
  \mathbf{Y}^{(2)}_{(x_j, y)} \PP^{\Lambda_{j + 1}}_s f \|) \right)  \\
  &\leq & te^{Ct - v \cdot  d (x_j, \Lambda (f))} \sup_{x \in \ZZ^d} \left\|
  \mathbf{q}_{x} \right\| \cdot \| \mathbf{Y}^{(1)} f \|   \\ &&+   2
  te^{Ct  - v\cdot d(x_j, \Lambda (f))} e^{2 vR} 
  | 2 R  |^d \sup_{x,y \in \ZZ^d} \| \mathfrak{S}_{xy} \| \cdot \| \mathbf{Y}^{(2)} f
  \|,
\end{eqnarray*}
with $\left\| \mathbf{q}_{x_{}} \right\| \equiv \sum_{\alpha} ||
{q}_{\alpha, x} ||_{\infty}$ and $\| \mathfrak{S}_{x_{} y} \| \equiv
\sum_{\gamma \gamma'} \| \mathfrak{S}_{\gamma \gamma', x_{} y} \| _{\infty}$.
Hence for any $\Lambda_k $ and $\Lambda_m$, $k \leq m$, we have
\[ \| \PP^{\Lambda_m}_t f - \PP^{\Lambda_k}_t f \|_{} \leq \sum^{}_{k \leq j
   \leq m} | \PP^{\Lambda_{j + 1}}_t f - \PP^{\Lambda_j}_t f | \leq A_t e^{-
   \frac{v}{2} d (x_k, \Lambda (f))} (\| \mathbf{Y}^{(2)} f \| + \|
   \mathbf{Y}^{(1)} f \|), \]
with a constant
\[ A_t \equiv 2 te^{Ct} B \max \left( \sup_x \left\|
   \mathbf{q}_{x} \right\|,   e^{2 vR} | 2 R |^d
   \sup_{xy} \| \mathfrak{S}_{x_{} y} \| \right) ,\]
where $B \equiv \sum_j e^{- \frac{v}{2} d (x_j, \Lambda (f))}$. Hence, for any \(t>0\), the
sequence $\PP^{\Lambda_j}_t f$, $j \in \mathbb{N}$, is Cauchy in the space of continuous functions equipped with the uniform
norm and there exists a (positivity and unit preserving), densely defined
linear operator $\PP_t$ such that

\[ \| \PP_t f - \PP^{\Lambda_k}_t f \|_{} \leq A_t e^{- \frac{v}{2} d
   (\Lambda^c_k, \Lambda (f))} (\| \mathbf{Y}^{(2)} f \| + \| \mathbf{Y}^{(1)}
   f \|). \]
\[  \]
By the density of smooth cylinder functions and contractivity of $\PP_t $, it can be
extended to all $C (\Omega)$. Using the last estimate one can also show the
semigroup property for $\PP_t$. The fact that $\PP_t (C (\Omega)) \subseteq C^{\infty}  (\Omega)$ follows by H\"ormander's theorem.

Next we consider  sequences of derivatives. For $\mathbf{x} \subset
\Lambda_j$ and \(n= |\mathbf{x}|\),  arguing as above and using the definition of the generators, for a
smooth cylinder function $f$ with bounded derivatives and $\Lambda (f) \subset
\Lambda_j$, we have
\[ \left| \mathbf{Y}^{(n)}_{\mathbf{x}} \PP^{\Lambda_j}_t f
   -\mathbf{Y}^{(n)}_{\mathbf{x}} \PP^{\Lambda_{j - 1}}_t f \right|  = \phantom{AAAAAAAAAAAAAAAAAAAAAAAAAAAAAA} \]
 \[ \phantom{AAAAAAAAAAA} \left| \int_0^t ds\,  \left( \mathbf{Y}^{(n)}_{\mathbf{x}} \PP^{\Lambda_j}_{t
   - s}  \left( \mathbf{q}_{x_j} \cdot \mathbf{Y}^{(1)}_{x_j} + \sum_{y
   \in \Lambda_{j + 1}} (\mathfrak{S}_{yx_j} \cdot
   \mathbf{Y}^{(2)}_{(y,x_j)} + \mathfrak{S}_{x_j y} \cdot
   \mathbf{Y}^{(2)}_{(x_j, y)}) \right) \PP^{\Lambda_j}_s f \right)  
   \right|. \]
Applying Theorem 3.1 to $\mathbf{Y}^{(n)}_{\mathbf{x}}
\PP^{\Lambda_j}_{t - s} F$ with
\[ F \equiv \left( \mathbf{q}_{x_j} \cdot \mathbf{Y}^{(1)}_{x_j} + \sum_{y \in
   \Lambda_{j + 1}} (\mathfrak{S}_{yx_j} \cdot
   \mathbf{Y}^{(2)}_{(y,x_j)} + \mathfrak{S}_{x_j y} \cdot
   \mathbf{Y}^{(2)}_{(x_j, y)}) \right) \PP^{\Lambda_j}_s f, \]
we get the following estimate
\[ \left\| \mathbf{Y}^{(n)}_{\mathbf{x}} \PP^{\Lambda_{j - 1}}_{t - s} F
   \right\|^2 \leq Be^{C (t - s)} \sum_{l = 1}^n \sum_{\mathbf{z}
   \subset \widetilde{\Lambda}_j : \left| \mathbf{z} \right| = l} \left\|
   \mathbf{Y}^{(l)}_{\mathbf{z}} F \right\|^2, \]
with $\widetilde{\Lambda}_j \equiv \{ x \in \mathbb{Z}^d : d (x, \Lambda_j)
\leq R \}$. We note that for the cylinder function $f$, the function $F$ is also
a smooth cylinder function with $\Lambda (F) \equiv \widetilde{\Lambda}_j$.
Thus the sum over $\mathbf{z} \subset \widetilde{\Lambda}_j $ such that $\lv
\mathbf{z} \rv = l$ contains less than $\frac{1}{l!} (| \Lambda_j | + 2 R)^l$
terms. Each of the terms can be bounded as follows
\begin{align*} 
\lv \mathbf{Y}^{(l)}_{\mathbf{z}} F \rv^2 &\leq  D_1 \sum_{k = 1}^l 
\sum_{\left| \mathbf{z}' \right| = k} 
\left\| \mathbf{Y}^{(k+ 1)}_{(\mathbf{z}',x_j)} \PP^{\Lambda_j}_s f \right\|^2\\
&+ D_2 \sum_{d (y, x_j) \leq R}\sum_{k= 1}^l 
\sum_{\left| \mathbf{z}' \right| = k}
   \left( \left\| \mathbf{Y}^{(k + 2)}_{(\mathbf{z}', y,x_j)}
   \PP^{\Lambda_j}_s f \right\|^2 + 
   \left\| \mathbf{Y}^{(k + 2)}_{(\mathbf{z}',x_j, y)} \PP^{\Lambda_j}_s f \right\|^2 \right),  
   \end{align*}
with
\begin{align*} D_1 &\equiv \max_{l = 1, \ldots, n} \sup_{\left\{ x_j \in \mathbb{Z}^d, \left| \mathbf{z} \right| = l \right\}} \sum_{k
   = 1  }^l \sum_{\left| \mathbf{z}' \right| = k} \left\| \mathbf{Y}^{(l -
   k)}_{\mathbf{z} \setminus \mathbf{z}'} \mathbf{q}_{x_j} \right\|^2, 
\\D_2& \equiv \max_{l = 1, \ldots, n} \sup_{\left\{ x_j \in \mathbb{Z}^d, \left| \mathbf{z} \right| = l \right\}}  \sum_{d (y, x_j) \leq R} \sum_{k
   = 1  }^l \sum_{\left| \mathbf{z}' \right| = k} 
   \left( \max \left\| \mathbf{Y}^{(l - k)}_{\mathbf{z} \setminus \mathbf{z}'}
   \mathfrak{S}_{yx_j} \right\|^2, \left\| \mathbf{Y}^{(l -
   k)}_{\mathbf{z} \setminus \mathbf{z}'} \mathfrak{S}_{x_j y} \right\|^2
   \right). \end{align*}
Since each tree connecting points in $\mathbf{z}' x_j$, $\mathbf{z}'
yx_j$ and $\mathbf{z}' x_j y$ with $\Lambda (f)$ is of length at least
$d (x_j, \Lambda (f))$, applying Theorem \ref{fspthm} we obtain
\[ \left\| \mathbf{Y}^{(l)}_{\mathbf{z}} F \right\|^2 \leq
   De^{(Cs - vd (x_j, \Lambda (f)))} \sum_{k = 1, \ldots,
   l + 2} \sum_{|\mathbf{z}'| = k } \left\| \mathbf{Y}^{(k)}_{\mathbf{z}'} f
   \right\|^2 ,\]
with some constant $D \in (0, \infty)$ independent of $s$, $x_j$ and the
function $f$. Combining our estimates we arrive at
\[ \left| \mathbf{Y}^{(n)}_{\mathbf{x}} \PP^{\Lambda_j}_t f
   -\mathbf{Y}^{(n)}_{\mathbf{x}} \PP^{\Lambda_{j - 1}}_t f \right| \leq D'
   e^{\frac{1}{2} (Ct - vd (x_j, \Lambda (f)))} \left( \sum_{k =
   1, \ldots, n + 2} \sum_{|\mathbf{z}'|=k} \left\| \mathbf{Y}^{(k)}_{\mathbf{z}'}
   f \right\|^2  \right)^{\frac{1}{2}}, \]
with some constant $D' \in (0, \infty)$ independent of $t$, $x_j$ and
the function $f$. Using a similar telescopic expansion as in the proof of
existence of the limit for the semigroup, this implies that the sequence
$\mathbf{Y}^{(n)}_{\mathbf{x}} \PP^{\Lambda_j}_t f$, $j \in \mathbb{N}$, is
Cauchy in the supremum norm for every $n \in \mathbb{N}$ and $\mathbf{x}$,
$\left| \mathbf{x} \right| = n$. This ends the proof of the theorem.
\endproof

\subsection{Smoothing properties of infinite dimensional semigroup}
\label{sec:smoothingng}
In this section we extend Theorem 2.1 to infinite dimensions, proving
smoothing estimates in the setup when the fields at each site of \(\ZZ^d\) satisfy the
commutation relations of Assumption  \textup{(\textbf{CR.I})}. Now our generator has the form
\[ \mathcal{L}_{} \equiv \ensuremath{\boldsymbol{L}}+ 
   \ensuremath{\boldsymbol{L}}_{\ensuremath{\operatorname{int}}} \]
with
\[ \ensuremath{\boldsymbol{L}} \equiv \sum_{x \in \mathbb{Z}^d} L_x  \]
where
\[ L_x \equiv Z_{J,x}^2 + B_x - \lambda D_x \]
and
\[  \ensuremath{\boldsymbol{L}}_{\ensuremath{\operatorname{int}}}
   \equiv \sum_{x \in \mathbb{Z}^d}  \mathbf{q}_x \cdot
   Z_x + \sum_{y, y' \in \mathbb{Z}^d} 
   \mathfrak{S}_{yy'} \cdot
   Z_{J, y} {Z}_{J, y'}, \]
 recalling that ${J \subset I}$ . 
For notational simplicity we  only describe one component type system, with
$\mathfrak{S}_{klyy'} \equiv
\mathfrak{S}_{00yy'} \equiv
\mathfrak{S}_{yy'}$ and
$Z_{J, y} \equiv Z_{0, y} $, but  provide  sufficient detail
to  make clear how to recover the more general case with many components.
For \(n \in \N\), we introduce the form $\ensuremath{\boldsymbol{\Gamma}}^{(n)}_t$ as follows. For \(n=1\), we consider the following
quadratic form
\[ \ensuremath{\boldsymbol{\Gamma}}_t^{(1)} (f_t) \equiv \sum_{x \in
   \mathbb{Z}^d} \Gamma_{t, x}^{(1)} (f_t), \]
with $\Gamma_{t, x}^{(1)} (f)$ being an isomorphic copy of the form \eqref{gamma_1} defined in
Section 2
\begin{eqnarray*}
  \Gamma_{t, x}^{(1)} (f_t) & \equiv & \sum_{i = 0, \ldots, N} (a_i t^{2 i_{}
  + 1} | Z_{i, x} f_t |^2 + b_i t^{2 i} Z_{i - 1, x} f_t \cdot Z_{i, x} f_t),
\end{eqnarray*}
with the convention that  $Z_{- 1} \equiv
Z_0$  (so de facto there is no spurious term in the second sum on the right
hand side). 
For \(n>1\), we define
\be\label{Gammantft}
\ensuremath{\boldsymbol{\Gamma}}^{(n)}_t  (f_t) \equiv \sum_{\mathbf{x} \in
   \mathbb{Z}^{nd}}
   \ensuremath{\boldsymbol{\Gamma}}^{(n)}_{t, \mathbf{x}} f_t 
\ee
with $\ensuremath{\boldsymbol{\Gamma}}^{(n)}_{t, \mathbf{x}}  (g) \equiv
\ensuremath{\boldsymbol{\Gamma}}^{(n)}_{t, \mathbf{x}}  (g, g),$ where
\begin{align*}
 {\boldsymbol{{\Gamma}}}^{n}_{t,\mathbf{x}} (g,h)   {\equiv}& 
  \sum_{|\mathbf{k}|_{n}=0}^{nN}a_{\mathbf{k},n} t^{2|\mathbf{k}|_{n}+n}
  \mathbf{Z}_{\mathbf{k},n,\mathbf{x}}g{\cdot} \mathbf{Z}_{\mathbf{k},n,\mathbf{x} }h  +
   \sum_{0 \leq |\mathbf{k}|_{n}:k_{1}{\geq}1}^{nN}b_{\mathbf{k},n}t^{2|\mathbf{k}|_{n}+n-1}(\mathbf{Z}_{\mathbf{k}-\mathbf{e}_{1},n,\mathbf{x}} g)(\mathbf{Z}_{\mathbf{k},n,\mathbf{x}}h).
\end{align*}
Here and later $\mathbf{Z}_{\mathbf{k}, n, \mathbf{x}} \equiv Z_{k_1, x_1}
\ldots Z_{k_n, x_n}$ , for $\mathbf{x} \equiv (x_1, \ldots, x_n) \in
\mathbb{Z}^{nd}$, and $\mathbf{k} \equiv (k_{1,
\ldots,} k_n) \in \{ 1, \ldots, N \} \times \{0 \dd N\}^{n-1}$. The main result of this section is the following. 
\begin{theorem}[Infinite Dimensional Smoothing Estimates]
\label{thm:smooth}
 Suppose that for every $x,y\in \ZZ^d$
\begin{eqnarray} 
    \phantom{\mathbb{A}\mathbb{A}\mathbb{A}\mathbb{A}\mathbb{A}\mathbb{A}A}
    Z_{i, x} q_{\ensuremath{\operatorname{}jy}} = 0  &  &
    \hspace{2em} \ensuremath{\operatorname{if}}j > i,
    \phantom{\mathbb{A}\mathbb{A}\mathbb{A}\mathbb{A}\mathbb{A}\mathbb{A}A}
    \label{si1a}\\
    \phantom{\mathbb{A}\mathbb{A}\mathbb{A}\mathbb{A}\mathbb{A}\mathbb{A}A}
    \sum_{j = 1, \ldots, N} c_{ijk}
    q_{jx} = 0 &  & \hspace{2em}
    \ensuremath{\operatorname{if}}k > i,
    \phantom{\mathbb{A}\mathbb{A}\mathbb{A}\mathbb{A}\mathbb{A}\mathbb{A}A}
        \label{si2a}\\
    \phantom{\mathbb{A}\mathbb{A}\mathbb{A}\mathbb{A}\mathbb{A}\mathbb{A}A}
    c_{i 0 k} = 0 &  & \hspace{2em} \ensuremath{\operatorname{if}}k > i,
    \phantom{\mathbb{A}\mathbb{A}\mathbb{A}\mathbb{A}\mathbb{A}\mathbb{A}A}
       \label{si3a} \\
    \phantom{\mathbb{A}\mathbb{A}\mathbb{A}\mathbb{A}\mathbb{A}\mathbb{A}A}
    \sum_{k = 1, \ldots, N} c_{i 0 k} c_{k 0 l} = 0 &  & \hspace{2em}
    \ensuremath{\operatorname{if}}l > i, 
    \phantom{\mathbb{A}\mathbb{A}\mathbb{A}\mathbb{A}\mathbb{A}\mathbb{A}A}  
     \label{si4a}
      \end{eqnarray}
and recall that the commutator relations of Assumption \textup{(\textbf{CR.I})} are assumed to hold at each site. 
  Then, there exist coefficients $a_{\mathbf{k},n}, b_{\mathbf{k},n}, d >0$ (appearing in the definition of ${\boldsymbol{{\Gamma}}}^{n}_{t,\mathbf{x}}$),  $\varepsilon \in (0,
  \infty)$ and $t_0 \in (0, 1)$, such that if
  \[ \sup_{k_j, z_j, i, y} \| Z_{k_j, z_j}
     q_{iy} \| <
     \varepsilon, \]
  then for any $t \in (0, t_0)$ one has
  \[ \sum_{l = 1, \ldots, n} \ensuremath{\boldsymbol{\Gamma}}_t^{(l)} (f_t)
     \leq d (\PP_t (f^2) - (\PP_t f)^2). \]
\end{theorem}
\begin{proof} The proof of Theorem \ref{thm:smooth} is very similar to the proof of Theorem \ref{thm.1} and  can be found in Appendix A. 
\end{proof}
Observe that, if Assumption \textbf{(CR.I)} holds, \eqref{si3a} is redundant as $c_{i0k}=- c_{0ik}=0$ for $k \ge i-1$. Also, the assumptions in the statement of this theorem are all purely technical and are there in order to enable us to extend the technique introduced in Section \ref{S.2:Short and long time behaviour} for the finite dimensional setting to the present infinite dimensional environment. 

\label{EndConstruction}

\section{Existence of invariant states for the infinite dimensional semigroup}\label{S4ExistenceInvariantStates} 
In this section we consider the operators \(\mathcal{L}\)   which are obtained as the limits of \(\mathcal{L}_\Lambda\)  as \(\Lambda \uparrow \ZZ^d\). We will provide a strategy for the associated semigroups \(\PP_t \) which a priori may depend on the initial configuration.  To start with, consider the operator   \(L\) given in Section 2,   on \(\R^m\) equipped with a metric \(\bfd\).  For any \(x \in \ZZ^d\),  we consider the semi-distance      \(\bfd_x(\omega) = \bfd(\omega_x)\), \((\R^m)^{\ZZ^d}\ni \omega = \{\omega_y\in \R^m \}_{y \in \ZZ^d}\), and set $\rho_x(\omega)\equiv \phi(\bfd_x(\omega))$,
for some smooth increasing \(\phi \) with bounded derivative.
Given summable weights $(\epsilon_x\in(0,\infty))_{x \in \ZZ^d}$, $\sum_{x \in \ZZ^d} \epsilon_x<\infty$, we define the set 
\begin{equation}\label{omega000}
\Omega_0 = \left\{\omega \in (\R^m)^{\ZZ^d}: \sum_{x \in \ZZ^d}  \epsilon_x \rho_x( \omega)< \infty\right\}.  
\end{equation}
The following assumption plays a key role in the proof of our results. 
\begin{assumption}
\label{assumptionOnExistenceOfRho}
There exists a smooth function
\(\rho : \R^m \rightarrow \R\),  such that \(\rho (u) \rightarrow \infty \) as \(\bfd (u)\rightarrow \infty  \), 
and with compact level sets (i.e. the sets
 \(\{\rho< \ell \}\), \(\ell>0\), are compact); moreover the function  $\rho$  satisfies the Lyapunov-type condition 
\be
\label{lyapunovcondition}
L\rho  < C_1 - C_2\rho,
\ee
for some constants \(C_1 \ge 0\) and \(C_2>0\).
\end{assumption}

Examples of generators for which \(\rho = \phi(\mathbf{d})\) 
satisfies \eqref{lyapunovcondition} will be given elsewhere \cite{MVII}. The moral behind Theorem \ref{ergodicitytheorem} below is the following: roughly speaking, if we are able to exhibit a Lyapunov function $\rho$ for the finite dimensional dynamics,  then $\sum_{x\in\ZZ^d}\rho_x$ (where $\rho_x$ is a ``copy" of $\rho$ acting at $x\in \ZZ^d$) is a candidate Lyapunov function for the infinite dimensional generator; this is, provided some assumptions involving the interaction functions are satisfied, see \eqref{ccn} and 
\eqref{ecn}. With this in mind, 
we have  the following result on existence of invariant measures.
\begin{theorem} \label{ergodicitytheorem} Suppose that \(\rho\) satisfies Assumption \ref{assumptionOnExistenceOfRho} and let \( \PP_t \)  be the semigroup  generated by 
\be\label{genL}
\cl =
\sum_{x\in \ZZ^d} L_x
-\sum_{x\in\ZZ^d}  \q_x\cdot Z_x
-\sum_{yy'\in\ZZ^d\atop y\neq y'}
\s_{yy'}\cdot Z_{0,y}Z_{0,y'}\,.
\ee 
Assume also that  \(\rho\) is such that 
\be \label{ccn}
-\sum_{z\in\ZZ^d}  \q_z\cdot Z_z\rho_z \leq C_3+ \sum_{y\in\ZZ^d} \eta_{x,y}\rho_y  \qquad {\mbox{for all }} x \in \ZZ^d, 
\ee
for some constant $C_3>0$ and for some sequence  $\{\eta_{xy}\}_{x,y \in \mathbb{Z}^d}$  of positive numbers  satisfying 
\be \label{ecn}
S\equiv  \sup_{x\in\ZZ^d}\sum_{y\in\ZZ^d}\eta_{x,y}<\infty \qquad and \qquad \sum_{x\in\ZZ^d}\epsilon_x \eta_{x,y} \leq C_4 \epsilon_y, 
\ee
for some constant $C_4<C_2$ (and with $\epsilon_x$ being the sequence introduced in \eqref{omega000}). 
Then there exists a subsequence \( (t_k)_{k \in \N}\subset \R \) and  a probability measure \(\mu_\omega\) such that \(\mu_\omega (\Omega_0) = 1\) and 
\be \PP_{t_k}  f(\omega) \rightarrow \mu_\omega (f), 
\ee
as \(k \rightarrow \infty\), for all bounded smooth cylinder functions \(f\) and all \(\omega \in \Omega_0\).
\end{theorem}
\begin{remark}
One can see that, if $\eta_{x,y} \equiv 0$ when $dist(x,y)\geq R$ (for some $R\in(0,\infty)$), then   condition \eqref{ecn} is satisfied for polynomially as well as exponentially decaying weights.
\end{remark}

\begin{proof}
The proof of Theorem \ref{ergodicitytheorem} consists of the following steps. We start by constructing a Lyapunov function for the operator \(\mathcal{L}\)   using a suitable function \(\rho\). We then use this function to deduce that the corresponding semigroup converges weakly to a probability measure, pointwise with respect to the initial configuration \(\omega \in \Omega_0\). Finally, we show that the limit measure is independent of the initial configuration. 
\par We  consider \( \sum_{x \in \ZZ^d}\rho_x\), with \(\rho_x\) as above. 
Since $Z_{lr,y}\rho_x=0$ whenever $x\neq y$, using Assumption \ref{assumptionOnExistenceOfRho}, we obtain 
\begin{equation}\label{llyapunov.1}
\mathcal{L} \rho_x  = L_x \rho_x -\sum_{x\in\La}  \q_x\cdot Z_x\rho_x \leq
C_1- C_2\rho_x -\sum_{x\in\La}  \q_x\cdot Z_x\rho_x
\end{equation}
Thus if \eqref{ccn} holds, 
setting  $\bar C \equiv C_1+C_3$, we have
\[\frac{d}{dt} \PP_{t} \rho_x = \PP_{t} \mathcal{L} \rho_x \leq
\bar C  - C_2\PP_{t}\rho_x + \sum_{y\in\ZZ^d} \eta_{x,y}\PP_{t}\rho_y\,.
\]
From  the above, 
\begin{align*}
\partial_t\left( e^{C_2t}  \mathcal{P}_t \rho_x\right) &= e^{C_2t} \partial_t \mathcal{P}_t \rho_x + C_2 e^{C_2 t}\mathcal{P}_t \rho_x \\
& \le \bar{C} e^{C_2t} + e^{C_2t} \sum_{y \in \mathbb{Z}^d} \eta_{x, y} \mathcal{P}_t \rho_y
\end{align*}
which after integration yields
\begin{align*}
e^{C_2t} \mathcal{P}_t \rho_x &\le \rho_x + \bar{C} \frac{e^{C_2t} - 1}{C_2} + \int_0^t e^{C_2s} \sum_{y \in \mathbb{Z}^d} \eta_{x, y} \mathcal{P}_s \rho _y ds
\end{align*}
Hence
\[ \PP_{t} \rho_x (\omega)\leq \bar C / C_2 + e^{-C_2t}\rho_x (\omega)+ \sum_{y\in\ZZ^d} \eta_{x,y} \int_0^tds e^{-C_2(t-s)}\PP_{s}\rho_y
\]
and for any summable weights $\epsilon_x\in(0,\infty)$ satisfying  $\sum_{x \in \ZZ^d} \epsilon_x = M < \infty$ we have
\[\PP_{t} \sum_{x\in\ZZ^d}\epsilon_x\rho_x(\omega) \leq   C' + 
e^{-C_2t}\sum_{x\in\ZZ^d}\epsilon_x\rho_x (\omega)+ \sum_{x\in\ZZ^d}\epsilon_x\sum_{y\in\ZZ^d} \eta_{x,y} \int_0^tds e^{-C_2(t-s)}\PP_{s}\rho_y
\]
with $C' = \bar C M/ C_2$.  Using the second condition in \eqref{ecn}, 
 for 
\[F_t\equiv\PP_{t} \sum_{x\in\ZZ^d}\epsilon_x\rho_x(\omega) \]
we get the following relation
\[
F_{t}  \leq   C' + 
e^{-C_2t}F_0+ C_4 \int_0^tds e^{-C_2(t-s)}F_{s}.
\]
This implies that
\be
\sup_{0\leq s\leq t} F_s \leq  (1- \bar \kappa )^{-1}\left(   C'+ 
 F_0 \right)  \ee
which is finite and uniformly bounded in $t$,  provided that \(\bar\kappa\equiv\frac{C_4}{C_2} \in(0,1)\) and  \( \sum_{x\in\ZZ^d}\epsilon_x\rho_x (\omega) <\infty\), i.e. we have

\be\label{ubd}
\sup_{t\geq 0}\left(\PP_{t} \sum_{x\in\ZZ^d}\epsilon_x\rho_x\right)(\omega) \leq 
(1- \bar \kappa )^{-1}\left(  C' + 
 \sum_{x\in\ZZ^d}\epsilon_x\rho_x (\omega) \right)
 \ee
(Strictly speaking one applies first all the above arguments to a smooth cutoff   $\rho_x^{A}\leq A<\infty$ of $\rho_x$ and after applying the formal Gronwall arguments, we pass to the limit $A\to\infty$. This is more lengthy to write, but there is no technical difficulty in that.)

 The existence of such uniform bound \eqref{ubd}  implies  (\cite{DKZ2011}, Section 3.2) the weak convergence of a subsequence of \((\PP_t)_{t\ge 0}\) for an initial configuration \(\omega \in \Omega_0\), i.e. the existence of a sequence \((t_k)_{k \in \N} \subset \R\) and a measure \(\mu_\omega\) such that for all bounded and smooth cylinder functions \(f\) 
\[
\PP_{t_k}f(\omega) \rightarrow \mu_\omega(f),
\]
as \(k \rightarrow \infty\), for all \(\omega \in \Omega_0 \). Consider the set \[\Omega_\ell = \left\{\widetilde\omega \in (\R^m)^{\ZZ^d}: \sum_{x\in\ZZ^d}\epsilon_x\rho_x (\widetilde\omega) < \ell \right\}.\]
Using Markov's inequality we obtain, for all \(\omega \in \Omega,\)
\[\mu_\omega(\Omega_\ell)\ge       1- \frac{1}{\ell}
\sup_{t\geq 0}\left(\PP_{t} \sum_{x\in\ZZ^d}\epsilon_x\rho_x\right)(\omega)\geq  1- \frac{1}{\ell}\left((1- \bar \kappa )^{-1}\left(  C' + 
 \sum_{x\in\ZZ^d}\epsilon_x\rho_x (\omega) \right)\right),\]
and thus taking the limit as \(\ell \rightarrow \infty\), we conclude that \(\mu_\omega(\Omega) = 1.\)

\end{proof}
\section{Ergodic properties of the infinite dimensional semigroup}\label{S5Ergodic Properties} 


We begin this section by a result for the  semigroup  constructed in Section \ref{S.3:Infinite dimensional semigroups}, in the case where the Lie algebra is stratified  (see Remark \ref{rem3.1}) for each $x\in\ZZ^d$, and equipped with a dilation generator $D_x$. In the remainder of this section, we will denote 
$$ \bar\Omega_\delta\equiv \left\{\omega \in (\R^m)^{\ZZ^d} : \sum_{x\in\ZZ^d}  \frac{{\bf{d}}(\omega_x)}{(1+dist(x,0))^{d+\delta}} < \infty\right\}$$
for $\delta > 0$.

\begin{theorem} \label{thm5.1}
Consider the operator
\[\mathcal{L}\equiv 
\sum_{x\in \ZZ^d} L_x
-\sum_{x\in\ZZ^d}  \q_x\cdot Y_x
-\sum_{yy'\in\ZZ^d\atop y\neq y'}
\s_{yy'}\cdot Y_{0,y}Y_{0,y'} -\lambda \sum_{x\in\ZZ^d} D_x
\]
  where $\lambda>0$ and
\[ L_x \equiv \ensuremath{{Y}}_{J,x}^2 + B_x, \]
 where ${J \subset I}$ (see notation in Section \ref{infinitDim}).  For every $n\in\mathbb{N}$ there exists $\lambda_n\in(0,\infty)$ such that
for any $\lambda\geq \lambda_n$ one has
\[ \sum_{\x\in\ZZ^{nd}} |\mathbf{Y}_{\x}^{(n)} f_t|^2 \leq e^{-m_nt} \sum_{\x\in\ZZ^{nd}} |\mathbf{Y}_{\x}^{(n)} f |^2  
\]
with some $m_n\in(0,\infty)$. 
Hence for any $\omega,\omega'\in\bar\Omega_\delta$, defined with some $\delta\in(0,\infty)$, the associated semigroup \(\mathcal{P}_t\equiv e^{t\mathcal{L}}\) satisfies, for all smooth cylider functions $f$,
\[  |\mathcal{P}_tf(\omega) - \mathcal{P}_tf(\omega')|\leq C e^{-m t} \sum_{\x\in\ZZ^{d}} \Vert  Y_x f  \Vert \]
with some $m\in(0,\infty)$ independent of $f$ and some constant $C$ dependent on $\Lambda(f)$ and $\omega,\omega'\in\bar\Omega_\delta$. 
\end{theorem}
For the full gradient bound estimate see \cite{DKZ2011} (Lemma 3.1 and Remark 3.2). The ergodicity statement follows via a similar strategy as in \cite{DKZ2011}. We consider the lexicographic order on the lattice introduced in the proof of Theorem  \ref{existencethm} and an interpolating sequence  of points in $(\R^m)^{\ZZ^d}$, $\{\omega^{(j)}\}$, where each point of the sequence is defined as follows:
 $$(\omega^{(j)})_{x_k}\equiv \begin{cases}
\omega_{x_k} \textup{, if }k\leq j\\
\omega'_{x_k} \textup{, if }  k> j . 
\end{cases}
$$
With this interpolation we consider the following telescopic expansion
\[\mathcal{P}_tf(\omega) - \mathcal{P}_tf(\omega')= \sum_k \left(\mathcal{P}_tf(\omega^{(k+1)}) - \mathcal{P}_tf(\omega^{(k)}) \right)
\] 
and notice that for a piecewise differentiable unit speed path $\gamma^{(k)}_\tau$ such that  $\gamma^{(k)}_{\tau=0} =\omega^{(k)}_{x_{k}}$  and $\gamma^{(k)}_{\tau=1} =\omega^{(k+1)}_{x_{k}}$ with tangent vectors given by $\mathbf{Y}$ (such a path exists by Chow's Theorem, see e.g. \cite{BLU}), we have 
\[ 
\left|\mathcal{P}_tf(\omega^{(k+1)}) - \mathcal{P}_tf(\omega^{(k)})\right| = 
\left| \int_0^1 d\tau  {\dot{\gamma}}^{(k)}_\tau\cdot \nabla_{\mathbf{Y}_{x_k}} \mathcal{P}_tf(\gamma^{(k)}_\tau)  \right| \leq {\bf{d}}(\omega_{x_k}, \omega_{x_k}') \|\nabla_{\mathbf{Y}_{x_k}} \mathcal{P}_tf\|.
\]
The sum of such terms over $\{k: |x_k|\geq Ct\}$, with suitable constant $C\in(0,\infty)$, can be bounded using finite speed of propagation of information
by a factor converging exponentially quickly to zero with respect to $t$.
The remaining contribution can be estimated as follows.
\[
\sum_{k: dist(x_k,0)\leq Ct} \left|\mathcal{P}_tf(\omega^{(k+1)}) - \mathcal{P}_tf(\omega^{(k)})\right|
\leq C^d t^d \max_{dist(x,0)\leq Ct} ({\bf{d}}(\omega_x), {\bf{d}}(\omega_x')) \cdot 
|\mathbf{Y} \mathcal{P}_t f|.
\] 
Thus for $\omega,\omega'
$ in the set $\bar\Omega_\delta$, to get the uniqueness of the limit it is sufficient to show
\[ |\mathbf{Y} \mathcal{P}_tf| \leq C' t^{-d-2\delta}
\] 
with some finite constant $C'$.
A similar idea to prove uniqueness of the limit $\lim_{t\to\infty}\mathcal{P}_tf$ can be used in the situation when additional restrictions on the commutation relations are imposed. 
 
\begin{theorem} \label{thm5.2}
%
  Suppose that Assumption \textup{(\textbf{CR.I})} is satisfied with 
  $c_j=0$ and $c_{0jk}=0$, $j,k=1,..,N$.
 Assume additionally that \eqref{si1a} and \eqref{si2a} hold, together with 
  \[Z_{i, x} {q}_{\ensuremath{\operatorname{}}jy} = 0  
    \hspace{2em} \ensuremath{\operatorname{if}} j \neq  i
    \]
    and \[Z_{i, x} \mathfrak{S}_{kk',yy'} = 0  .
    \]
  Under these assumptions, (and recalling the notation and definition \eqref{Gammantft}) there exist coefficients $a_i , b_i  , d_0  , \varepsilon \in (0,
  \infty)$, such that if
  \[ \sup_{k_j, z_j, i, y} \| \ensuremath{{Z}}_{k_j, z_j}
     \ensuremath{{q}}_{\ensuremath{\operatorname{}iy}} \| <
     \varepsilon,
      \]
  then for any $t \in (0, \infty)$ one has
  \[ \sum_{l = 1, \ldots, n} \ensuremath{\boldsymbol{\Gamma}}_t^{(l)} (f_t)
     \leq d_0 (\mathcal{P}_t (f^2) - (\mathcal{P}_t f)^2) \]
     for all smooth cylinder functions $f$.  Hence, if $[d+2\delta]\leq N$, then for bounded smooth functions $f$
\[ \|(B,Z_{0}, ...,Z_{j_{max}})\mathcal{P}_tf\| \leq C t^{-d-2\delta},
\]
for some constant $C \in (0, \infty)$ dependent on $f$, with  $j_{\max}\equiv \min([d+2\delta], N)$. 
Moreover  the limit $\lim_{t\to\infty} \mathcal{P}_tf(\omega)$ is unique for $\omega\in\bar\Omega_\delta$.
\end{theorem}
We notice that the decay in the directions of $Z_j$ with $j> [d+2\delta]$ is automatically sufficiently fast. Thus for the question of uniqueness it is sufficient
to concentrate on estimates in direction $B$ and $Z_j,\, j\leq [d+2\delta]$. Finally we mention that in the same situation one can take advantage of higher order estimates as follows.  
\begin{theorem}\label{thm5.3}
%
  Under the conditions of Theorem \ref{thm5.2}, assume the higher order bounds including $\mathbf{Y}\equiv (B,\mathbf{Z})$ are true globally in time.
If for some configuration $\widetilde\omega\in\bar\Omega_\delta$ one has
for any bounded cylinder function $f$, 
\[  |\mathbf{Y}^{(n)}\mathcal{P}_tf(\widetilde\omega)| \leq C_n t^{-d-2\delta},
\]
for some constants $C_n \in (0, \infty)$ and $n\leq n_{\max}\equiv [d+2\delta]$, with $\mathbf{Y}\equiv (B,\mathbf{Z})$,
then the limit $\lim_{t\to\infty} \mathcal{P}_tf(\omega)$ is unique for all $\omega\in \bar\Omega_\delta$.
\end{theorem}
This result follows in the similar fashion as before re-expanding $\nabla_{\mathbf{Y}_{x_k}} \mathcal{P}_tf(\gamma^{(k)}_\tau) $ sufficiently many times.
\vspace{0.5 cm}

{\bf Acknowledgments.} The authors are grateful to the thoughtful referees that helped improving the paper, both its content and exposition. 

\section*{Appendix A}
This Appendix contains  the proofs of  Proposition \ref{pro3.1} and  Theorem \ref{thm:smooth}. 
\subsection*{A.1 Proof of Proposition \ref{pro3.1} }
\begin{proof}[Proof of Proposition \ref{pro3.1}.]
\label{BeginPRF} For $t \geq s \geq 0$, we have:
\begin{align}
  \frac{\partial}{\partial s}\mathcal{P}_{t-s}^{\Lambda}
  \lv\mathbf{Y}_{{\boldsymbol\iota},\mathbf{x}}^{(n)}f_s^{\Lambda}\rv^2
  & =\mathcal{P}_{t-s}^{\Lambda}\left\{\left( -\mathcal{L}_{\Lambda}+\frac{\partial}{\partial_s} \right) \lv \mathbf{Y}_{{\boldsymbol\iota},\mathbf{x}}^{(n)} f_s^{\Lambda}\rv^2 \right\}\nonumber\\
  & =\mathcal{P}_{t-s}^{\mathcal{L}_{\Lambda}}\left\{
 -\mathcal{L}_{\Lambda}\left|\mathbf{Y}_{{\boldsymbol\iota},\mathbf{x}}^{(n)}f_s^{\Lambda}\right|^2+2\left( \mathbf{Y}_{{\boldsymbol\iota},\mathbf{x}}^{(n)} f_s^{\Lambda}\right)
 \left(\mathcal{L}_{\Lambda}\mathbf{Y}_{{\boldsymbol\iota},\mathbf{x}}^{(n)}
f_s^{\Lambda}\right)
  +2\left( \mathbf{Y}_{{\boldsymbol\iota},\mathbf{x}}^{(n)}f_s^{\Lambda}\right)
  \left[\mathbf{Y}_{{\boldsymbol\iota},\mathbf{x}}^{(n)},\mathcal{L}_{\Lambda}\right]   f_s^{\Lambda} 
\right\}.  \label{63}
\end{align}
First we note that
\begin{align}
  -\mathcal{L}_{\Lambda}\left|\mathbf{Y}_{{\boldsymbol\iota},\mathbf{x}}^{(n)}f_s^{\Lambda}\right|^2
  +2\left(\mathbf{Y}_{{\boldsymbol\iota},\mathbf{x}}^{(n)}f_s^{\Lambda}\right) \mathcal{L}_{\Lambda}\mathbf{Y}_{{\boldsymbol\iota},\mathbf{x}}^{(n)} f_s^{\Lambda}
   =&-2 \sum_{z\in\mathbb{Z}^d}\left|Y_{J,z}
   \mathbf{Y}_{{\boldsymbol\iota},\mathbf{x}}^{(n)}f_s^{\Lambda}\right|^2 \nonumber\\
  & -2 \sum_{z,z'\in\Lambda} \mathfrak{S}_{zz'}{\cdot}
  \left(Y_{z}\mathbf{Y}_{{\boldsymbol\iota},\mathbf{x}}^{(n)}f_s^{\Lambda}\right){\cdot}
  \left(Y_{z'}\mathbf{Y}_{{\boldsymbol\iota},\mathbf{x}}^{(n)}f_s^{\Lambda}\right)\label{64}.
\end{align}
Next the last addend in \eqref{63} can be decomposed as follows:

\begin{align*}
  &2\left(\mathbf{Y}_{{\boldsymbol\iota},\mathbf{x}}^{(n)}f_s^{\Lambda}\right)\left[\mathbf{Y}_{{\boldsymbol\iota},\mathbf{x}}^{(n)},\mathcal{L}_{\Lambda}\right] f_s^{\Lambda}  \\
 & =\sum_{z\in\mathbb{Z}^d}2\left(\mathbf{Y}_{{\boldsymbol\iota},\mathbf{x}}^{(n)} f_s^{\Lambda}\right){\hspace{0.25em}}\left[\mathbf{Y}_{{\boldsymbol\iota},\mathbf{x}}^{(n)},
L_{z}\right]f_s^{\Lambda}  \qquad \qquad \qquad \qquad \qquad \qquad \qquad \qquad \qquad \qquad(\mathrm{{ \bf T_1}}) \\
 & +2 \sum_{z\in\Lambda}\sum_{\beta\in I}\mathbf{q}_{{\beta},z}\left(\mathbf{Y}_{{\boldsymbol\iota},\mathbf{x}}^{(n)}f_s^{\Lambda}\right) {\hspace{0.25em}} \left[\mathbf{Y}_{{\boldsymbol\iota},\mathbf{x}}^{(n)},Y_{{\beta},z}\right]f_s^{\Lambda}
  \qquad \qquad \qquad \qquad \qquad \qquad \qquad \qquad \,\,   (\mathrm{{ \bf T_2}})\\
&+\sum_{z\in\Lambda} \sum_{\beta\in I} 2\left(\mathbf{Y}_{{\boldsymbol\iota},\mathbf{x}}^{(n)}
f_s^{\Lambda}\right){\hspace{0.25em}}\langle\mathbf{Y}_{{\boldsymbol\iota},\mathbf{x}}^{(n)},\mathbf{q}_{{\beta},z}\rangle
Y_{{\beta},z} f_s^{\Lambda}   \qquad \qquad \qquad \qquad \qquad \qquad\,\,\, \qquad \qquad \,\,\,\,
(\mathrm{{ \bf T_3}})\\
 &+\sum_{z,z'\in\Lambda}\sum_{{\gamma} {\gamma}'\in J}
  2\mathfrak{S}_{{\gamma}{\gamma}',zz'}\left(\mathbf{Y}_{{\boldsymbol\iota},\mathbf{x}}^{(n)}f_s^{\Lambda}\right) \left( \left[\mathbf{Y}_{{\boldsymbol\iota},\mathbf{x}}^{(n)},Y_{{\gamma},z}\right]
  Y_{{\gamma}',z'}+Y_{{\gamma},z}\left[\mathbf{Y}_{{\boldsymbol\iota},\mathbf{x}}^{(n)},Y_{{\gamma}',z'}\right]\right) f_s^{\Lambda}  \qquad  \,\, (\mathrm{{ \bf T_4}})  \\
 &  +\sum_{z,z'\in\Lambda}\sum_{{\gamma} {\gamma}'\in J}
  2\left(\mathbf{Y}_{{\boldsymbol\iota},\mathbf{x}}^{(n)}f_s^{\Lambda}\right)\langle\mathbf{Y}_{{\boldsymbol\iota},\mathbf{x}}^{(n)},\mathfrak{S}_{{\gamma}{\gamma}',zz'}\rangle \,Y_{{\gamma},z}  Y_{{\gamma}',z'} f_s^{\Lambda}. \qquad \qquad \qquad \qquad\qquad \qquad \,\,\, (\mathrm{{ \bf T_5}})
\end{align*}
Regarding the terms $(\mathbf{T_2})$ and $(\mathbf{T_3})$ (and similar comments hold for $(\mathbf{T_4})$ and $(\mathbf{T_5})$), such terms have been obtained by writing
$$
\left[\mathbf{Y}_{{\boldsymbol\iota},\mathbf{x}}^{(n)},{q}_{{\beta},z} Y_{\beta,z}\right] = {q}_{{\beta},z} \left[ \mathbf{Y}_{{\boldsymbol\iota},\mathbf{x}}^{(n)}, Y_{\beta,z} \right] + \langle\mathbf{Y}_{{\boldsymbol\iota},\mathbf{x}}^{(n)},{q}_{{\beta},z}\rangle \,.
$$
An explicit expression for the bracket $\langle\mathbf{Y}_{{\boldsymbol\iota},\mathbf{x}}^{(n)},g\rangle $, where $g$ is any sufficiently smooth function, can be found below in \eqref{3.1.7}. Beyond the specific expression written in \eqref{3.1.7}, what is important to notice is that the bracket $\langle\mathbf{Y}_{{\boldsymbol\iota},\mathbf{x}}^{(n)},g\rangle $ contains only differential operators of order up to $n-1$. Therefore 
 in  $\mathbf{T_3}$, the term $\langle\mathbf{Y}_{{\boldsymbol\iota},\mathbf{x}}^{(n)},\mathbf{q}_{{\beta},z}\rangle
Y_{{\beta},z}$  contains only differential operators of order up to $n$. 

To estimate each of the  terms  $(\mathbf{T_1})$- $(\mathbf{T_5})$   we use lengthy but elementary arguments of
which we list the result in Lemma  \ref{lem3.2} to Lemma \ref{lem3.6} below,   and briefly sketch an idea of the proof.
The estimates of $(\mathbf{T_1})$  and $(\mathbf{T_2})$ in the first two lemmas below are based on 
our locality assumption, i.e. the fact that 
$\mathbf{Y}_{{\boldsymbol\iota}, \mathbf{x}}^{(n)}$ and $L_z$,
$Y_{\beta, z}$ \ commute unless $z \in \mathbf{x}$, and the structure of
$L_z$, together with the quadratic Young's inequality \eqref{young}. We recall the notation $c\equiv \sup_{\alpha,\gamma,\beta}\lv 
c_{\alpha\gamma\beta} \rv$, where $c_{\alpha\gamma\beta}$ are as in Assumption (\textup{\textbf{GCR}}).
\begin{lemma}[Estimate of ($\mathrm{{ \bf T_1}}$)]\label{lem3.2} 
 Under the assumptions of Proposition \ref{pro3.1}, for any
  $\varepsilon \in (0, \infty)$ we have 
  \[  2 \sum_{{\boldsymbol\iota}}  \left| \left( 
     \mathbf{Y}_{{\boldsymbol\iota}, \mathbf{x}}^{(n)}
     f_s^{\Lambda} \right)  \hspace{0.25em}  \left[ 
     \mathbf{Y}_{{\boldsymbol\iota}, \mathbf{x}}^{(n)}, \sum_{z
     \in \mathbb{Z}^d} L_z  \right] f_s^{\Lambda}  \right| 
     \leq (- \lambda n
     \kappa + A_n)  \hspace{0.25em} | \mathbf{Y}_{\mathbf{x}}^{(n)}
     f_s^{\Lambda} |^2 + \varepsilon \sum_{j = 1}^n  \hspace{0.25em} | Y_{J,
     x_j} \mathbf{Y}_{\mathbf{x}}^{(n)} f_s^{\Lambda} |^2, \]
  where $\kappa \equiv \inf_{\alpha \in I} \kappa_{\alpha}$, $b \equiv
  \sup_{\alpha \in I, x \in \mathbb{Z}^d} |b_{\alpha, x} |$ and
  $A_n \equiv 2 nbc |I| + n \varepsilon^{- 1} |I| + \frac{1}{2} n^2 c^2
       |I|^2  (|I| + 1)$.
\end{lemma}
\begin{lemma}[Estimate of $(\mathrm{{ \bf T_2}}$)]\label{lem3.3}
Under the assumptions of Proposition \ref{pro3.1},
  \[ \label{3.1.6} \sum_{{\boldsymbol\iota}} \left| 2 \sum_{z \in
     \Lambda} \sum_{\beta \in I}  {q}_{\beta, z}  \left( 
     \mathbf{Y}_{{\boldsymbol\iota}, \mathbf{x}}^{(n)}
     f_s^{\Lambda} \right)  \hspace{0.25em}  \left[ 
     \mathbf{Y}_{{\boldsymbol\iota}, \mathbf{x}}^{(n)}, Y_{\beta,
     z}  \right] f_s^{\Lambda}  \right| \leq 2 n \bar{q} c|I| \hspace{0.25em}
     | \mathbf{Y}_{\mathbf{x}}^{(n)} f_s^{\Lambda} |^2 \]
  with $\bar{q} \equiv \sup_{\alpha, z} || \mathbf{q}_{\alpha, z} ||_{\infty}$.
\end{lemma}
The key to the next estimate is contained in the following 
expression.  For any sufficiently smooth function $g$, we have
\begin{equation} \label{3.1.7} 
\langle\mathbf{Y}_{{\boldsymbol\iota},
   \mathbf{x}}^{(n)}, g\rangle = \sum_{l = 0}^{n-1}
   \sum_{\ensuremath{\boldsymbol{\gamma}} \subset \boldsymbol{\iota} : |
   \ensuremath{\boldsymbol{\gamma}} | = l} 
   \sum_{\ensuremath{\boldsymbol{z}} \subset \boldsymbol{x} : |
   \ensuremath{\boldsymbol{z}} | = l}
   \varphi_n(l) \left(\mathbf{Y}_{\check{\ensuremath{\boldsymbol{\gamma}}},
   \check{\mathbf{z}}}^{(n - l)} g \right) \hspace{0.25em} 
   \mathbf{Y}_{\ensuremath{\boldsymbol{\gamma}}, \mathbf{z}}^{(l)},  
\end{equation}
where 
\begin{align*}
\varphi_n(l)=
\begin{cases}
 1 & \text{if }   l\le n/2
 \\-1 & \text{otherwise,}
\end{cases} 
\end{align*}
and with the convention that $\mathbf{Y}_{\check{{\boldsymbol\gamma}},
\check{\mathbf{z}}}^{(0)} \equiv id$ and the elements of  $\ensuremath{\boldsymbol{\gamma}},\check{\ensuremath{\boldsymbol{\gamma}}}$ are ordered in the same way as in  ${\boldsymbol\iota}$ and those of  $\check{\mathbf{z}}, \check{\mathbf{z}}$ in the same way as in $ \mathbf{x}$.

\begin{lemma}[Estimate of ($\mathrm{{ \bf T_3}}$)]\label{lem3.4}
 Under the assumptions of Proposition \ref{pro3.1}, for any
  $\varepsilon \in (0, 1)$, we have
  \begin{eqnarray*}
    \sum_{{\boldsymbol\iota}}  \left| \sum_{y \in \Lambda}
    \sum_{\beta} 2 \left( \mathbf{Y}_{{\boldsymbol\iota},
    \mathbf{x}}^{(n)} f_s^{\Lambda} \right) \hspace{0.25em}  \langle
    \mathbf{Y}_{{\boldsymbol\iota}, \mathbf{x}}^{(n)},
    {q}_{\beta, y} \rangle Y_{\beta, y} f_s^{\Lambda} \right| & \leq &
    \varepsilon^{- 1} B_n | \mathbf{Y}_{\mathbf{x}}^{(n)} f_s^{\Lambda} |^2 
    \\
    &  & + \varepsilon \sum_{k = 0}^{n-1}  \sum_{\mathbf{z} \subset
    \mathbf{x}, y \in \mathbb{Z}^d} B_{\mathbf{x}, k} (\mathbf{z}, y)
    \hspace{0.25em} | \mathbf{Y}_{(\mathbf{z}, y)}^{(k + 1)} f_s^{\Lambda} |^2
  \end{eqnarray*}
  \[  \]
  with
  \[ B_n  \equiv  \sum_{y \in \mathbb{Z}^d} \sum_{\beta \in I} \sum_{k = 1}^{n} \sup_{({\boldsymbol\iota}, \mathbf{x})} 
     \sum_{(\ensuremath{\boldsymbol{\gamma}}, \mathbf{z}) \subset
     ({\boldsymbol\iota}, \mathbf{x}) : |
     \ensuremath{\boldsymbol{\gamma}} | = k} \|
     \mathbf{Y}_{\check{\ensuremath{\boldsymbol{\gamma}}},
     \check{\mathbf{z}}}^{(n - k)} {q}_{\beta, y} \|_{\infty} \]
  and
  \[ B_{\mathbf{x}, k} (\mathbf{z}, y) \equiv \delta_{\{ \mathbf{z} \subset
     \mathbf{x} : | \mathbf{z} | = k\}} \sup_{\beta,
     \check{\ensuremath{\boldsymbol{\gamma}}}} \|
     \mathbf{Y}_{\check{\ensuremath{\boldsymbol{\gamma}}},
     \check{\mathbf{z}}}^{(n - k)} {q}_{\beta, y} \|_{\infty}. \]
\end{lemma}
We remark that when we consider an interaction with finite range $R \in
\mathbb{N}$, i.e. when $\mathbf{q}_{\beta, y}$ is a cylinder function
dependent only on coordinates $\omega_z$ with
$\ensuremath{\operatorname{dist}} (z, y) < R$, we have
\[ \mathbf{Y}_{\check{\ensuremath{\boldsymbol{\gamma}}},
   \check{\mathbf{z}}}^{(n - k)} {q}_{\beta, y} = 0, \hspace{2em}
   \ensuremath{\operatorname{dist}} \left( \check{\mathbf{z}}, y \right) \geq
   R. \]

The next estimate uses our locality assumption together with the
quadratic Young's inequality.

\begin{lemma}[Estimate of $(\mathrm{{ \bf T_4}})$]\label{lem3.5} Under the assumptions of Proposition \ref{pro3.1}, for any
  $\varepsilon \in (0, 1)$ \label{3.1.9a}
\begin{align*}
    &\sum_{\boldsymbol\iota} \left|\sum_{yy'\in\Lambda}\sum_{{\gamma} {\gamma}'\in J}2\mathfrak{S}_{{\gamma}{\gamma}',yy'}\left(\mathbf{Y}_{{\boldsymbol\iota},\mathbf{x}}^{(n)}f_s^{\Lambda}\right){\hspace{0.25em}}\left(\left[\mathbf{Y}_{{\boldsymbol\iota},\mathbf{x}}^{(n)},Y_{{\gamma},y}\right]Y_{{\gamma}',y'}
    +Y_{{\gamma},y} \left[\mathbf{Y}_{{\boldsymbol\iota},\mathbf{x}}^{(n)},Y_{{\gamma}',y'}\right]\right)f_s^{\Lambda}\right|\\
    &{\leq}C_n\left|\mathbf{Y}_{\mathbf{x}}^{(n)}f_s^{\Lambda}\right|^2 +{\varepsilon} \sum_{y\in\Lambda}\left|Y_{J,y} \mathbf{Y}_{\mathbf{x}}^{(n)}
f_s^{\Lambda}\right|^2    
  \end{align*}
  with positive constants
  \[ \label{3.1.9b} C_n \leq \varepsilon^{- 1}  n^3 c|I|
     \hspace{0.25em} \sup_{z \in \mathbb{Z}^d}  \sum_{y \in \mathbb{Z}^d} 
     \sum_{\gamma \gamma' \in J} (| \mathfrak{S}_{\gamma' \gamma, yz} | + | \mathfrak{S}_{\gamma \gamma',zy} |) + \frac{1}{2} n^2 \bar{C}_n (|I|^2 + 1),
  \]
where $\bar{C}_n=2c \,\sup_{\gamma\in I, y,z\in \ZZ^d}  \sum_{\gamma'\in J} \lv \s_{\gamma' \gamma, yz}\rv $.   
\end{lemma}

The last estimate is similar to that of ($\mathrm{{ \bf T_3}}$).

\begin{lemma}[Estimate of ($\mathrm{{ \bf T_5}}$)] \label{lem3.6} 
 Under the assumptions of Proposition \ref{pro3.1}, for any $n\geq 1$, for every $ l=1\dd n$, for any $\x=(x_1 \dd x_n)\subset \ZZ^d$ and for every $\mathbf{z} \subset \mathbf{x}$  there exist  positive constants   $\mathcal{D}_{\mathbf{x}, n}^{(l)} (\mathbf{z}, y)
  \in (0, \infty)$,  satisfying
  \[ \sup_{| \mathbf{x} | = n, \mathbf{z} \subset \mathbf{x}}  \sum_{y \in
     \mathbb{Z}^d} \mathcal{D}_{\mathbf{x}, n}^{(l)} (\mathbf{z}, y) < \infty
  \]
  such that for any $\varepsilon \in (0, 1)$ the following bound is true
  \begin{align*}
    &\sum_{{\boldsymbol\iota}} \left| 2 \left(
    \mathbf{Y}_{{\boldsymbol\iota}, \mathbf{x}}^{(n)}
    f_s^{\Lambda} \right) \sum_{yy' \in \Lambda}
     \sum_{\gamma, \gamma' \in J}
     \langle \mathbf{Y}_{{\boldsymbol\iota},
    \mathbf{x}}^{(n)}, \mathfrak{S}_{\gamma \gamma', yy'} \rangle 
    Y_{\gamma, y} Y_{\gamma', y'} f_s^{\Lambda} \right|
    \\
   & \leq \varepsilon \sum_{l = 1}^n \sum_{y \in \Lambda} \sum_{\mathbf{z}
    \subset \mathbf{x} : | \mathbf{z} | = l} \mathcal{D}_{\mathbf{x}, n}^{(l)}
    (\mathbf{z}, y) \left| Y_{J, y} \mathbf{Y}_{\mathbf{z}}^{(l)}
    f_s^{\Lambda} \right|^2 + \sum_{l = 1}^n
    \sum_{\mathbf{z}\subset\mathbf{x}:|\mathbf{z}|=l} D_{\mathbf{x}, n}^{(l)} (\mathbf{z})
    \hspace{0.25em} \left| \mathbf{Y}_{\mathbf{z}}^{(l)} f_s^{\Lambda}
    \right|^2
  \end{align*}
   for  some $D_{\mathbf{x},
  n}^{(l)} (\mathbf{z}) \in (0, \infty)$, $l = 1, .., n, \hspace{0.25em}
  \mathbf{z} \subset \mathbf{x}$.
\end{lemma}
We remark that because of our strong assumption of locality of 
$\mathfrak{S}_{\gamma \gamma', yy'}$, we have
\[ \mathbf{Y}_{\check{{\boldsymbol{\beta}}},
   \check{\mathbf{z}}}^{(n - l)}  \mathfrak{S}_{\gamma \gamma',yy'} = 0, \hspace{2em}
   {if}\hspace{2em} \check{\mathbf{z}} \subseteq \{ yy' \}. \]
Now we combine all estimates of Lemma \ref{lem3.2} to Lemma \ref{lem3.6}, i.e. all the estimates of ($\mathrm{{ \bf T_1}}$)-($\mathrm{{ \bf T_5}}$):
\begin{align}
  &
  \left|2\left(\mathbf{Y}_{{\boldsymbol\iota},\mathbf{x}}^{(n)}f_s^{\Lambda}\right){\hspace{0.25em}}\left[\mathbf{Y}_{{\boldsymbol\iota},\mathbf{x}}^{(n)},\mathcal{L}_{\Lambda}\right] f_s^{\Lambda}\right| \nonumber\\
  &
  {\leq}
  (-{\lambda} n {\kappa}+A_n){\hspace{0.25em}}|\mathbf{Y}_{\mathbf{x}}^{(n)}f_s^{\Lambda}|^2+{\varepsilon} \sum\ensuremath{_{\textrm{j=1}}}^n {\hspace{0.25em}}|Y_{J,x_{j}}Y_{\mathbf{x}}^{(n)} f_s^{\Lambda}|^2\nonumber\\
  &
  +2 n \bar{q} c|I|{\hspace{0.25em}}|\mathbf{Y}_{\mathbf{x}}^{(n)}f_s^{\Lambda}|^2 \nonumber\\
  &+{\varepsilon}^{-1}B_n|\mathbf{Y}_{\mathbf{x}}^{(n)}f_s^{\Lambda}|^2
  +{\varepsilon} \sum_{k=0}^{n-1}
  \sum_{{\mathbf{z}{\subset}\mathbf{x},y{\in}{\Lambda}} \atop {|\mathbf{z}|=k}}
  B_{\mathbf{x},k} (\mathbf{z},y) {\hspace{0.25em}}|\mathbf{Y}_{(\mathbf{z},y)}^{(k+1)} f_s^{\Lambda}|^2\nonumber \\
  &+C_n\left|\mathbf{Y}_{\mathbf{x}}^{(n)}f_s^{\Lambda}\right|^2
  +{\varepsilon}n^2 \bar{C}_n \sum_{y\in\Lambda}\left|Y_{J,y}\mathbf{Y}_{\mathbf{x}}^{(n)}f_s^{\Lambda}\right|^2
  \nonumber\\
  &
  +{\varepsilon} \sum_{l=1}^n\sum_{y\in\Lambda}
  \sum_{\mathbf{z}{\subset}\mathbf{x}:|\mathbf{z}|=l} \mathcal{D}_{\mathbf{x},n}^{(l)}(\mathbf{z},y)\left|Y_{J,y} \mathbf{Y}_{\mathbf{z}}^{(l)}f_s^{\Lambda}\right|^2+\sum_{l=1}^n
  \sum_{{\mathbf{z}{\subset}\mathbf{x}}\atop{|\mathbf{z}|=l}} D_{\mathbf{x},n}^{(l)}(\mathbf{z}){\hspace{0.25em}} \left|\mathbf{Y}_{\mathbf{z}}^{(l)}f_s^{\Lambda}\right|^2. \label{77}
\end{align}
This can be rewritten as follows:
\begin{align*}
  \left|2\left(\mathbf{Y}_{{\boldsymbol\iota},\mathbf{x}}^{(n)}f_s^{\Lambda}\right)
  {\hspace{0.25em}}\left[\mathbf{Y}_{{\boldsymbol\iota},\mathbf{x}}^{(n)},\mathcal{L}_{\Lambda}\right] f_s^{\Lambda}\right|{\hspace{0.25em}} &{\leq}{\hspace{0.25em}}{\varepsilon} \sum_{l=1}^n\sum_{y\in\Lambda}\sum_{\mathbf{z}{\subset}\mathbf{x}:|\mathbf{z}|=l}\mathcal{A}_{\mathbf{x},n}^{(l)}(\mathbf{z},y)
\left|Y_{J,y}\mathbf{Y}_{\mathbf{z}}^{(l)}f_s^{\Lambda}\right|^2\\
  &
  +\left(-{\lambda} n {\kappa}+A_n\right){\hspace{0.25em}}|\mathbf{Y}_{\mathbf{x}}^{(n)}f_s^{\Lambda}|^2+\sum_{l=1}^{n-1}\sum_{{\mathbf{z}{\subset}\mathbf{x}}\atop{|\mathbf{z}|=l}}
  \mathcal{B}_{\mathbf{x},n}^{(l)}(\mathbf{z}){\hspace{0.25em}}\left|\mathbf{Y}_{\mathbf{z}}^{(l)}f_s^{\Lambda}\right|^2\\
  & +{\varepsilon} \sum_{k=0}^{n-1}
  \sum_{{\mathbf{z}{\subset}\mathbf{x},y{\in}{\Lambda}} \atop{|\mathbf{z}|=k}} B_{\mathbf{x},k}(\mathbf{z},y){\hspace{0.25em}}|\mathbf{Y}_{(\mathbf{z},y)}^{(k+1)}f_s^{\Lambda}|^2,
\end{align*}
where  $\mathcal{A}_{\mathbf{x},n}^{(l)}(\mathbf{z},y)$ and $\mathcal{B}_{\mathbf{x},n}^{(l)}(\mathbf{z})$ are positive constants depending on the constants appearing in \eqref{77}.  
Putting this together with \eqref{63} and \eqref{64}, we obtain
\begin{align*}
  {\frac{\partial}{\partial s}}\mathcal{P}_{t-s}^{\Lambda}\left|\mathbf{Y}_{\mathbf{x}}^{(n)}f_s^{\Lambda}\right|^2
  & {\leq}\mathcal{P}_{t-s}^{\Lambda}\left\{-2 \sum_{z\in\mathbb{Z}^d}\left|Y_{J,z}\mathbf{Y}_{\mathbf{x}}^{(n)}f_s^{\Lambda}\right|^2-2 \sum_{z,z'\in\Lambda}\mathfrak{S}_{zz'}{\cdot}\sum_{\boldsymbol\iota}\left(Y_{z}\mathbf{Y}_{{\boldsymbol\iota},\mathbf{x}}^{(n)}f_s^{\Lambda}\right){\cdot}\left(Y_{z'}
  \mathbf{Y}_{{\boldsymbol\iota},\mathbf{x}}^{(n)}f_s^{\Lambda}\right)\right\}\\
  &
  +\mathcal{P}_{t-s}^{\Lambda}\left\{{\varepsilon} \sum_{l=1}^n\sum_{y\in\Lambda}\sum_{\mathbf{z}{\subset}\mathbf{x}:|\mathbf{z}|=l}\mathcal{A}_{\mathbf{x},n}^{(l)}(\mathbf{z},y)\left|Y_{J,y}\mathbf{Y}_{\mathbf{z}}^{(l)}f_s^{\Lambda}\right|^2{\phantom{A*A*A*A*A*A*A}}\right.\\
  & \phantom{A*A*A}  +\left(-{\lambda} n {\kappa}+\mathbf{A}_n\right){\hspace{0.25em}}|\mathbf{Y}_{\mathbf{x}}^{(n)}f_s^{\Lambda}|^2+\sum_{l=1}^{n-1}\sum_{{\mathbf{z}{\subset}\mathbf{x}}\atop{|\mathbf{z}|=l}}\mathcal{B}_{\mathbf{x},n}^{(l)}(\mathbf{z}){\hspace{0.25em}}\left|\mathbf{Y}_{\mathbf{z}}^{(l)}f_s^{\Lambda}\right|^2\\
  & \phantom{A*A*A} \left.+{\varepsilon} \sum_{l=0}^{n-1}\sum_{{\mathbf{z}{\subset}\mathbf{x},y}\atop{|\mathbf{z}|=l}}B_{\mathbf{x},k}^{(l)}(\mathbf{z},y){\hspace{0.25em}}|\mathbf{Y}_{(\mathbf{z},y)}^{(l+1)}f_s^{\Lambda}|^2\right\}.
\end{align*}
Assuming that for some  $\delta \in (0, 1)$  we have  $\mathfrak{S}_{zz'} \leq
\delta \,\textup{Id}$ in the sense of quadratic forms, we can simplify the above as
follows.
\begin{align*} 
{\frac{\partial}{\partial s}}\mathcal{P}_{t-s}^{\Lambda}\left|\mathbf{Y}_{\mathbf{x}}^{(n)}f_s^{\Lambda}\right|^2
  &
  {\leq}\mathcal{P}_{t-s}^{\Lambda}\left\{-2 (1-{\delta}) \sum_{z\in\mathbb{Z}^d}\left|Y_{J,z}\mathbf{Y}_{\mathbf{x}}^{(n)}f_s^{\Lambda}\right|^2\right\}\\
  &
  +\mathcal{P}_{t-s}^{\Lambda}\left\{{\varepsilon} \sum_{l=1}^n\sum_{y\in\Lambda}\sum_{\mathbf{z}{\subset}\mathbf{x}:|\mathbf{z}|=l}\mathcal{A}_{\mathbf{x},n}^{(l)}(\mathbf{z},y)\left|Y_{J,y}\mathbf{Y}_{\mathbf{z}}^{(l)}f_s^{\Lambda}\right|^2{\phantom{A*A*A*A*A*A*A}}\right.\\
  & {\phantom{A*A*A}}+
  \left(-{\lambda} n {\kappa}+A_n\right){\hspace{0.25em}}|\mathbf{Y}_{\mathbf{x}}^{(n)}f_s^{\Lambda}|^2+\sum_{l=1}^{n-1}\sum_{{\mathbf{z}{\subset}\mathbf{x}} \atop{|\mathbf{z}|=l}}\mathcal{B}_{\mathbf{x},n}^{(l)}(\mathbf{z}){\hspace{0.25em}}\left|\mathbf{Y}_{\mathbf{z}}^{(l)}f_s^{\Lambda}\right|^2\\
  &
  \left.{\phantom{A*A*A}}+{\varepsilon} \sum_{l=0}^{n-1}
  \sum_{{\mathbf{z}{\subset}\mathbf{x},y}\atop{|\mathbf{z}|=l}} B_{\mathbf{x},n}^{(l)}(\mathbf{z},y){\hspace{0.25em}}
  |\mathbf{Y}_{(\mathbf{z},y)}^{(l+1)}f_s^{\Lambda}|^2\right\} . 
\end{align*}
Choosing $\varepsilon$ so that $\varepsilon \sup_{\mathbf{x}, y, l} 
\mathcal{A}_{\mathbf{x}, n}^{(n)} (\mathbf{z}, y) < 2 (1 - \delta)$, we get
the following bound
\begin{align*}
{\frac{\partial}{\partial s}}\mathcal{P}_{t-s}^{\Lambda}\left|\mathbf{Y}_{\mathbf{x}}^{(n)}f_s^{\Lambda}\right|^2
  &
  {\leq}\mathcal{P}_{t-s}^{\Lambda}\left\{\left(-{\lambda}n{\kappa}+A_n\right){\hspace{0.25em}}|\mathbf{Y}_{\mathbf{x}}^{(n)}f_s^{\Lambda}|^2+\sum_{l=1}^{n-1}\sum_{{\mathbf{z}{\subset}\mathbf{x}}\atop{|\mathbf{z}|=l}}\mathcal{B}_{\mathbf{x},n}^{(l)}(\mathbf{z}){\hspace{0.25em}}\left|\mathbf{Y}_{\mathbf{z}}^{(l)}f_s^{\Lambda}\right|^2\right.\\
  &
  \left.+{\varepsilon} \sum_{l=1}^{n-1}\sum_{y\in\Lambda}\sum_{\mathbf{z}{\subset}\mathbf{x}:|\mathbf{z}|=l}\mathcal{A}_{\mathbf{x},n}^{(l)}(\mathbf{z},y)\left|Y_{J,y}\mathbf{Y}_{\mathbf{z}}^{(l)}f_s^{\Lambda}\right|^2+{\varepsilon} \sum_{l=0}^{n-1}\sum_{{\mathbf{z}{\subset}\mathbf{x},y{\in}{\Lambda}}\atop{|\mathbf{z}|=l}} B_{\mathbf{x},n}^{(l)}(\mathbf{z},y){\hspace{0.25em}}|\mathbf{Y}_{(\mathbf{z},y)}^{(l+1)}
f_s^{\Lambda}|^2\right\} .
\end{align*}
Setting 
\be\label{vn}
 \mathbf{v}_n \equiv \left( - \lambda n \kappa + A_n \right), 
 \ee
 we obtain the statement of Proposition \ref{pro3.1}.
\end{proof}
\subsection*{A.2 Proof of Theorem \ref{thm:smooth}}
Let us  define 
\[ \mathbf{Q}_t^{(1)} (f) \equiv \ensuremath{\boldsymbol{\Gamma}}_t^{(1)} (f)
   + d \lv f\rv^2 \]
and, for $n>1$,
\[ \mathbf{Q}_t^{(n)} (g) \equiv \ensuremath{\boldsymbol{\Gamma}}^{(n)}_t  (g)
   + \varsigma_n \mathbf{Q}_t^{(n - 1)} (g) ,\]
with some $\varsigma_n>0$ to be chosen later. 
\proof[Proof of Theorem \ref{thm:smooth}]
We begin with an estimate for the case $n = 1$. 
For $f_t \equiv \PP_t f \equiv e^{t \mathcal{L}} f$, we have
\begin{align*}
  \partial_s& \PP_{t - s} \mathbf{Q}_s^{(1)} (f_s)  =   \PP_{t - s} (\partial_s -
  \mathcal{L}_{}) \mathbf{Q}_s^{(1)} (f_s)
  \\&-   2 \sum_{x \in \mathbb{Z}^d}  \sum_{i = 0, \ldots, N} \PP_{t - s} (a_i
  s^{2 i_{} + 1} \mathcal{E}_\mathcal{L} (Z_{i, x} f_s)  + b_{i + 1} s^{2 i_{}}
  \mathcal{E_\mathcal{L}} (Z_{i - 1, x} f_s, Z_{i, x} f_s)) - 2 d\mathcal{E_\mathcal{L}} (f_s)\\
  &    + \sum_{x \in \mathbb{Z}^d}  \sum_{i = 0, \ldots, N} \PP_{t - s} (2 a_i
  s^{2 i_{} + 1} (Z_{i, x} f_s \cdot [Z_{i, x}, \mathcal{L}_{}] f_s) + b_i
  s^{2 i_{}} ([Z_{i - 1, x}, \mathcal{L}] f_s \cdot Z_{i, x} f_s + Z_{i - 1,
  x} f_s \cdot [Z_{i, x}, \mathcal{L}_{}] f_s) )\\
  &    + \sum_{x \in \mathbb{Z}^d}  \sum_{i = 0, \ldots, N} \PP_{t - s} ((2
  i_{} + 1) a_i s^{2 i_{}} | Z_{i, x} f_s |^2  + 2 i_{} b_i s^{2 i_{} - 1}
  Z_{i - 1, x} f_s \cdot Z_{i, x} f_s)
  \\  &\equiv (\mathbf{I}),
\end{align*}
where $\mathcal{E} _\mathcal{L}(V, W) \equiv {\frac{1}{2}} (\mathcal{L}_{}
(\ensuremath{\operatorname{VW}}) - V \mathcal{L}_{} W - (\mathcal{L} V)_{} W)$
and $\mathcal{E_\mathcal{L}} (V) \equiv \mathcal{E_\mathcal{L}} (V, V)$ and $Z_{-1}\equiv 0$.  The right-hand side can be written as  \[(\mathbf{I}) = (\mathbf{II})+(\mathbf{III}),\]
with
\begin{align*}
(&\mathbf{II})=  - 2 \sum_{x \in \mathbb{Z}^d}  \sum_{i = 0, \ldots, N} \PP_{t - s} (a_i s^{2
  i_{} + 1} \mathcal{E}_{L_x} (Z_{i, x} f_s)  + b_i s^{2 i_{}}
  \mathcal{E}_{L_x} (Z_{i - 1, x} f_s, Z_{i, x} f_s)) - 2 d\mathcal{E}_{L_x}
  (f_s)\\
  &+\sum_{x \in \mathbb{Z}^d}  \sum_{i = 0, \ldots, N} \PP_{t - s} (2 a_i s^{2
  i_{} + 1} (Z_{i, x} f_s \cdot [Z_{i, x}, L_x] f_s)  + 2 b_i s^{2 i_{}}
  ([Z_{i - 1, x}, L_x] f_s \cdot Z_{i, x} f_s + Z_{i - 1, x} f_s \cdot [Z_{i,
  x}, L_x] f_s) ) 
 \\ &+ \sum_{x \in \mathbb{Z}^d}  \sum_{i = 0, \ldots, N} \PP_{t - s} ((2 i_{} + 1)
  a_i s^{2 i_{}} | Z_{i, x} f_s |^2  + 2 i_{} b_i s^{2 i_{} - 1} Z_{i - 1, x}
  f_s \cdot Z_{i, x} f_s) ,
\end{align*}
and 
\begin{align}
(&\mathbf{III})= - 2 \sum_{x, y \in \mathbb{Z}^d}  \sum_{i = 0, \ldots, N} \PP_{t - s} (a_i   s^{2 i_{} + 1} \mathcal{E}_{L_y} (Z_{i, x} f_s)  + b_i s^{2 i} \mathcal{E}_{L_y} (Z_{i - 1, x} f_s, Z_{i, x} f_s))  \label{eq3} 
\\& - 2 \sum_{y, y' \in \mathbb{Z}^d} 
  \mathfrak{S}_{\ensuremath{\operatorname{yy}}'} \cdot \sum_{x, y \in
  \mathbb{Z}^d}  \sum_{i = 0, \ldots, N} \PP_{t - s} a_i s^{2 i_{} + 1} 
  (Z_{0, y} Z_{i, x} f_s) \cdot
  (Z_{0, y'} Z_{i, x} f_s)   \label{eq4}
\\&  - 2 \sum_{y, y' \in \mathbb{Z}^d} 
  \mathfrak{S}_{\ensuremath{\operatorname{yy}}'} \cdot \sum_{x, y \in
  \mathbb{Z}^d}  \sum_{i = 0, \ldots, N} \PP_{t - s} b_i s^{2 i_{}}
  (Z_{0, y} Z_{i - 1, x} f_s) \cdot
  (Z_{0, y'} Z_{i, x} f_s)  \label{eq5}
  \\& + \sum_{x, y \in \mathbb{Z}^d}  \sum_{i = 0, \ldots, N} \PP_{t - s} \left( 2
  a_i s^{2 i_{} + 1} \left( Z_{i, x} f_s \cdot \left[ Z_{i, x}, \mathbf{q}_y
  \cdot \ensuremath{\boldsymbol{Z}}_y \right] f_s \right)  \right)  \label{eq6}
  \\& + \sum_{x, y \in \mathbb{Z}^d}  \sum_{i = 0, \ldots, N} \PP_{t - s} \left( 2
  b_i s^{2 i_{}} \left( \left[ Z_{i - 1, x}, \mathbf{q}_y \cdot
  \ensuremath{\boldsymbol{Z}}_y \right] f_s \cdot Z_{i, x} f_s + Z_{i - 1, x}
  f_s \cdot \left[ Z_{i, x}, \mathbf{q}_y \cdot \ensuremath{\boldsymbol{Z}}_y
  \right] f_s \right)  \right)   \label{eq7}
  \\& + \sum_{y, y' \in \mathbb{Z}^d} \sum_{x \in \mathbb{Z}^d}  \sum_{i = 0,
  \ldots, N} \PP_{t - s} (2 a_i s^{2 i_{} + 1} (Z_{i, x} f_s \cdot [Z_{i, x},
  \mathfrak{S}_{\ensuremath{\operatorname{yy}}'} \cdot
  \ensuremath{\boldsymbol{Z}}_{0, y} \ensuremath{\boldsymbol{Z}}_{0, y'}] f_s)
  )  \label{eq8}
  \\& + \sum_{x \in \mathbb{Z}^d} \sum_{y, y' \in \mathbb{Z}^d}  \sum_{i = 0,
  \ldots, N} \PP_{t - s} (2 b_i s^{2 i_{}} ([Z_{i - 1, x},
  \mathfrak{S}_{\ensuremath{\operatorname{yy}}'} \cdot
  Z_{0, y} Z_{0, y'}] f_s
  \cdot Z_{i, x} f_s \\ & \qquad \qquad\qquad\qquad\qquad\qquad\qquad\qquad\qquad\qquad+ Z_{i - 1, x} f_s \cdot [Z_{i, x},
  \mathfrak{S}_{\ensuremath{\operatorname{yy}}'} \cdot
  Z_{0, y} Z_{0, y'}] f_s)
  ). \label{eq9}
\end{align}
We have studied \((\mathbf{II})\) (called the `free part' later on) before.
From our assumptions about the free part,  \eqref{eq3} is strictly negative and can be
made so that it dominates contributions from \eqref{eq4}-\eqref{eq5}. The contributions from \eqref{eq6}-\eqref{eq9} can
not be dominated by the free part without additional assumptions about the
interaction which we will discuss in the following. First of all we remark
that
\begin{equation}
  \left[ Z_{i, x}, \mathbf{q}_y \cdot \ensuremath{\boldsymbol{Z}}_y \right] =
  \sum_{j = 1, \ldots, N} \left( Z_{i, x} 
  q_{jy} \right) Z_{j, y} +
  \delta_{xy} \sum_{j, k = 1, \ldots, N}
  c_{ijk}
  q_{jx} Z_{k, x} .  \nonumber
\end{equation}
Assumptions \eqref{si1a} and \eqref{si2a} enable us to dominate the terms involving these commutators by the
free part for small times.  
 The newly generated terms will come accompanied by a sufficiently
high power of $s$  so they will be irrelevant for small times. Next we note that
\begin{eqnarray}
  {}[Z_{i, x}, \mathfrak{S}_{\ensuremath{\operatorname{yy}}'} \cdot
  Z_{0, y} Z_{0, y'}] & =
  & (Z_{i, x}  \mathfrak{S}_{yy'}) \cdot
  Z_{0, y} Z_{0, y'} \nonumber\\
  &  & + \delta_{xy} \sum_{k = 1, \ldots, N} c_{i
  0 k} \mathfrak{S}_{xy'} Z_{0, y'} Z_{k, x}   \nonumber\\
  &  & + \delta_{xy}
  \delta_{xy'} \sum_{k, l = 1, \ldots, N} c_{i 0
  k} c_{k 0 l}  \mathfrak{S}_{yx} Z_{l, x} \nonumber\\
  &  & + \delta_{xy'} \sum_{k = 1, \ldots, N}
  c_{i 0 k} \mathfrak{S}_{yx} Z_{0, y} Z_{k, x} . 
 \nonumber 
\end{eqnarray}
For the corresponding terms to be dominated by the free part it is sufficient that
 \eqref{si3a} and \eqref{si4a} hold, 
because in this case the terms will be accompanied by a sufficiently large power
of $s$. (One can in fact see that \eqref{si3a} implies \eqref{si4a}.) Thus, under the conditions \eqref{si1a}
-\eqref{si3a}, for sufficiently small time we have
\[ \partial_s \PP_{t - s} \mathbf{Q}_s^{(1)} (f_s) \leq 0. \]
Hence, we arrive to the following smoothing estimates for infinite dimensional
system, when \ $n = 1$
\[ \ensuremath{\boldsymbol{\Gamma}}_t^{(1)} (f_t) \leq d (\PP_t f^2 - (\PP_t f)^2)
   . \]
Now we proceed by induction.  We have
\begin{eqnarray*}
  (\partial_s - \mathcal{L}) \ensuremath{\boldsymbol{\Gamma}}^{(n)}_{s,
  \mathbf{x}}  (f_s ) & = & \left( \partial_s
  \ensuremath{\boldsymbol{\Gamma}}^{(n)}_{s, \mathbf{x}}  \right)
   (f_s ) -
  \left[ \ensuremath{\boldsymbol{L}},
  \ensuremath{\boldsymbol{\Gamma}}^{(n)}_{s, \mathbf{x}} \right]_{} 
  (f_s ) -
  \left[ \ensuremath{\boldsymbol{L}}_{\ensuremath{\operatorname{int}}},
  \ensuremath{\boldsymbol{\Gamma}}^{(n)}_{s, \mathbf{x}} \right]_{} 
  (f_s ) 
\end{eqnarray*}
where
\[ \left[ \ensuremath{\boldsymbol{L}},
   \ensuremath{\boldsymbol{\Gamma}}^{(n)}_{s, \mathbf{x}} \right]_{}
    (g, h)
   \equiv \ensuremath{\boldsymbol{L}} \left(
   \ensuremath{\boldsymbol{\Gamma}}^{(n)}_{s, \mathbf{x}} (g, h) \right)
   -\ensuremath{\boldsymbol{\Gamma}}^{(n)}_{s, \mathbf{x}}
   (\ensuremath{\boldsymbol{L}}g, h)
   -\ensuremath{\boldsymbol{\Gamma}}^{(n)}_{s, \mathbf{x}} (g,
   \ensuremath{\boldsymbol{L}}h) \]
and similarly for the second commutator involving 
$\ensuremath{\boldsymbol{L}}_{\ensuremath{\operatorname{int}}} $.

Analogously to the proof of Theorem \ref{thm.1}, let us start with assuming for simplicity that $Z\ensuremath{_{\textrm{0,}y}}$ fields commute with
all the other $Z_{\alpha, x}$ fields; in this case the  terms which will appear on the
$n^{th}$ level  will be as follows
\begin{align}
 & - 2 \sum_{y \in \mathbb{Z}^d} \sum_{ \mathbf{x} \in
  \mathbb{Z}^{nd}}  \sum_{| \mathbf{k} |_n =
  0}^{nN} a_{\mathbf{k}, n} s^{2 | \mathbf{k} |_n + n} \mathcal{E}_{L_y}
  \left( \mathbf{Z}_{\mathbf{k}, n, \mathbf{x}} f_s  \right)
        \label{eq19} \\
   &- 2 \sum_{y \in \mathbb{Z}^d} \sum_{ \mathbf{x} \in
  \mathbb{Z}^{nd}}  \sum_{0 \leq | \mathbf{k} |_n
  : k_1 \geq 1}^{nN} b_{\mathbf{k}, n} s^{2 | \mathbf{k} |_n + n - 1}
  \mathcal{E}_{L_y} \left( \mathbf{Z}_{\mathbf{k} - \mathbf{e}_1, n,
  \mathbf{x}} f_s, \mathbf{Z}_{\mathbf{k}, n, \mathbf{x}} f_s \right)   
    \label{eq20} \\
  &- 2 \sum_{y, y' \in \mathbb{Z}^d} 
  \mathfrak{S}_{yy'} \cdot \sum_{ \mathbf{x} \in
 \mathbb{Z}^{nd}}  \sum_{| \mathbf{k} |_n =
  0}^{nN} a_{\mathbf{k}, n} s^{2 | \mathbf{k} |_n + n} Z_{0, y}
  \mathbf{Z}_{\mathbf{k}, n, \mathbf{x}} f_s \cdot Z_{0, y'}
  \mathbf{Z}_{\mathbf{k}, n, \mathbf{x}} f_s      \label{eq21}\\
  &- 2 \sum_{y, y' \in \mathbb{Z}^d} 
  \mathfrak{S}_{yy'} \sum_{ \mathbf{x} \in
  \mathbb{Z}^{d}} \sum_{0 \leq | \mathbf{k} |_n :
  k_1 \geq 1}^{nN} b_{\mathbf{k}, n}
   s^{2 | \mathbf{k} |_n + n - 1}     
 \cdot Z_{0, y} \mathbf{Z}_{\mathbf{k} - \mathbf{e}_1, n, \mathbf{x}} f_s
  \cdot Z_{0, y'} \mathbf{Z}_{\mathbf{k}, n, \mathbf{x}} f_s    \label{eq22}\\
  &+ 2 \sum_{y \in \mathbb{Z}^d} \sum_{ \mathbf{x} \in
  \mathbb{Z}^{nd}}  \sum_{| \mathbf{k} |_n =
  0}^{nN} a_{\mathbf{k}, n} s^{2 | \mathbf{k} |_n + n} 
  \mathbf{Z}_{\mathbf{k}, n, \mathbf{x}} f_s \cdot \left[
  \mathbf{Z}_{\mathbf{k}, n, \mathbf{x}}, \ensuremath{\boldsymbol{q}}_y 
  \cdot
  \ensuremath{\boldsymbol{Z}}_y \right] f_s     \label{eq23} \\
  &+ 2 \sum_{y \in \mathbb{Z}^d} \sum_{ \mathbf{x} \in
  \mathbb{Z}^{nd}} \sum_{0 \leq | \mathbf{k} |_n :
  k_1 \geq 1}^{nN} b_{\mathbf{k}, n} s^{2 | \mathbf{k} |_n + n - 1} 
  \left[ \mathbf{Z}_{\mathbf{k} - \mathbf{e}_1, n, \mathbf{x}}
 \ensuremath{\boldsymbol{q}}_y \cdot \ensuremath{\boldsymbol{Z}}_y \right]
 f_s \cdot \mathbf{Z}_{\mathbf{k}, n, \mathbf{x}} f_s 
 \\& \qquad \qquad \qquad \qquad \qquad \qquad \qquad \qquad \qquad 
 \qquad +
  \mathbf{Z}_{\mathbf{k} - \mathbf{e}_1, n, \mathbf{x}} f_s \cdot \left[
  \mathbf{Z}_{\mathbf{k}, n, \mathbf{x}}, \ensuremath{\boldsymbol{q}}_y \cdot
  \ensuremath{\boldsymbol{Z}}_y \right] f_s   .    \label{eq24}
      \end{align}
Under the conditions for which the original finite dimensional case is
negative,  \eqref{eq19}-\eqref{eq20} are also negative and, if
$\mathfrak{S}_{yy'}$ is assumed sufficiently
small, they can dominate contributions from \eqref{eq21} and \eqref{eq22}. As discussed before one has
the following expressions for the commutators in  \eqref{eq23} and \eqref{eq24}
\[ \left[ \mathbf{Z}_{\mathbf{k}, n, \mathbf{x}},
   \ensuremath{\boldsymbol{q}}_y \cdot \ensuremath{\boldsymbol{Z}}_y 
   \right] =
   \sum_{i = 1}^n \left( \sum_{l = 1}^{n - 1} \sum_{\widehat{\mathbf{z}}
   \subset \mathbf{z}, \widehat{\mathbf{k}} \subset \mathbf{k} : \atop \left|
   \mathbf{x} \setminus \widehat{\mathbf{z}} \right| = l} \left(
   \mathbf{Z}_{\mathbf{k} \setminus \widehat{\mathbf{k}}, n - l, \mathbf{x}
   \setminus \widehat{\mathbf{z}}}
   q_{iy} \right) \cdot
   \mathbf{Z}_{\widehat{\mathbf{k}}, l, \widehat{\mathbf{z}}}
   \ensuremath{Z}_{i, y}
   +\ensuremath{q}_{iy} \cdot \left[
   \mathbf{Z}_{\mathbf{k}, n, \mathbf{x}}, 
   Z_{i, y}
   \right] \right) \]
and
\[ \left[ \mathbf{Z}_{\mathbf{k} - \mathbf{e}_1, n, \mathbf{x}},
   \ensuremath{\boldsymbol{q}}_y \cdot \ensuremath{\boldsymbol{Z}}_y
    \right] =
   \sum_{i = 1}^n \left( \sum_{l = 1}^{n - 1} \sum_{\widehat{\mathbf{z}}
   \subset \mathbf{z}, \widehat{\mathbf{k}} \subset \mathbf{k} - \mathbf{e}_1
   \atop \left| \mathbf{x} \setminus \widehat{\mathbf{z}} \right| = l} \left(
   \mathbf{Z}_{\mathbf{k} - \mathbf{e}_1 
   \setminus \widehat{\mathbf{k}}, n -
   l, \mathbf{x} \setminus \widehat{\mathbf{z}}}
   q_{\ensuremath{iy}} \right) \cdot
   \mathbf{Z}_{\widehat{\mathbf{k}}, l, \widehat{\mathbf{z}}}
   Z_{i, y}
   +q_{\ensuremath{iy}} \cdot \left[
   \mathbf{Z}_{\mathbf{k} - \mathbf{e}_1, n, \mathbf{x}},
   Z_{i, y} \right] \right), \]
with a rule that \ $\mathbf{Z}_{\mathbf{k} \setminus \widehat{\mathbf{k}}, n -
l, \mathbf{x} \setminus \widehat{\mathbf{z}}} \ensuremath{\boldsymbol{q}}_y$,
$\mathbf{Z}_{\mathbf{k} - \mathbf{e}_1 \setminus \widehat{\mathbf{k}}, n - l,
\mathbf{x} \setminus \widehat{\mathbf{z}}} \ensuremath{\boldsymbol{q}}_y \neq
0$ only in the case when $\ensuremath{\operatorname{dist}} (x_i, y) \leq R$
for each $x_i \in \mathbf{x} \setminus \widehat{\mathbf{z}}$, $i = 1, \ldots,
l$. In both cases we produce the terms of order at most $n$, but the terms
with $l < n - 1$ will be accompanied by higher power of
$\ensuremath{\operatorname{time}}s$ and can be compensated for sufficiently
small time by terms in $\mathbf{Q}_s^{(n - 1)} $ by a choice of sufficiently
large $\varsigma_n$. When $l = n - 1$ and $i = 0$ the corresponding terms can
be compensated by terms coming from the derivative of $\mathbf{Q}_s^{(n - 1)}
$ for small times provided \ $\varsigma_n$ is sufficiently large. Otherwise
for $l = n - 1$ and $i \neq 0$, the corresponding terms can be dominated by
terms coming from the derivative of the free part (i.e. the part coming from
the commutator with $\mathbf{L}$) provided $\sup_{k_j, z_j, i,
y} \left\| Z_{k_j, z_j}
 q_{iy}  \right\| $ is
sufficiently small. 

As in the finite dimensional case of Theorem \ref{thm.1},  if we no longer assume that the $Z\ensuremath{_{\textrm{0,}y}}$ fields commute with
all the other $Z_{\alpha, x}$, then we will obtain extra terms. Such terms can be controlled, like in the proof  of Theorem \ref{thm.1}, thanks to our assumptions on the commutators. We do not repeat the whole calculation here, as it is completely analagous to the one done in finite dimensions. 
\endproof

%
%

\begin{thebibliography}{99} \label{Bibliography}

\bibitem{AKR} S. Albeverio, Y.G. Kondratiev  and M. R\"ockner,  Symmetrizing measures for infinite dimensional diffusions: an analytic approach, Geometric analysis and nonlinear partial differential equations (Stefan Hildebrandt et al., eds.), Springer, Berlin, 2003, 475-486. 

\bibitem{BE} D. Bakry and M.Emery, Diffusions hypercontractives. In: S\'em. de Probab. XIX. Lecture Notes in Math., vol. 1123, 177-206. Springer, Berlin 1985.

\bibitem{BBBC}
D. Bakry, F. Baudoin, M. Bonnefont and  D. Chafai,
On gradient bounds for the heat kernel on the Heisenberg group, J. Funct. Anal. 255 (2008), 1905-1938.

\bibitem{BHT} F. Baudoin,   M. Hairer and  J.  Teichmann,  Ornstein-Uhlenbeck processes on Lie groups,  J. Funct. Anal. 255 (2008),  877--890. 

\bibitem{BT} F. Baudoin and J. Teichmann, Hypoellipticity in infinite dimensions and an application in interest rate theory, Ann. Appl. Probab. 15 (2005), 1765-1777. 


\bibitem{BLU}
A.~Bonfiglioli, E.~Lanconelli, and F.~Uguzzoni, \emph{Stratified {L}ie groups
  and potential theory for their sub-{L}aplacians}, Springer Monographs in
  Mathematics, Springer, Berlin, 2007.
%
%
\bibitem{Crisan}
 D.~Crisan, K.~Manolarakis, C.~Nee. {\em Cubature methods and applications}.  Paris-Princeton Lectures on Mathematical Finance, Springer Verlag (2013).

\bibitem{DaPZ} G. Da Prato and J. Zabczyk,  Stochastic Equations in infinite dimensions.
Cambridge University Press (1992). 

%
\bibitem{DKZ2011}
F. Dragoni, V. Kontis and  B. Zegarli{\'n}ski, { Ergodicity of Markov semigroups with H\"ormander type generators in infinite dimensions}. Potential Anal. 37  (2011) 199-227.
%
%
 %
\bibitem{GZ}
A. Guionnet and B. Zegarli{\'n}ski, { Lecture notes on Logarithmic Sobolev Inequalities}. S\'eminaire de Probabilit\'es, XXXVI, vol. 1801, pp. 1-134. Lecture Notes in Math. Springer (2003).
%

%
\bibitem{HZ} W. Hebisch and B. Zegarli{\'n}ski, Coercive inequalities on metric measure spaces. J. Funct. Anal. 258 (2010) 814-851.
%

\bibitem{Herau2007}
F.~H{\'e}rau.
\newblock Short and long time behavior of the {F}okker-{P}lanck equation in a  confining potential and applications,  J. Funct. Anal. 244(1) (2007) 95-118.

%
%
\bibitem{H1} L. H\"ormander,  Hypoelliptic second order differential equations. Acta Math. 119 (1967) 147-171. 
%
\bibitem{I} J.  Inglis,  Coercive inequalities for generators of H\"ormander type. PhD Thesis, IC
2010. 
%
\bibitem{IKZ}
J. Inglis, V. Kontis and B. Zegarli{\'n}ski, From U-bounds to isoperimetry with applications to H-type groups,  J. Funct. Anal. 260 (2011)  76-116. 
%
\bibitem{IP} J. Inglis  and I. Papageorgiou,  Logarithmic Sobolev inequalities for infinite dimensional H\"ormander
type generators on the Heisenberg group.  Potential Anal. 31  (2009) 79-102.
%
%
\bibitem{MVII}
V. Kontis, M.~Ottobre and B. Zegarli{\'n}ski,  Long- and short-time behaviour of some hypocoercive-type operators in infinite dimensions: an analytic approach. Preprint.
%

\bibitem{L} T.M. Liggett,   Interacting Particle Systems. Springer 1985,
Stochastic Interacting Systems: Contact, Voter and Exclusion Processes. Springer 1999. 
%
\bibitem{LZ}
{  P.~{\L}ugiewicz and B.~Zegarli{\'n}ski}, { Coercive inequalities for
  {H}\"ormander type generators in infinite dimensions}, J. Funct. Anal. 247
  (2007) 438-476.
%
\bibitem{LuZ}  Xu Lihu and B.~Zegarli{\'n}ski, 
    Existence and Exponential mixing of infinite white $\alpha$-stable Systems with unbounded interactions,   Electronic J. Probab.   15 (2010)  1994-2018;
         Ergodicity of finite and infinite dimensional $\alpha$-stable systems 
    Stoch. Anal. Appl. 27 (2009) 797-824.    
%
\bibitem{MO_thesis}
M.~Ottobre. \textit{Asymptotic analysis for Markovian models in non-equilibrium Statistical Mechanics}, PhD Thesis, Imperial College London, 2012. 
\bibitem{R} M. R\"ockner, $L_p$-analysis of finite and infinite dimensional diffusion operators, Stochastic PDE's and Kolmogorov's equations in infinite dimensions (Giuseppe Da Prato, ed.), Lect. Notes Math., vol. 1715, Springer, Berlin, 1999, pp. 65-116;
An analytic approach to Kolmogorov's equations in infinite dimensions and probabilistic consequences, XIVth International Congress on Mathematical Physics 2003, World Scientific  2005, pp. 520-526. 
%
\bibitem{CSCV} N. T. Varopoulos, L. Saloff-Coste and T. Couhlon. Analysis and geometry on groups. Cambridge University Press, Cambridge, 1992.
%
\bibitem{V}
C.~Villani,
\newblock Hypocoercivity.
\newblock { Mem. Amer. Math. Soc.}, 202 (950) 2009.
%
%
\bibitem{Z1} B.~Zegarli{\'n}ski, The strong decay to equilibrium for the stochastic dynamics of unbounded spin systems on a lattice, Comm. Math. Phys. 175(2) (1996) 401-432. 

\end{thebibliography}
\end{document}